\documentclass[12pt]{article}
\usepackage[amsmath]{e-jc_bis}
\input{mymacros.sty}
\usepackage{graphicx}
\usepackage{enumerate}
\usepackage{stackrel}
\usepackage{comment}
\usepackage{dsfont}
\usepackage{tikz}
\usetikzlibrary{graphs, shapes.geometric}
\usetikzlibrary{arrows.meta, positioning}
\usepackage{subcaption}
\usepackage{mathtools}
\usepackage{nicematrix}
\usepackage{tikz-cd} 
\usetikzlibrary{cd}
\usepackage{caption}
\usepackage{theoremref} 
\usepackage{tabularx} 
\usepackage{tabu}
\usepackage{multirow}
\usepackage{array}
\usepackage{longtable}

\makeatletter
\newcommand{\thickhline}{%
    \noalign {\ifnum 0=`}\fi \hrule height 1pt
    \futurelet \reserved@a \@xhline
}
\newcolumntype{"}{@{\hskip\tabcolsep\vrule width 1pt\hskip\tabcolsep}}
\makeatother

\title{Linear Extensions of Rotor-Routing in Directed Graphs: Reachability Problems}

\author{David Auger
\and
Pierre Coucheney
\and
Kossi Roland Etse
}

\authortext{}{Laboratoire DAVID, Université de Versailles-Saint-Quentin-En-Yvelines, France,
   (\email{david.auger@uvsq.fr}, \email{pierre.coucheney@uvsq.fr}, \email{kossi-roland.etse@uvsq.fr})}

\begin{document}

\maketitle

\begin{abstract}
We develop a unified framework for rotor-routing that extends the classical model to a broad class of multigraphs equipped with Generalized Rotor Mechanisms (GRM). This perspective places rotor-routing on the same footing as abelian sandpiles by interpreting both as conservative instances of Vector Addition Systems (VAS). Within this framework, routing becomes a linear transformation governed by arc mechanisms, while legality is enforced through non-negativity constraints.

We introduce four routing models —free routing, standard rotor-routing, cyclic GRM routing, and fully general GRM routing— and study their reachability problems in both the linear and legal settings. Our results generalize previous characterizations for standard rotor-routing and extend them to the GRM setting. In particular, we show that legal reachability in GRM multigraphs is NP-complete, whereas the cyclic GRM routing model, which includes the classical rotor-router, admits a polynomial-time algorithm.
\end{abstract}

\section{Introduction and related works}

\subsection{The rotor-routing model}

The \emph{rotor-routing model}, also known as \emph{rotor-router} or \emph{rotor-walk} (see \cite{holroyd_chip-firing_2008} for a comprehensive overview), was first introduced by Priezzhev in 1996 as Eulerian walkers \cite{priezzhev_eulerian_1996, povolotsky_dynamics_1998}. Independently, in 1996, Propp and Wilson proposed it for generating random spanning trees in graphs \cite{wilson1996get}, and it relates to Yanovski et al.'s patrolling algorithms \cite{yanovski_distributed_2003}. This model is closely associated with \emph{abelian sandpiles} (or \emph{chip-firing}) model \cite{bjorner_chip-firing_1991, holroyd_chip-firing_2008}. For general introductions to abelian sandpiles, see \cite{giacaglia_local--global_2011} and \cite{holroyd_chip-firing_2008}, and Section \ref{sub:vas} hereafter. 

In the basic rotor-routing model, a single particle (or pebble, chip) moves deterministically on the vertices of a graph. When the particle is at a vertex $v$, it follows a predetermined sequence of arcs: the first time it visits $v$, it takes arc 1; the next time, arc 2; and so on, cycling back to arc 1 after all arcs have been used. This deterministic movement is applied at every vertex. On average, each arc (or transition) is crossed with the same frequency as in a random walk.

This simple rule defines the rotor-routing model, which exhibits many intriguing properties. For example, in an eulerian and strongly connected graph, the particle will eventually circulate repeatedly on an eulerian tour \cite{priezzhev_eulerian_1996}. In sufficiently connected graphs,
the particle will ultimately reach target vertices, called sinks, although the exploration time can be exponential in the number of vertices. The problem of determining the first sink reached, given a starting arc configuration and an initial vertex, is known as the {\sc Arrival} problem. Defined in \cite{dohrau_arrival_2017}, it was shown to belong to the complexity class NP~$\cap$~co-NP. While it is not known to be in P, it was proved to lie in the smaller class UP~$\cap$~co-UP \cite{gartner_arrival_2018}. Additionally, a subexponential algorithm based on computing a Tarski fixed point was proposed in \cite{gartner_subexponential_2021}.

\subsection{Abelian sandpiles and rotor-routing as special Vector Addition Systems }
\label{sub:vas}

In this paper, we will show how, similarly to the abelian sandpiles model, we may envision the model of rotor-routing as special case of Vector Addition Systems (VAS) \cite{karp1969parallel}, which are for reachability issues equivalent to Petri Nets.

A VAS is defined by a finite set of {\it transitions} $T \subset \Z^d$  where $d \geq 1$. The {\it states} of the system are elements of $\N^d$. A transition $t \in T$ from a state $v \in \N^d$ is {\it legal}, if $v + t \geq 0$, i.e. $v + t$ is a state. This defines an elementary {\it legal transition} from $v$ to $v+t$, and the reachability problem consists of deciding the existence of a finite sequence of legal transitions $(t_i)$ from some state $v_0$ to another state $v_1$, i.e. $v_1 = v_0 + t_1 + t_2 + \cdots + t_k$ and every intermediate step $v_0 + t_1 + t_2 + \cdots + t_i$ for $i \leq k$ is nonnegative. 

As an elementary example of reachability in VAS, consider $d=2$ and 
\[T=\{(1,-1),(-1,2)\}.\] Then state $(1,1)$ can legally reach state $(2,1)$ by the sequence of transitions 
\[((-1,2),(1,-1),(1,-1)),\] but $(0,0)$ cannot reach legally $(1,0)$, since applying any transition in $(0,0)$ would violate the nonnegativity constraint.

Reachability in VAS thus consists of determining whether an algebraic, linear relation
(i.e. $v_1-v_0$ belongs to the subgroup of $\Z^n$ generated by $T$) can be decomposed as a sequence of legal transitions -- which is what we call a legal sequence in this work. The problem of deciding reachability has been an important question in the field of VAS, and it is known that reachability is decidable \cite{mayr1981algorithm}.

Now, let us turn our attention to abelian sandpiles, a model which is intrinsically linked to rotor-routing and is a special case of VAS. Consider a finite directed graph $G=(V,A)$. A state is an element of $\N^V$, and is called a particle (or chip) configuration. It is interpreted as the quantity of particles lying on every vertex of $G$. Transitions, here called {\it firings}, are defined for every vertex $v \in V$, and consist of adding one particle to every outneighbour of $v$, and removing those particles from $v$. This transition is called a {\it firing at v}, and such a firing is {\it legal} if the resulting configuration is nonnegative, which amounts to saying that before firing, $v$ must have more particles than its number of outneighbours. See Fig.~\ref{fig:firing} for an example of firing in an abelian sandpile graph. Note that this model is a special case of VAS, with conservative transitions (i.e. the total number of particles remain constant). It is known that
in the general case, legal reachability for abelian sandpiles is not likely to be in complexity class P \cite{tothmeresz_rotor-routing_2022}.

The definition of legal firings in abelian sandpiles has been extended by several authors \cite{ farrell_coeulerian_2016, tothmeresz_rotor-routing_2022} to particle configurations in $\Z^V$, i.e. not necessarily nonnegative. In this more general context, we require for a legal firing that  before firing at $v$, the configuration has enough particles on $v$ to remain nonnegative at $v$ after firing (i.e. more particles than its outdegree). During such a process, a vertex containing a negative number of particles, can only receive particles and not emit them. As was the case for VAS, the legal reachability issues in abelian sandpiles (even in this generalized context) involve determining whether an algebraic relation between configurations of particles can be decomposed into a legal sequence of firings, which amounts to checking non-negativity conditions. 

In this paper, we show that standard rotor-routing in graphs,
even in its generalized form with negative 
particles \cite{giacaglia_local--global_2011, tothmeresz_algorithmic_2018}, 
can be viewed as a special case of a more general model, 
namely Generalized Rotor Mechanisms (GRM) multigraphs,
where rotor-routing is a special case of conservative VAS, similar to abelian sandpiles. Legal reachability in this model also involves determining whether a linear relation between rotor and particle configurations can be decomposed into a series of legal transitions (routing steps). As in all previous cases, legality is determined by verifying non-negativity conditions.

\begin{figure}[!htbp]
    \centering
    \begin{subfigure}{0.45\textwidth}
        \centering
        \begin{tikzpicture}
            \node[draw] (a) at (0, 0) {5};
            \node[draw] (b) at (-2, 0) {0};
            \node[draw] (c) at (2, 0) {4};
            \node[draw] (d) at (0, 2) {1};

            \draw[->] (a) -- (b);
            \draw[->] (a) -- (c);
            \draw[->] (a) -- (d);

        \end{tikzpicture}
        \caption{Initial configuration.\newline}
    \end{subfigure}
    \hspace{0.05\textwidth}
    \begin{subfigure}{0.45\textwidth}
        \centering
        \begin{tikzpicture}
            \node[draw] (a) at (0, 0) {2};
            \node[draw] (b) at (-2, 0) {1};
            \node[draw] (c) at (2, 0) {5};
            \node[draw] (d) at (0, 2) {2};

            \draw[->] (a) -- (b);
            \draw[->] (a) -- (c);
            \draw[->] (a) -- (d);
        \end{tikzpicture}
        \caption{Configuration obtained after the central vertex  has fired.}
    \end{subfigure}
    \caption{Example of a transition in an abelian sandpile graph. The values  represent the number of particles at each vertex.} 
    \label{fig:firing}
\end{figure}
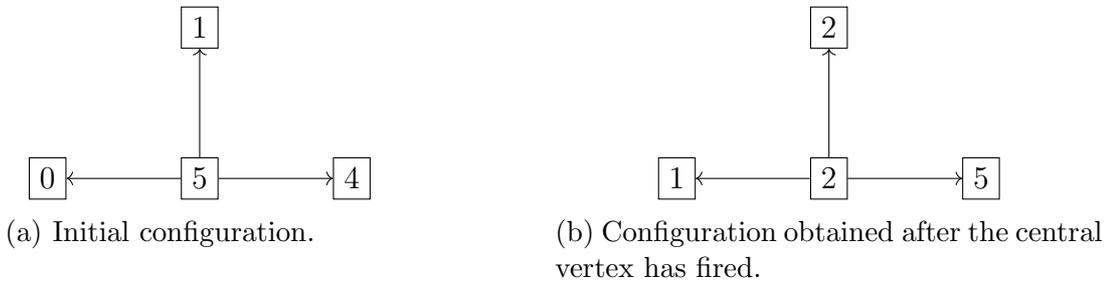

\subsection{Different notions of routing}

In this paper, we consider the movement of particles in a directed graph with vertex set $V$ and arc set $A$ (specifically, a multigraph). The term 'particle' can be replaced by pebble, chip, or counter, and has no physical significance. The positions of multiple particles are represented by a map $\sigma : V \rightarrow \N$, which counts the number of particles at each vertex. This is called a particle configuration. Notably, particles are indistinguishable and characterized solely by their positions.

In this paper, the term {\it routing} refers to the process of moving particles along the arcs of a graph according to specific rules that vary depending on the context. Mathematically, routing is defined as a rule that transforms one particle configuration  into another one. Elementary routing operations involve moving a single particle along an arc: the particle count at the head of the arc is incremented by 1, while the count at the tail is decremented by 1. If these rules can be applied to all particle configurations without additional constraints, the routing is termed {\it linear}. Conversely, if certain conditions (typically non-negativity constraints) must be satisfied, the routing is referred to as {\it legal}. 

In this work, we outline four distinct notions of routing, each with both a linear and a legal variant:

\begin{itemize}
    \item {\bf Standard rotor-routing in rotor multigraphs}: This refers to the classical concept of rotor-routing with standard cyclic rotor mechanisms. In this context, a configuration consists of both the particle configuration and the rotor configuration, where each non-sink vertex is associated with an outgoing arc. 

    \item {\bf Free routing in multigraphs}: 
    This is the most basic form of routing, where particles move through a multigraph without the involvement of rotor mechanisms or rotor configurations. While free routing is not rotor-routing per se, it serves as a foundational tool for  rotor-routing in GRM multigraphs.

    \item {\bf Rotor-routing in cyclic GRM multigraphs}: This type of routing occurs in a
    subclass of GRM multigraphs, where the arc mechanism operates cyclically, which simulates the case of standard rotor-routing. However, it diverges from standard rotor-routing by allowing 
     any formal sums of arcs instead of rotor configurations.

    \item {\bf Rotor-routing in GRM multigraphs}: In this more
    general context, mechanisms of rotors
    are no more limited to a cyclic behavior.
\end{itemize}

\subsection{Reachability problems}

The reachability problem we address in this paper is whether, given an initial and a final configuration, there exists a routing from the initial to the final configuration. This problem is examined across all the routing types we previously mentioned, for both the linear and legal cases, extending the result given in~\cite{tothmeresz_rotor-routing_2022} for standard rotor-routing. We also solve these problems in two scenarios: when the routing vector (i.e., the number of elementary operations at each vertex or arc) is specified, and when it is not. The results are summarized in the following table.

\begin{table}[h!]
  \centering
  \begin{tabular}{c c"c c}

    & & Routing vector not specified     & Routing vector specified\\
  \thickhline

  \multirow{2}{*}{Standard rotor} & Linear & \cite{tothmeresz_rotor-routing_2022}, Proposition~3.3  & $(*)$\\
  &  Legal & \cite{tothmeresz_rotor-routing_2022}, Theorem~3.4 & \thref{thm:charac_legal_rotor}   \\
  \hline

  \multirow{2}{*}{Free routing} & Linear & \thref{prop:image_bord}  & $(*)$ \\
   &  Legal & \thref{prop:flow_problem} &  \thref{prop:CNS_routing_vector}  \\
  \hline

  \multirow{2}{*}{Cyclic GRM} & Linear&\thref{lem:routing_vector}  &  $(*)$ \\
  &  Legal & \thref{thm:lr_cy_grm} &  \thref{thm:charac_legal_rotor}   \\
  \hline 
 
  \multirow{2}{*}{GRM} & Linear& \thref{lem:routing_vector} & $(*)$ \\
   &  Legal & \thref{thm:np_hard1} &  \thref{thm:routing_rotor_mecha} \\

  \end{tabular}
  \caption{This table presents the results concerning routing reachability problems and/or their computational complexity in the different cases. All of these problems are in P, except for legal rotor-routing in GRM multigraphs, which is NP-complete. Cases marked with a star $(*)$ correspond to computing the result of a linear operator using the specified routing vector, which can be efficiently done via matrix multiplication.
  }
  \label{tab:Summarize_table}
\end{table}

\subsection{Organization of the paper}

The paper is organized as follows. In Section~2, we introduce the standard rotor-routing framework and present significant results from prior research. Here, multiple particles are routed legally within a graph following standard rotor rules. This section is crucial as it establishes the notation used throughout the paper. While most of the standard rotor-routing results are not directly applied, they are generalized within the context of GRM multigraphs in Section~4.

In Section~3, we explore free routing and introduce the boundary operator, a key tool for  Section~4. The focus of this section is on characterizing legal free routings 
within the framework of free routing.

Section~4 formally defines rotor-routing in multigraphs with generalized rotor mechanisms (GRM), extending the standard 'cyclic' mechanism. Two generalizations are presented: the method for updating arcs in each rotor and the notion
of arc configurations that replace rotor configurations.

In Section~5, we characterize legal routings in GRM multigraphs when the routing vector is given as input. We then provide a simplified characterization for the cyclic case, which includes the standard rotor-routing model.

Finally, in Section~6, we address the reachability problem in cases where the routing vector is not specified a priori. We demonstrate that this problem is NP-complete in general, and present a polynomial-time algorithm for the cyclic case.

In a forthcoming paper, we will explore the algebraic properties of the GRM model, demonstrating how it enables a symmetric treatment of the rotor group and sandpile group, while also offering new insights into the {\sc Arrival} problem.

\newpage
\section*{Index of notations}

\begin{longtable}{lll}
\caption{Notation summary} \\
\hline
Symbol & Meaning & Reference \\
\hline
\endfirsthead

\hline
Symbol & Meaning & Reference \\
\hline
\endhead

\hline
\endfoot

\hline
\endlastfoot

$G$     & multigraph & \ref{sec:graphs} \\
$V$     & vertex set & \ref{sec:graphs} \\
$A$     & arc set & \ref{sec:graphs} \\
$S$     & set of sink vertices & \ref{sec:rotor_struct} \\

$F(v)$  & set/sum of faces in $A^+(v)$ & \ref{sec:linear_rotor_routing}, \ref{sec:toto} \\
$F$     & set of all faces & \ref{sec:linear_rotor_routing} \\

$\mathcal{W}$ & connected components of $G$ & \ref{sec:boundary_operator} \\

$A(V_1,V_2)$  & arcs with tail in $V_1$ and head in $V_2$ & \ref{sec:graphs} \\
$A^+(v)$ & $A(v,V)$ & \ref{sec:graphs} \\
$A^-(v)$ & $A(V,v)$ & \ref{sec:graphs} \\

$G^A(v)$ & multigraph $(A^+(v),F(v),\head^A,\tail^A)$ & \ref{sec:linear_rotor_routing} \\
$G^A$   & union of all $G^A(v)$ & \ref{sec:linear_rotor_routing} \\

$\head$, $\tail$ & maps $A \rightarrow V$ or $C_A \rightarrow C_V$ & \ref{sec:graphs},  \ref{sec:boundary_operator} \\
$\head^V, \tail^V, G^V, \deg^V$ & see $\head, \tail, G, \deg$ & \ref{sec:linear_rotor_routing} \\
$\head^A, \tail^A$ & maps $C_F \rightarrow C_A$ & \ref{sec:linear_rotor_routing} \\

$\deg$ & degree map $C_V \rightarrow C_{\mathcal{W}}$ & \ref{sec:boundary_operator} \\
$\deg^A$ & degree map for $G^A$ & \ref{sec:linear_rotor_routing} \\

$C_X$ & free group on $X$ & \ref{sec:free_abelian} \\
$C_X^+$ & nonnegative elements of $C_X$ & \ref{sec:free_abelian} \\

$B_V$ & boundaries (kernel of $\deg$) & \ref{sec:boundary_operator} \\
$Z_A$ & cycle space (kernel of $\partial$) & \ref{sec:boundary_operator} \\

$\Delta, L$ & Laplacian operator and matrix & \ref{sec:laplacian_matrix} \\

$\theta$ & circular ordering: $A^+(v) \to A^+(v)$ & \ref{sec:rotor_struct} \\

$\partial$ & boundary operator $C_A \rightarrow C_V$ & \ref{sec:boundary_operator} \\
$\partial^V$ & see $\partial$ & \ref{sec:boundary_operator} \\
$\partial^A$ & boundary operator $C_F \rightarrow C_A$ & \ref{sec:rotor_as_crs} \\

$\trans(\alpha,\sigma')$ & transitory vertices for $\sigma \lineareq{\partial}{\alpha} \sigma'$ & \ref{sec:transitory} \\

$\mathcal{L}$ & GRM routing operator & \ref{sec:toto} \\

$\sigma \lineareq{\partial}{\alpha} \sigma'$ & $\sigma' = \sigma + \partial(\alpha)$ & \ref{sec:legal_routing} \\
$\sigma \lineareq{\partial}{*} \sigma'$ & linear equivalence modulo $\partial$ & \ref{sec:legal_routing} \\
$\sigma \legalseq{\partial}{\alpha} \sigma'$ & legal sequence with routing vector $\alpha$ & \ref{sec:legal_routing} \\
$\sigma \legalseq{\partial}{*} \sigma'$ & existence of legal sequence & \ref{sec:legal_routing} \\

$(r,\sigma) \lineareq{\mathcal{L}}{\phi} (r',\sigma')$ & $(r',\sigma') = (r,\sigma) + \mathcal{L}(\phi)$ & \ref{sec:toto} \\
$(r,\sigma) \lineareq{\mathcal{L}}{*} (r',\sigma')$ & linear equivalence modulo $\mathcal{L}$ & \ref{sec:toto} \\
$(r,\sigma) \legalseq{\mathcal{L}}{\phi} (r',\sigma')$ & legal sequence with routing vector $\phi$ & \ref{sec:toto} \\
$(r,\sigma) \legalseq{\mathcal{L}}{*} (r',\sigma')$ & existence of legal sequence & \ref{sec:toto} \\

\end{longtable}

\begin{longtable}{lll}
\caption{Typical elements and notations} \\
\hline
Symbol & Meaning & Type / Space \\
\hline
\endfirsthead

\hline
Symbol & Meaning & Type / Space \\
\hline
\endhead

\hline
\endfoot

\hline
\endlastfoot

$\sigma$     & particle configuration & $C_V$ \\
$\sigma_v$   & coefficient of $\sigma$ & $\mathbb{Z}$ \\

$r$          & arc configuration & $C_A$ \\
$r_a$        & coefficient of $r$ & $\mathbb{Z}$ \\

$\alpha$     & routing vector modulo $\partial$ & $C_A$ \\
$\phi$       & routing vector modulo $\partial^A$ or $\mathcal{L}$ & $C_F$ \\

$\rho$       & rotor configuration & $C_A^+$ \\

$\mathcal{F}$  & directed / guiding forest & $\subset A$ \\

\end{longtable}

\section{Standard rotor-routing context and background}

In this section, we recall the framework of directed graphs and rotor-routing, together with some known results that we use or generalize in the rest of the paper. We call the model of rotor-routing that is presented here the {\it standard rotor-routing} context,
which is the model that can be found in most articles on the subject. For simplicity, we leave out the case of linear rotor-routing \cite{tothmeresz_algorithmic_2018} and focus here solely on the case where positive particles are routed, according to a so-called \emph{rotor configuration}.

\subsection{Graphs}
\label{sec:graphs}

\paragraph{Multigraphs.} 
A \textbf{directed multigraph} $G$ is a tuple $G=(V,A,\head,\tail)$
where $V$ and $A$ are respectively finite sets of \emph{vertices} and \emph{arcs}, and
{\it head} and {\it tail} are maps from $A$ to $V$ defining the
incidence between arcs and vertices. An arc with tail $x$ and head $y$ is said to be from $x$ to $y$.
Note that multigraphs can have multiple arcs with the same head and tail,  as well as loops. 

For two sets $V_1,V_2 \subset V$, define $A(V_1,V_2)$ as the set of arcs with tail in $V_1$ and head in $V_2$. For a vertex $u\in V$, we denote by $A^{+}(u)$ the subset of arcs going out of $u$, i.e. $A^{+}(u)=A({u},V)$, as well as $A^{-}(u)=A(V,{u})$. 

\paragraph{Paths and connectedness.}
If $x,y \in V$, a {\bf directed path from $x$ to $y$} is a finite sequence of 
arcs $a_1,a_2\cdots a_k$  such that $\head(a_i)=\tail(a_{i+1})$ for $1 \leq i \leq k-1$, and also $\tail(a_1)=x$ and $\head(a_k)=y$ (note that such a path is usually defined as a sequence of vertices, but both definitions are equivalent). This definition includes the empty sequence from $x$ to $x$. The graph is said to be {\bf strongly connected} if there is a directed path from any vertex to any other vertex. A {\bf directed cycle} is a nonempty directed path from a vertex $x$ to the same vertex $x$. This includes the case of a single arc (a loop) with the same tail and head.

The {\bf strongly connected components} of $G$
are the equivalence classes of the equivalence relation on vertices, where $x$ and $y$ are considered equivalent if there is a directed path from $x$ to $y$ and a directed path from $y$ to $x$ in $G$. A strongly connected component $C \subset V$ is a {\bf leaf component} if $A(C,V \setminus C) = \emptyset$. In particular, if $v$ is a {\bf sink vertex} of $G$, i.e. $A^+(v)=\emptyset$, then $\{v\}$ is a leaf component of $G$.

We define an {\bf undirected path from $x$ to $y$} as  a finite sequence of 
arcs $a_1,a_2\cdots a_k$  such that there is a sequence $v_0, v_1, \cdots, v_k$ of vertices with $v_0 =x$, $v_k=y$ and such that for all $1 \leq i \leq k$ we have
\begin{itemize}
    \item either $\tail(a_i) = v_{i-1}$ and $\head(a_i) = v_i$ ;
    \item or $\tail(a_i) = v_{i}$ and $\head(a_i) = v_{i-1}.$ 
\end{itemize}
In the first case we say that $a_i$ is  forward oriented in the path, and in the second case that it is backward oriented. Note that an undirected path corresponds to a path with the usual definition in the undirected graph where we replace every arc of $G$ by an undirected edge.

A {\bf weakly connected component} of $G$ is a maximal subset of vertices $V_1 \subset V$ such that there is an undirected path from any vertex of $V_1$ to any other vertex of $V_1$. The graph is said to be {\bf weakly connected} if there is only one connected component which is $V$ itself.

\paragraph{Directed Forests.}
If $X \subset V$, a {\bf directed forest} with {\bf domain} $X$ is a set of arcs ${\cal F} \subset A(X,V)$ such that:
\begin{enumerate}[(i)]
    \item for every $v \in X$, there is exactly one arc $a \in {\cal F}$ such that $\tail(a) = v$ ;
    \item ${\cal F}$ contains no directed cycles.
\end{enumerate}

Note that such a directed forest exists if and only if
for every vertex $v \in X$, there is a directed path from $v$ to some $s \notin X$.

\subsection{Free abelian groups}

\label{sec:free_abelian}

We refer to \cite{lang2012algebra} for standard notions of abelian groups. We use additive notation. An abelian group $(H,+)$ is {\bf free} if it admits a {\bf basis}, i.e. a family $(h_i)_{i \in I}$ of elements such that each $h \in H$ can be written uniquely as a sum $h = \sum_{i \in I} c_i h_i$ with integers $c_i \in \Z$, where all but finitely many of the $c_i$ are zero. If there is a finite basis, then all basis have the same cardinal which is called the {\bf rank} of $H$.

The {\bf universal property of free abelian groups} says that if $f : X \rightarrow R$ is a map from the elements of a basis $X$ of $H$ to an abelian group $R$, then $f$ extends uniquely to a group homomorphism from $H$ to $R$. Indeed, if $(h_i)_{i \in I}$ is
a basis of $H$, and $f(h_i)=r_i$, we define an homomorphism by
\[ f\left( \sum_{i \in I} c_i h_i \right) = \sum_{i \in I} c_i r_i,\]
where all but finitely many $c_i$ are zero.

If $X$ is a finite set, the {\bf free abelian group} on $X$ is the set of formal sums with integer coefficients
$$c = \sum_{x \in X} c_x \cdot x$$
where $c_x \in \Z$ for every $x \in X$, with pointwise sum of coefficients. It can also be viewed as $\Z^X$, the set of maps from $X$ to $\Z$ together with standard pointwise sum. We denote this group as $C_X$.
Elements of $X$ are identified with particular elements of $C_X$, and $X$ can be viewed as a basis of $C_X$, which is called the canonical basis. If $x \in X$, we write $c_x$ for the coordinates in the canonical basis of some $c \in C_X$. 

We say that $c \in C_X$ is nonnegative, and write $c \geq 0$, if $c_x \geq 0$ for all $x \in X$. Likewise, we say that $c_1 \leq c_2$ if $c_2 - c_1 \geq 0$.
Let $C^+_X$ the set of $c \in C_X$ such that $c \geq 0$. If $c \in C^+_X$ and $c_x >0$, we say that {\bf $x$ is an element of $c$} and write $x \in c$. Indeed, $c \in C^+_X$ can be identified with a multiset on $X$.

If $G=(V,A,\head,\tail)$ is a directed multigraph, we denote by  $C_V$ the free group on $V$; its elements are called {\bf particle configurations}. 
Likewise, $C_A$ denotes the free group on $A$ and its elements are {\bf arc configurations}. In the whole paper, instead of considering vectors or sets like $\Z^V$ and $\Z^A$, we consider formal sums $C_V$ and $C_A$. This allows more concise notation: for instance, the element of $C_V$ with coefficient $3$ on $v_1$ and $-5$ on $v_2$ and $0$ elsewhere will be simply denoted by $3v_1-5v_2$.

\subsection{Laplacian homomorphism and matrix}
\label{sec:laplacian_matrix}

Let $G=(V,A,\head,\tail)$ be a directed multigraph. The Laplacian homomorphism is the homomorphism $\Delta$ from $C_V$ to itself whose value on every vertex $v \in V$ (viewed as an element of $C_V$) is

   $$\Delta(v) =  -\sum_{a \in A^+(v)} (\head(a) - \tail(a))$$

The {\bf Laplacian matrix} $L$ is the matrix of $\Delta$ in the canonical basis of $C_V$. Its rows and columns are indexed by $V$, and its entries are given by
\[
L_{v',v} =
\begin{cases}
|A(v,V\setminus v)| & \text{if } v' = v, \\[1ex]
-|A(v,v')| & \text{if } v' \neq v.
\end{cases}
\]
 
Let $S \subset V$. The {\bf $S$-reduced Laplacian matrix} is the matrix $L'$ obtained by removing from $L$ the rows and columns corresponding to $S$. 
The following result is classic for undirected graphs and known as \emph{Kirchoff's matrix-tree theorem}, but the directed version is a bit less common (see~\cite{pitman_tree_2018} for a proof). 

\begin{theorem}
    Let $S \subset V$ and let $L'$ be the $S$-reduced Laplacian matrix. Then the determinant of $L'$ is the number of directed forests with domain $V\setminus S$.
\end{theorem}

Of particular interest to us is also the kernel (null space) of $\Delta$, which is given by the following result whose proof can be found in~\cite{bjorner_chip-firing_1992}.

\begin{theorem} \thlabel{thm:kerdelta}
Let $k \geq 0$ be the number of leaf components of $G$.
Denote by $S_1, S_2, \dots S_k$ these components.
Then there exist $\sigma^1, \sigma^2, \dots, \sigma^k \in C_V$ such that
\begin{enumerate}[(i)]
    \item $(\sigma^1, \sigma^2, \dots, \sigma^k)$
    is a basis of the kernel of $\Delta$;
    \item for all $1 \leq i \leq k$, $\sigma^i_v >0$
    for all $v \in S_i$ and  $\sigma^i_v =0$
    for all $v \in V \setminus S_i$.
\end{enumerate}
This basis $(\sigma^1, \sigma^2, \dots, \sigma^k)$
is unique and its elements are called {\bf primitive period vectors} of $G$.
\end{theorem}

In particular, if $G$ is strongly connected, it admits
a unique primitive period vector and the kernel of $\Delta$
has rank $1$. Moreover, if $G$ is {\bf eulerian}, i.e. $|A^+(v)|=|A^-(v)|$ for every $v \in V$, then the only primitive vector is $\sum_{v \in V} v$, i.e. its coordinates in the canonical basis are all $1$. This is the case for the graph $G_1$ of Fig. \ref{fig:exempleG1}.

\subsection{Rotor structure}
\label{sec:rotor_struct}

Let $G=(V,A,\head,\tail)$ be a directed multigraph.

\paragraph{Rotor orders and rotor graphs.} A {\bf circular ordering} on a finite set $X$ is a map $\theta : X \rightarrow X$ 
such that, for all $x \in X$, the sequence of iterates $x, \theta(x), \theta^2(x), \dots$ generates the whole set $X$. 

If $G$ is a multigraph, a {\bf rotor} at $v \in V$ is a circular ordering on $A^+(v)$. A {\bf rotor multigraph} $G=(V,A,\head,\tail,\theta)$ is such that:

\begin{itemize}
    \item  $(V, A, \head, \tail)$  is a   multigraph;
    \item for all vertices $v \in V$, the restriction of $\theta$ to $A^+(v)$ is a rotor at $v$.
\end{itemize}

If $v$ is a {\bf sink} vertex, i.e. $A^+(v)=\emptyset$, then
the second condition is trivial. When the rotor multigraph $G$ 
is fixed, we denote by $S$ the set of its sinks. Depending on the context, $S$ can be empty or not. A directed multigraph is \textbf{stopping} if for every vertex $u$, there is a directed path from $u$ to a sink.

A {\bf rotor configuration} in a rotor multigraph is a map $\rho$ that associates
to every $v \in V \setminus S$ an arc $\rho(v) \in A^+(v)$. 
We can identify rotor configurations with particular elements of $C_A$, namely $\sum_{v \in V \setminus S} \rho(v)$.



Two examples of rotor multigraphs are given on Fig. \ref{fig:exempleG1} and \ref{fig:exempleG2}. The first one is strongly connected and sinkless, while the second is stopping.

\begin{figure}[!htbp]
   \centering
\begin{tikzpicture}[->, >=stealth, node distance=3cm]

        \draw [ thick,->,>=stealth](25:5) arc (0:330:0.4cm);

        \node[draw] (v2) at (330:3) {$v_2$};
        \node[draw] (v0) at (90:3) {$v_0$};
        \node[draw] (v1) at (210:3) {$v_1$};

        \draw[bend right=40] (v2) to node[pos=.3, left=.1mm] {$a'_{2,0}$} (v0);
        \draw[bend right=60] (v2) to node[pos=.3, right=.1cm] {$a_{2,0}$} (v0);
        \draw[bend right=15] (v0) to node[pos=.2, right] {$a_{0,2}$} (v2);
        \draw[line width=1.5pt, bend right=60] (v0) to node[pos=.3, left] {$a_{0,1}$} (v1);
        \draw[bend right=40] (v0) to node[pos=.3, right] {$a'_{0,1}$} (v1);
        \draw[bend right=15] (v1) to node[pos=.3, left] {$a_{1,0}$} (v0);
        \draw[bend right=60] (v1) to node[pos=.3, below] {$a_{1,2}$} (v2);
        \draw[line width=1.5pt, bend right=40] (v1) to node[pos=.3, above] {$a'_{1,2}$} (v2);
        \draw[line width=1.5pt, bend right=15] (v2) to node[pos=.3, below] {$a_{2,1}$} (v1);

    \end{tikzpicture}

        \caption{A rotor multigraph $G_1$ with no sinks, which is strongly connected. Every vertex has out-degree $3$ and in-degree $3$. As an example, we have $\head(a_{2,0})=v_0$ and $\tail(a_{2,0})=v_2$. The rotor order at every vertex is given
        by anticlockwise ordering; e.g. $\theta(a_{2,0})=a'_{2,0}$,
        $\theta(a'_{2,0})=a_{2,1}$ and $\theta(a_{2,1})=a_{2,0}$.
        A rotor configuration $\rho_1$ with $\rho_1(v_0)=a_{0,1}$,
        $\rho_1(v_1)=a'_{1,2}$ and $\rho_1(v_2)=a_{2,1}$ is depicted in bold.}
        \label{fig:exempleG1}
\end{figure}
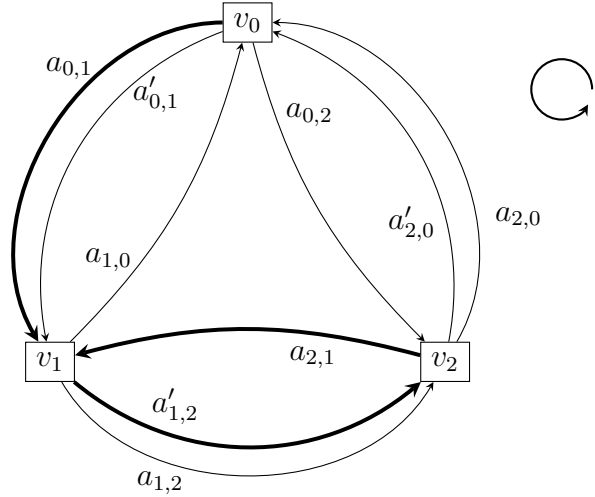

\begin{figure}[!htbp]
        \centering

\begin{tikzpicture}[->, >=stealth, node distance=2cm]

        \draw [ thick,->,>=stealth](25:4) arc (0:330:0.4cm);

        \node[draw] (v2) at (330:3) {$v_2$};
        \node[draw] (s1) at (310:6) {$s_1$};
        \node[draw] (v3) at (90:3) {$v_3$};
        \node[draw] (v4) at (210:3) {$v_4$};
        \node[draw] (s0) at (230:6) {$s_0$};

        \draw[bend right=15] (v2) -- (s0) node[midway, below] {$a_{2,0}$};
        \draw (v2) -- (s1) node[midway, below] {$a_{2,1}$};
        \draw (v4) -- (s0) node[midway, right] {$a_{4,0}$};
        \draw[bend right=15] (v4) -- (s1) node[midway, right] {$a_{4,1}$};
        \draw[bend right=15] (v2) to node[midway, right] {$a_{2,3}$} (v3);
        \draw[bend right=15] (v3) to node[midway, below] {$a_{3,2}$} (v2);
        \draw[bend right=15] (v3) to node[midway, left] {$a_{3,4}$} (v4);
        \draw[bend right=15] (v4) to node[midway, right] {$a_{4,3}$} (v3);
        \draw[bend right=15] (v4) to node[midway, below] {$a_{4,2}$} (v2);
        \draw[bend right=15] (v2) to node[midway, above] {$a_{2,4}$} (v4);
    
    \end{tikzpicture}

    \caption{A stopping rotor multigraph $G_2$, with two sinks $s_0$ and $s_1$.}
    \label{fig:exempleG2}
\end{figure}
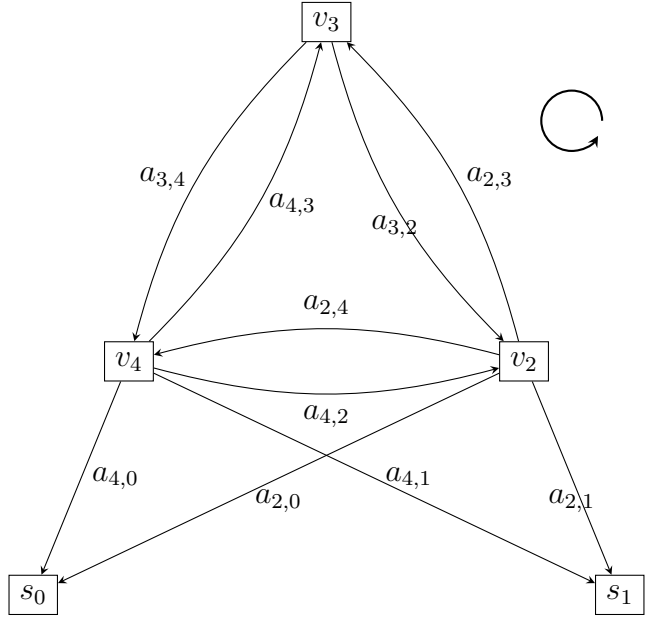

\subsection{Standard rotor-routing}
\label{sec:standard}

Let $G=(V,A,\head,\tail,\theta)$ be a rotor multigraph.
Classically, rotor-routing is concerned with so-called {\it chip configurations}, namely {nonnegative particle configurations} with our current terminology. If $\sigma \in C_V$ and $\sigma_v>0$, we interpret $\sigma_v$ as the number of particles on vertex $v$. 
Routing these particles consists of moving them along an arc of a rotor configuration. 

\paragraph{Rotor-routing operation and rotor walks.}

Consider a rotor configuration $\rho$ and a nonnegative particle configuration $\sigma \in C^+_V$. A rotor-routing operation at $v \in V \setminus S$ is valid, only if $\sigma_v>0$. In this case, the routing operation transforms $(\rho,\sigma)$ into $(\rho',\sigma')$ where:
\begin{itemize}
\item  $\rho'$ is equal to $\rho$ except on $v$ where $\rho'(v) = \theta(\rho(v))$;
    \item $\sigma'$ is equal to $\sigma$ except $\sigma'_v=\sigma_v-1$ and $\sigma'_{v_1} = \sigma_{v_1}+1$ where $v_1=\head(\rho(v))$.
\end{itemize}

We interpret this as a particle moving along the arc $\rho(v)$ from $v$ to $v_1$, then the rotor configuration at $v$ being updated to the next arc in the rotor ordering. Note that the resulting particle configuration $\sigma'$ is also nonnegative. A {\bf rotor walk} is a finite or infinite sequence of configurations $(\rho_0,\sigma_0), (\rho_1,\sigma_1), (\rho_2,\sigma_2), \dots$ such that each new couple of configurations is obtained from the previous one by a routing operation. This sequence \emph{starts} at $(\rho_0,\sigma_0)$, and if finite \emph{ends} at some $(\rho_k,\sigma_k)$. Such a rotor walk is \emph{maximal} if it is infinite or ends in a configuration $(\rho_k,\sigma_k)$ where no valid routing operation can be applied, i.e. $\sigma_k$ is $0$ on $v$ for all $v \in V \setminus S$. Fig.  \ref{fig:routage_simple} shows an example of routing.

If there is a single particle (i.e. $\sigma=v$ for some vertex $v \in V$), then there is a unique maximal rotor walk starting in $(\rho,\sigma)$. If there are several particles, there is a choice in the next particle to be routed. The first fundamental result is the following~\cite{holroyd_chip-firing_2008, tothmeresz_algorithmic_2018}:

\begin{theorem}
\thlabel{thm:commute_rotor_routing}
If $G$ is stopping, all maximal rotor walks are finite. Moreover, from a starting configuration $(\rho,\sigma)$ where $\rho$ is a rotor configuration and $\sigma \geq 0$, all maximal rotor walks end in the same configuration $(\rho',\sigma')$, and for every vertex $v \in V \setminus S$, the number of times that a routing operation is performed at $v$ does not depend on the choice of the maximal rotor walk.

If $G$ is strongly connected, all maximal  rotor walks are infinite. 
    Moreover, in the case of a single particle configuration $\sigma$,
    let $p \in C_V$ be the unique primitive period vector
    of $G$. Then all maximal rotor walks that start in $(\rho,\sigma)$ are ultimately periodic, and during a least period the number of times that a  routing is performed at a vertex $v$ is equal to $|A^+(v)| \cdot p(v)$.
    
\end{theorem}

A more general result describing the asymptotic structure
of maximal rotor walks, when the graph is neither stopping nor strongly connected,
can be found in~\cite{duhaze2023reachability} (Thm. 2.5.10).

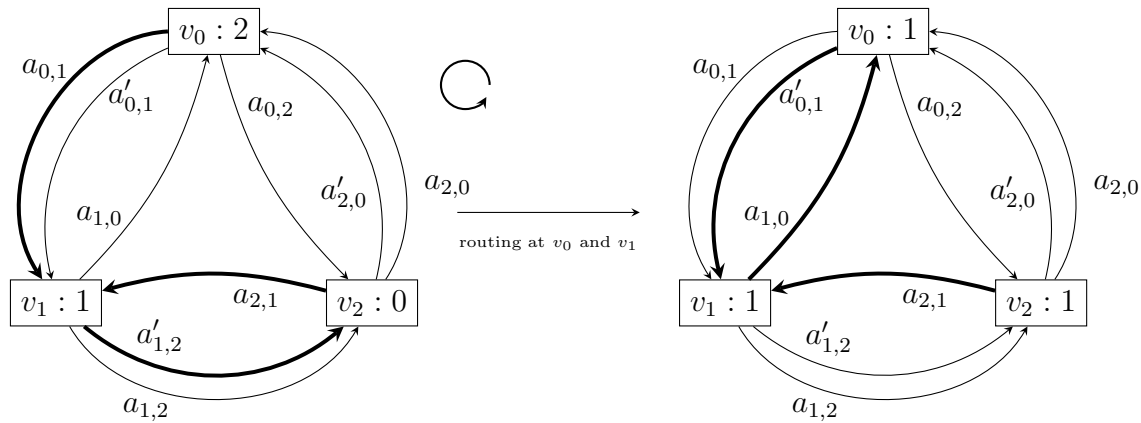
\begin{figure}[!htbp]
   \centering
\begin{tikzpicture}[->, >=stealth, scale=0.8]


        \draw [ thick,->,>=stealth](25:5) arc (0:330:0.4cm);

        \node[draw] (v2) at (330:3) {$v_2 : 0$};
        \node[draw] (v0) at (90:3) {$v_0 : 2$};
        \node[draw] (v1) at (210:3) {$v_1 : 1$};

        \draw[bend right=40] (v2) to node[pos=.3, left=.1mm] {$a'_{2,0}$} (v0);
        \draw[bend right=60] (v2) to node[pos=.3, right=.1cm] {$a_{2,0}$} (v0);
        \draw[bend right=15] (v0) to node[pos=.2, right] {$a_{0,2}$} (v2);
        \draw[line width=1.5pt, bend right=60] (v0) to node[pos=.3, left] {$a_{0,1}$} (v1);
        \draw[bend right=40] (v0) to node[pos=.3, right] {$a'_{0,1}$} (v1);
        \draw[bend right=15] (v1) to node[pos=.3, left] {$a_{1,0}$} (v0);
        \draw[bend right=60] (v1) to node[pos=.3, below] {$a_{1,2}$} (v2);
        \draw[line width=1.5pt, bend right=40] (v1) to node[pos=.3, above] {$a'_{1,2}$} (v2);
        \draw[line width=1.5pt, bend right=15] (v2) to node[pos=.3, below] {$a_{2,1}$} (v1);

    \draw[->] (4,0) -- (7,0);
    \node at (5.5,-.5) {\tiny routing at $v_0$ and $v_1$};

    \begin{scope}[shift={(11,0)}]

        \node[draw] (v2) at (330:3) {$v_2 : 1$};
        \node[draw] (v0) at (90:3) {$v_0 : 1$};
        \node[draw] (v1) at (210:3) {$v_1 : 1$};

        \draw[bend right=40] (v2) to node[pos=.3, left=.1mm] {$a'_{2,0}$} (v0);
        \draw[bend right=60] (v2) to node[pos=.3, right=.1cm] {$a_{2,0}$} (v0);
        \draw[bend right=15] (v0) to node[pos=.2, right] {$a_{0,2}$} (v2);
        \draw[bend right=60] (v0) to node[pos=.3, left] {$a_{0,1}$} (v1);
        \draw[line width=1.5pt, bend right=40] (v0) to node[pos=.3, right] {$a'_{0,1}$} (v1);
        \draw[line width=1.5pt, bend right=15] (v1) to node[pos=.3, left] {$a_{1,0}$} (v0);
        \draw[bend right=60] (v1) to node[pos=.3, below] {$a_{1,2}$} (v2);
        \draw[bend right=40] (v1) to node[pos=.3, above] {$a'_{1,2}$} (v2);
        \draw[line width=1.5pt, bend right=15] (v2) to node[pos=.3, below] {$a_{2,1}$} (v1);
        \end{scope}
    
    \end{tikzpicture}

        \caption{With the same graph $G_1$ as in Fig \ref{fig:exempleG1}: on the left, a rotor configuration $\rho_1$
        (arcs in bold) and a particle configuration $\sigma_0$ (numbers in vertices) are given.
        The rotor walk $(\rho_1,\sigma_1), (\rho_2,\sigma_2), (\rho_3,\sigma_3)$ is legal,
        and consists of routing a particle in $v_0$
        and a particle in $v_1$. The resulting configurations
        $(\rho_3,\sigma_3)$ are given on the right.
        }
        \label{fig:routage_simple}
\end{figure}

\paragraph{{\sc Arrival} Problem.}

Suppose that $G$ is a stopping rotor multigraph.
The {\sc Arrival} Problem, introduced in~\cite{dohrau_arrival_2017}, consists of determining the final configuration $\sigma'$ when applying a maximal rotor walk from $(\sigma, \rho)$, where $\sigma \in C^+_V$ and $\rho$ is a rotor configuration. The configuration $\sigma'$ is then uniquely determined by  \thref{thm:commute_rotor_routing}.
Equivalently, one can ask how many particles will settle on a given sink or use a decision version of the problem.

No polynomial-time algorithm is currently known for solving the {\sc Arrival} problem, even when  $\sigma$ consists of a single particle. In particular, routing such a particle may require an exponential number of steps. The most efficient algorithm known to date is subexponential, as described in~\cite{gartner_subexponential_2021}.

\paragraph{Flows.}
Let $G$ be a stopping multigraph, $\rho$ a rotor configuration and $\sigma \in C^+_V$,
such that all maximal rotor walks that start from $(\rho,\sigma)$ end in $(\rho',\sigma')$, with $\sigma'_v=0$ for all $v \notin S$, as stated in \thref{thm:commute_rotor_routing}.

Define the {\bf run} of $(\rho,\sigma)$ as the map $f : A \rightarrow \N$ where $f(a)$ is the number of times
that a particle travels along arc $a$ during such a maximal walk (this number is also independent of the choice of the maximal walk, by the same result). Then the run $f$ satisfies the following equations:

\begin{itemize}
    \item  for all $v \in V$, flow conservation at $v$:
    \begin{equation} \label{eq:flow_conservation}
        \sum_{a \in A^-(v)} f(a) + \sigma_v = \sum_{a \in A^+(v)} f(a) + \sigma'_v
    \end{equation}
    \item for all $v \in V \setminus S$, rotor condition at $v$:
\begin{equation} \label{eq:rotor_condition}
    f(\rho(v)) \geq f(\theta(\rho(v))) \geq \cdots f(\theta^i (\rho(v))) \geq f(\theta^{i+1}(\rho(v))) \geq \cdots \geq f(\rho(v)) - 1
\end{equation}
\end{itemize}

Conversely, we call a {\bf flow} for $(\rho,\sigma,\sigma')$ a map $f : A \rightarrow \N$ satisfying all equations (\ref{eq:flow_conservation}) and (\ref{eq:rotor_condition}); hence the run of $(\rho,\sigma)$ is a flow for $(\rho,\sigma,\sigma')$. 
Appendix \ref{sec:flow_example} gives detailed examples of run and flows in the graph $G_2$ of Fig. \ref{fig:exempleG2}.
It turns out that these equations are sufficient to characterize $\sigma'$, by the following result, due to Dohrau et al \cite{dohrau_arrival_2017}.

\begin{theorem}
\thlabel{prop:certif_flow}
    Suppose that $G$ is stopping, and that all maximal rotor walks that start from $(\rho,\sigma)$ end in $(\rho',\sigma')$. Let $\sigma_1 \in C_V$ with $\sigma_1(v)=0$ for all $v \in V \setminus S$, then  $\sigma_1 = \sigma'$ if and only if there exists a flow for $(\rho,\sigma,\sigma_1)$.
\end{theorem}

This result proves that {\sc Arrival} belongs to the complexity class NP and to the class co-NP. Indeed, one can certify that the ending configuration is $\sigma'$ by giving a flow for $(\rho,\sigma,\sigma')$, and one can also certify that the ending configuration is not $\sigma'$ by giving a flow for $(\rho,\sigma,\sigma'')$ where $\sigma'' \neq \sigma'$.

We can note that if $f$ is the run for $(\rho,\sigma)$,
then it is easy to compute $\sigma'$ with equations (\ref{eq:flow_conservation}). It is also possible to compute $\rho'$ with equations (\ref{eq:rotor_condition}),
since for all $v \in V \setminus S$:
\begin{itemize}
    \item either $f(a)$ is constant on $A^+(v)$, which implies that $\rho'(v)=\rho(v)$;
    \item or there is $i>0$ such that $f(\theta^i(\rho(v))) < f(\rho(v))$, and then
    the minimal such $i$ gives $\rho'(v) = \theta^i(\rho(v))$.
\end{itemize}

If $f$ is a flow for $(\rho,\sigma,\sigma')$ but not the run for $(\rho,\sigma)$, we can also build
from $f$ 
a corresponding rotor configuration $\rho''$ by
equations (\ref{eq:flow_conservation}), with the same relations. However, this configuration will not be $\rho'$. If we want to give a certificate  for $\rho'$, we need to characterize the run among flows by additional conditions. The following result is due to Gärtner et al \cite{gartner_arrival_2018}
and implies in particular that {\sc Arrival} belongs to UP and co-UP.

\begin{theorem} \thlabel{thm:run_carac}
    Let $G$ be a stopping rotor multigraph. Let $f$ be a flow for $(\rho,\sigma,\sigma')$ and $\rho'$ the rotor configuration built from $f$ and $\rho$. Let $V_1 \subset V$
    the set of active vertices for $f$, i.e. $v \in V$ such that $\sum_{a \in A^+(v)} f(a) > 0$, and let $A_1 \subset A$ the set of arcs $\theta^{-1}(\rho'(v))$ for $v \in V_1$. Then $f$ is the run for $(\rho,\sigma)$
    if and only if $(V,A_1,\head,\tail)$ contains no directed cycles.
\end{theorem}

The idea behind this result is that $A_1$ is the set of traces of the run, i.e. the arcs corresponding to the last time a particle is routed at each vertex, and following these arcs must lead to a sink.

This result was generalized by T\'othmérész \cite{tothmeresz_rotor-routing_2022} to the case of graphs that are not necessarily stopping and allow the presence of negative particles. 
The main results of this paper, notably \thref{thm:routing_rotor_mecha}, \thref{thm:charac_legal_rotor}, 
\thref{thm:np_hard1} and \thref{thm:lr_cy_grm}, can be viewed as further generalizations of \thref{thm:run_carac}.

\paragraph{Recurrent configurations.}

Let $G$ be a strongly connected rotor multigraph. A configuration $(\rho, \sigma)$ is called \emph{recurrent} if there exists a non-empty rotor walk that returns from $(\rho, \sigma)$ to itself. Such configurations exist since, 
by \thref{thm:commute_rotor_routing}, any maximal rotor walk on $G$ is infinite. Characterization of recurrent configurations has been analyzed in~\cite{tothmeresz_algorithmic_2018}:

\begin{theorem}
    \thlabel{thm:recurrent}
    Let $G$ be a strongly connected multigraph. Let $A_1 \subset A$ be the set of arcs $\theta^{-1}(\rho(v))$ for $v \in V \setminus S$. Then    
    configuration $(\rho, \sigma)$ is recurrent if and only if for every directed cycle $C$ in $A_1$, there is an arc
    in $C$ whose head $v$ satisfies $\sigma_{v} > 0$. 
\end{theorem}


\paragraph{Turn and Move routing.}

We note that another definition of rotor-routing exists and is also widely used in the literature, which we refer to as \emph{Turn and Move} routing (in contrast with  the \emph{Move and Turn} routing that we use). In this version,  when a  routing step at $v$ for $(\rho,\sigma)$ is processed, the rotor $\rho(v)$ is first updated to $\theta(\rho(v))$ and then a particle moves along the latter arc. This definition is also quite standard, and both definitions have their merits and flaws in terms of description of properties of routing. Ultimately, they are equivalent, as they are conjugate by the turn operator $\theta$, as shown by the following commutative diagram, which provides a correspondence between two rotor walks, routed either with Move and Turn rotor (M\&T) routing as used in this paper, or Turn and Move (T\&M) as just described.

\begin{center}
\begin{tikzcd}[column sep=large]
(\rho_0, \sigma_0) \arrow[r, "T \&  M(v_0)"] \arrow[d,"\theta \times id"]  & (\rho_1, \sigma_1) \arrow[r, "T \&  M(v_1)"] \arrow[d,"\theta \times id"] & \dots \arrow[r, "T \&  M(v_{k-1})"]  & (\rho_k, \sigma_k) \arrow[d,"\theta \times id"]\\
(\rho_0^+, \sigma_0) \arrow[r, "M \&  T(v_0)"] & (\rho_1^+, \sigma_1) \arrow[r, "M \&  T(v_1)"] & \dots \arrow[r, "M \&  T(v_{k-1})"]  & (\rho_k^+, \sigma_k) 
\end{tikzcd}
\end{center}

\section{Free routing in multigraphs} \label{sec:boundary}

In this section, we study the free routing of particles in a graph without any constraints, notably without rotors. In this model, a particle can always move along any arc, hence the term 'free'. We first introduce the boundary operator and highlight its key properties. This operator is central to our analysis, enabling us to count how many particles traverse each arc, thereby extending the concept of run/flow from Section \ref{sec:standard}. The terminology for the boundary operator, boundaries, and cycles is derived from standard simplicial homology of graphs (see, for instance,~\cite{hatcher2002algebraic}) (Section~\ref{sec:boundary_operator}).

We then introduce linear and legal routing: moving particles along arcs according to a routing vector~$\alpha$. Linear routing is fully governed by the image of~$\partial$, whereas legal routing additionally enforces local non-negativity constraints (Section~\ref{sec:legal_routing}).

Next, we address the central question of the existence of a legal routing. We provide a necessary and sufficient characterization in terms of a nonnegative routing vector satisfying specific structural conditions, and explain how such a vector yields a legal routing sequence (Section~\ref{sec:arc_routing_vec}). This leads to a polynomial-time reduction to a classical flow or weighted matching problem to determine the existence of such a vector.

Finally, when a routing vector~$\alpha$ is prescribed, we analyze when it can be fully realized by a legal sequence in \thref{prop:CNS_routing_vector}, the main result of this section. We identify the possible obstructions and introduce additional criteria ensuring that the entire vector~$\alpha$ can be legally executed (Section~\ref{sec:transitory}).

\subsection{Boundary operator}
\label{sec:boundary_operator}
Let $G=(V,A,\head,\tail)$ be a multigraph.  Recall from Subsection~\ref{sec:free_abelian} 
that $C_V$ is the free group on $V$ and $C_A$ the free group on $A$, and their elements are respectively called particle configurations and arc configurations. By the universal property of free groups, note that $\head$ and $\tail$ extend in a unique way as homomorphisms from $C_A$ to $C_V$.
Denote by  $\cW$  the set of weakly connected components of $G$ and by $C_\cW$ the free group on $\cW$.  
The connected component of
a vertex $v$ is denoted by $\deg(v)$, and we keep the same notation for the extension of $\deg$ as an homomorphism called degree from $C_V$ to $C_\cW$. If we denote by $\deg_w(\sigma)$ the coefficient of $w$ in $\deg(\sigma)$, then $\deg_w(\sigma) = \sum_{v \in w} \sigma_v$ is the sum of all $\sigma_v$ for all $v \in V$ in the component $w$. 
This {\bf degree} measures the total number of particles in each component, and for all particle routing definitions used in this paper, it remains invariant. As particles move along arcs, they stay within the same connected components, ensuring the degree is preserved.

Routing involves moving particles along arcs in both forward and backward directions, requiring careful tracking of the arcs used. To this end, we introduce the boundary operator $\partial : C_A \rightarrow C_V$, defined as 
$$\partial = \head - \tail.$$ 

For a particle configuration $\sigma \in C_V$ and an arc $a$, $\sigma + \partial(a)$ represents the configuration obtained by moving a particle from $\tail(a)$ to $\head(a)$, while $\sigma + \partial(-a)$ represents the reverse movement. This movement is algebraic, meaning $\sigma$ does not need to be positive where particles are taken. More generally, $\sigma + \partial(\alpha)$ for $\alpha \in C_A$ consists of moving particles of $\sigma$ along the arcs appearing in $\alpha$.

If $P = (a_1, a_2, \cdots, a_k)$ is an undirected path in $G$ from $x$ to $y$, we associate to this path the sum $\alpha = \sum_{i=1}^k \alpha_i a_i$ with coefficient $\alpha_i=1$ if $a_i$ is  forward-oriented in the path, and $\alpha_i=-1$ if it is backward-oriented. If $\alpha \in C_A$ is constructed from an undirected path $P$ from $x$ to $y$ in this way, we say that it represents the path $P$. In this case, we have 
$$\partial(\alpha)=\partial(\sum_{i=1}^k \alpha_i a_i) = y - x.$$
The rest of this section is devoted to studying some properties of $\partial$.

The image of $\partial$, i.e. the set of $\sigma \in C_V$ such that there is $\alpha \in C_A$ with $\partial(\alpha)=\sigma$, is denoted by $B_V$. It is a subgroup of $C_V$ whose elements are called {\bf boundaries}. The following is well known and easy:

\begin{proposition}
\thlabel{prop:image_bord}
    The image $B_V \subset C_V$ of $\partial$ is the kernel of $\deg$.
\end{proposition}

The kernel of $\partial$, i.e. the subgroup of arc configurations $\alpha \in C_A$ such that $\partial(\alpha)=0_V$, is denoted as $Z_A$. Its elements are called \textbf{cycles} and $Z_A$ the \textbf{cycle space}.  

\begin{proposition} \thlabel{prop:kirchoff}
    Elements of $Z_A$ are characterized by Kirchoff's Law at every $v \in V$, i.e. for $\alpha \in C_A$
    $$ \sum_{a \in A^+(v)} \alpha_a =  \sum_{a \in A^-(v)} \alpha_a$$
\end{proposition}

If $\alpha_a$ is interpreted as a flow along arc $a$, this equation means that the algebraic sum of flows leaving every vertex $v$ is equal to the sum of entering flows.

\begin{proof}
    We have for any $\alpha \in C_A,$
    $$\partial(\alpha) = \sum_{a \in A} \alpha_a (\head(a) - \tail(a))$$
    $$=\sum_{v \in V} \left(  \sum_{a \in A^-(v)} \alpha_a -  \sum_{a \in A^+(v)} \alpha_a \right)\cdot v,$$
    hence the result.
\end{proof}

\subsection{Linear and legal free routings}
\label{sec:legal_routing}

In the multigraph $G$, a {\bf linear free routing} with {\bf routing vector} $\alpha \in C_A$  is the operation of transforming $\sigma \in C_V$ into $\sigma + \partial(\alpha)$. This amounts to moving 
simultaneously $|\alpha_a|$ particles along each arc $a \in A$, forward if $\alpha_a >0$ and backward if $\alpha_a <0$. This routing 
can always be applied without further conditions. A {\bf routing sequence} for $\alpha$ is a finite sequence $(a_i)_{0 \leq i \leq k-1}$ of arcs, such that $\alpha=\sum_{0 \leq i \leq k-1} a_i$. It describes an order in which we can route along arcs in $\alpha$ one by one. A routing sequence from $\sigma$ to $\sigma'$ is a routing sequence such that $\sigma' = \sigma + \partial(\alpha)$, where $\alpha$ is the routing vector of the sequence.

If $\sigma' = \sigma + \partial(\alpha)$ with $\sigma, \sigma' \in C_V$ and $\alpha \in C_A$, we write  $\sigma \lineareq{\partial}{\alpha}  \sigma'$ and say that there is a linear routing
from $\sigma$ to $\sigma'$ with routing vector $\alpha$. Alternatively, we write $\sigma \lineareq{\partial}{*}  \sigma'$ if there exists an $\alpha \in C_A$ such that $\sigma \lineareq{\partial}{\alpha}  \sigma'$. This clearly defines an equivalence relation $\lineareq{\partial}{*}$ on $C_V$.
From the results of Sections \ref{sec:boundary_operator} and \ref{sec:boundary_operator}, it follows that 
$\sigma \lineareq{\partial}{*}  \sigma'$ if and only if $\sigma'- \sigma \in B_V$, which is equivalent to $\deg(\sigma)=\deg(\sigma')$. If $\sigma \lineareq{\partial}{\alpha}  \sigma'$, then the set of all possible routing vectors from $\sigma$ to $\sigma'$ is $\alpha + Z_A$.

Among these linear routings, we want to identify those that correspond to legal routings, which we now define. Once again, this is not to be confused with rotor-routing as defined in Sec.~\ref{sec:standard}: here there are neither rotor orderings nor rotor configurations, and we simply move particles  along arcs of the graph. 
We say that the linear free routing from $\sigma$ with routing vector $a \in A$ is {\bf legal} for $\sigma$ if $\sigma_{\tail(a)} \geq 1$. 
A routing sequence $(a_i)_{0 \leq i \leq k-1}$ from $\sigma$ to $\sigma'$ is legal if for every $i$, routing along $a_i$ is legal for $\sigma_{i}$,  where $\sigma_0 = \sigma$ and $\sigma_{i+1} = \sigma_i + \partial(a_i)$ for $0 \leq i \leq k-1$ (and consequently $\sigma_k = \sigma'$). If there exists a legal routing sequence from $\sigma$
to $\sigma'$, we write $\sigma \legalseq{\partial}{*}  \sigma'$, and use notation $\sigma \legalseq{\partial}{\alpha}  \sigma'$ if we want to point out that there is a legal sequence with routing vector $\alpha$. Note that by definition, $\sigma \legalseq{\partial}{\alpha}  \sigma'$ implies $\sigma \lineareq{\partial}{\alpha}  \sigma'$ and $\sigma \legalseq{\partial}{*}  \sigma'$, either of which imply $\sigma \lineareq{\partial}{*}  \sigma'$.

With these definitions, we observe that legal free routing is a special case of VAS as discussed in Sec.~\ref{sub:vas}, which is conservative in the total number of particles. The legal reachability problem, in this context, involves decomposing a linear relation into a legal sequence.

\subsection{Existence of a legal free routing}
\label{sec:arc_routing_vec}

We now study how we can decide if $\sigma \legalseq{\partial}{*} \sigma'$ and compute a legal sequence. A necessary condition for its existence is that there is a nonnegative routing vector from $\sigma$ to $\sigma'$, i.e. $\alpha \in C^+_A$ with $\sigma \lineareq{\partial}{\alpha} \sigma'$. 
In the case where $\sigma, \sigma' \geq 0$, the problem reduces to moving particles in any order from  $V^+ = \{v \in V : \sigma_v > \sigma'_v \}$ to $V^- = \{v \in V : \sigma_v < \sigma'_v\}$ and can be viewed as a flow or matching problem, by deciding which particle in $V^+$
will be routed to which vertex in $V^-$, counted with multiplicities.
However, in presence of vertices with negative values, the situation is a little more complicated, as shows the example on Fig. \ref{fig:contre_ex_r_positif}.

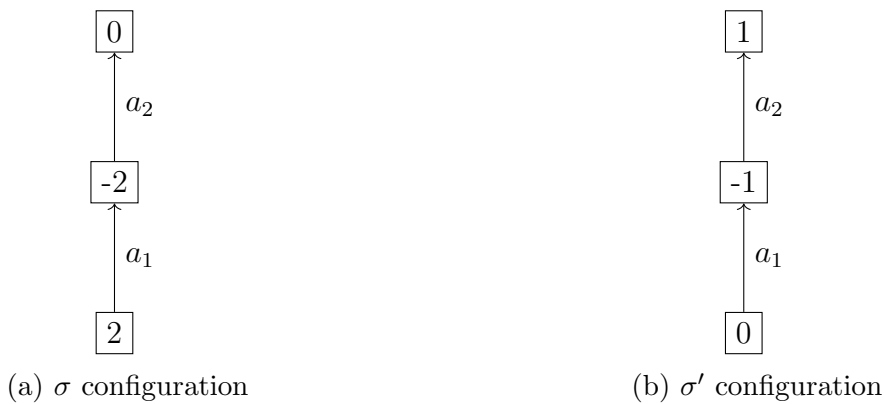
\begin{figure}[!htbp]
    \centering
    \begin{subfigure}{0.45\textwidth}
        \centering
        \begin{tikzpicture}
            \node[draw] (a) at (0, 0) {2};
            \node[draw] (b) at (0, 2) {-2};
            \node[draw] (c) at (0, 4) {0};
            \draw[->] (a) -- node[midway, right] {$a_1$} (b);
            \draw[->] (b) -- node[midway, right] {$a_2$} (c);
    
        \end{tikzpicture}
        \caption{ $\sigma$ configuration}
    \end{subfigure}
    \hspace{0.05\textwidth}
    \begin{subfigure}{0.45\textwidth}
        \centering
        \begin{tikzpicture}
            \node[draw] (a) at (0, 0) {0};
            \node[draw] (b) at (0, 2) {-1};
            \node[draw] (c) at (0, 4) {1};
            \draw[->] (a) -- node[midway, right] {$a_1$} (b);
            \draw[->] (b) -- node[midway, right] {$a_2$} (c);
    
        \end{tikzpicture}
        \caption{$\sigma'$ configuration}
    \end{subfigure}
    \caption{A graph with $\sigma$ given on the left and $\sigma'$ on the right. There is a single linear
    routing vector $\alpha$ from $\sigma$ to $\sigma'$
    which is $2a_1+a_2$. However, none of the three sequences $(a_1, a_1, a_2)$, $(a_1, a_2, a_1)$, $(a_2, a_1, a_1)$ are legal because the middle vertex remains nonpositive at all times.} 
    \label{fig:contre_ex_r_positif}
\end{figure}

\vskip .5cm

The first problem is that vertices $v$ with $\sigma'_v < 0$ can never be involved in a legal routing: they must remain \emph{inactive}. These vertices act like an obstacle to legal routing operations and could as well be turned into sinks. However, vertices can begin with $\sigma_v < 0$, receive particles and become active during the sequence. 

In the rest of this section, we suppose that $G, \sigma, \sigma'$
are given and that we want to find a legal routing sequence from $\sigma$ to $\sigma'$. 
If $\alpha$ is the routing vector of a legal routing sequence, then $\alpha \in C^+_A$. We recall that if $\alpha \in C^+_A$, we say that
$a$ is an element of $\alpha$ and write $a \in \alpha$, if $\alpha_a \geq 1$.

We solve the problem in two steps:

\begin{itemize}
    \item first, we show that $\sigma \legalseq{\partial}{*} \sigma'$ is equivalent to the existence of $\alpha \in C_A^+$
    with $\sigma \lineareq{\partial}{\alpha} \sigma'$, 
    satisfying additional properties (\thref{prop:CNS_routing1}). We call
    such a vector a \emph{legal routing vector} from $\sigma$
    to $\sigma'$, since it asserts the existence of the legal routing sequence.  We also describe how to construct greedily a legal routing sequence from a legal routing vector (\thref{lemma:CNS_routing1}).
    \item Second, we show that deciding existence and computing a legal routing vector can be decided by standard 
    weighted matching or network flow algorithms (\thref{prop:flow_problem}) (hence in polynomial time).
\end{itemize}

\subsubsection{From legal vectors to legal sequences}

If $\alpha,\alpha' \in C^+_A$, we write $\alpha \leq \alpha'$ if for all $a \in A$,  we have $\alpha_a \leq \alpha'_a$. A vertex $v$ is {\bf active} in $\alpha$ if there is an element $a \in \alpha$ with $\tail(a)=v$. An active vertex is a vertex that will emit particles during a routing with routing vector $\alpha$.

A {\bf legal routing vector} from $\sigma$ to $\sigma'$ is $\alpha \in C^+_A$ 
with $\sigma \lineareq{\partial}{\alpha}  \sigma'$, such that every active vertex $v$ in $\alpha$ satisfies $\sigma'_v \geq 0$.    

\begin{proposition}
    \thlabel{prop:CNS_routing1}
    Let $\sigma, \sigma' \in C_V$. Then $\sigma \legalseq{\partial}{*}  \sigma'$  if and only if there is a legal routing vector $\alpha$ from $\sigma$ to $\sigma'$.
    In this case there is $\alpha' \in C^+_A$ with $\sigma \legalseq{\partial}{\alpha'}  \sigma'$
     and $\alpha' \leq \alpha$.
    \end{proposition}

The proof is based on the following lemma, which explains how to construct greedily a legal routing sequence from a legal routing vector. We say that a routing sequence is {\bf $\alpha$-bounded}, if the routing vector $\alpha'$ of the sequence satisfies $\alpha' \leq \alpha$. Let ${\cal S}(\sigma,\alpha)$ be the set
of $\alpha$-bounded routing sequences that are legal for $\sigma$.
If $s_1,s_2 \in {\cal S}(\sigma,\alpha)$, we say that $s_1 \leq  s_2$
if $s_1$ is a prefix of $s_2$. 

\begin{lemma} 
    \thlabel{lemma:CNS_routing1}

    Suppose $\alpha \in C^+_A$ is a legal routing vector from $\sigma$ to $\sigma'$.
    Let $s \in {\cal S}(\sigma,\alpha)$ with routing vector $\alpha^s$
    such that $\sigma^s = \sigma +\partial(\alpha^s) \neq \sigma'$. 
    Then the set of arcs $a \in \alpha-\alpha^s$ whose routing
    is legal for $\sigma^s$ is nonempty, and
    adding any such arc $a$ to $s$ extends $s$ into $s \cdot a \in {\cal S}(\sigma,\alpha)$.
\end{lemma}

\begin{proof}
    Since $\deg(\sigma^s)=\deg(\sigma)=\deg(\sigma')$, there is a vertex $v$ such that $\sigma^s_v > \sigma'_v$.
    Since $\sigma^s \lineareq{\partial}{\alpha-\alpha^s}  \sigma',$
    there is $a \in A^+(v)$ with $a \in \alpha - \alpha^s$. Then $a \in \alpha$, which is a legal routing vector, so
    $\sigma'_v \geq 0$ and  $\sigma^s_v \geq 1$.
    Routing along $a$ is legal for $\sigma^s$, hence $s \cdot a \in {\cal S}(\sigma,\alpha)$.
\end{proof}

\begin{proof}[Proof of \thref{prop:CNS_routing1}]
    \thref{lemma:CNS_routing1}   shows that the condition is sufficient. Conversely, suppose that there is a legal routing sequence $(a_i)_{0 \leq i \leq k-1}$ from $\sigma$ to $\sigma'$,
    and let $\alpha \in C^+_A$ be its routing vector. If $v$ is an active vertex in $\alpha$, then there is $a_i$ with $\tail(a_i)=v$; consider the last such arc, denoted by $a_\ell$. Let $\sigma^i = \sigma + \partial(\sum_{j=0}^{i-1} a_j)$. Since the routing is legal, we have $\sigma^\ell_v \geq 1$ and $\sigma^{\ell+1}_v \geq 0$. All arcs $a_i$ for $i>\ell$ satisfy $\tail(a_i) \neq v$,
    hence $\sigma^i_v$ cannot decrease for $i=\ell+1,...,k-1$. Hence $\sigma'_v \geq 0$. We proved that $\alpha$ is legal.
\end{proof}

\subsubsection{Existence of legal routing vectors}

We now reduce the problem of deciding the existence, and computing a legal routing vector from $\sigma$ to $\sigma'$, to the following problem:

\vskip .5cm

\begin{tabularx}{15cm}{lp{11cm}}
    \hline \multicolumn{2}{c}{{\sc Bipartite Weighted Degree Constraint Problem (BWDC) }} \\
    \hline
    {\sc input:} &  a simple bipartite undirected graph $(V_1,  V_2,E)$ and a weight  function $w : V_1 \cup V_2 \rightarrow \N$.\\       
    {\sc question:} & is there a weight function on edges $f : E \rightarrow \N$ such that for every vertex $v \in V_1 \cup V_2$, the sum of all weights $f(e)$ of edges incident to $v$ is equal to $w(v)$ ? \\
    \hline
\end{tabularx}

\vskip .5cm

This problem is a classic network flow problem which can be solved by standard strongly polynomial algorithms like Edmond-Karp's or Dinic's algorithms \cite{edmonds1972theoretical, dinic1970algorithm}. 

For the reduction, suppose that $G, \sigma, \sigma'$ are given. 
Let $V^+ = \{v \in V : \sigma_v > \sigma'_v \}$ and $V^- = \{v \in V : \sigma_v < \sigma'_v\}$, 
and consider the following instance $(V^+,V^-,E_1,w)$ of {\sc BWDC} where:

\begin{itemize}
    \item $E_1$ contains an edge $e$ between $v^+ \in V^+$
      and $v^- \in V^-$ if and only if there is a directed path from
      $v^+$ to $v^-$ in $G'$, where $G'$ is 
      obtained from $G$ by removing all arcs with tail in $\{v \in V : \sigma'_v < 0\}$
      \item $(V^+,V^-,E_1)$  is the bipartite graph
      \item the weight of $v^+ \in V^+$ is $w(v^+)=\sigma_{v^+} - \sigma'_{v^+}$ and the weight
      of $v^- \in V^-$ is $w(v^-)=\sigma'_{v^-} - \sigma_{v^-}$
      
\end{itemize}

\begin{proposition} \thlabel{prop:flow_problem}
        There is a legal routing vector from $\sigma$ to $\sigma'$
        in $G$ if and only if {\sc BWDC} has a solution on instance $(V^+, V^-,E_1,w)$. From
        any solution, a legal routing vector can be computed in polynomial time.
\end{proposition}

\begin{proof}
    First suppose that there is a legal routing vector $\alpha$. 
    Let $M=\sum_{v \in V^+} (\sigma_v-\sigma'_v)$. A consequence of the existence of a legal routing sequence (\thref{prop:CNS_routing1}) is that we can find a legal routing vector $\alpha' \leq \alpha$  from $\sigma$ to $\sigma'$, that can be written as $\alpha' = \sum_{i=1}^M \alpha'_i$ and $\alpha'_i$ represents a directed path from $V^+$ to $V^-$. 
    We now define the weight $f(e)$ of an edge of $E_1$ from $v^+$
    to $v^-$ as the number of directed paths from $v^+$
    to $v^-$ among $\alpha'_i$; then $f$ is a solution to the {\sc BWDC} instance.

    Conversely, if $f$ is a solution of the {\sc BWDC}
    problem, we consider for every edge $e \in E$
    from $v^+$ to $v^-$, an arc configuration $\alpha_e \in C_A$ representing
    a directed path from $v^+$ to $v^-$ in $G'$. Then
    $$\alpha = \sum_{e \in E_1} f(e) \cdot \alpha_e$$
    is a routing vector from $\sigma$ to $\sigma'$,
    which is legal since all its elements belong to $G'$.
\end{proof}

Before ending this section, remark that the existence of a legal routing vector $\alpha$ just certifies that $\sigma \legalseq{\partial}{\alpha'} \sigma'$ with some $\alpha' \leq \alpha$, and does not imply that $\sigma \legalseq{\partial}{\alpha} \sigma'$, as shows the example on Fig. \ref{fig:unlegal_routing_0}. The situation is very similar to the notion of flows and runs for rotor-routing as explained in Sec. \ref{sec:standard}.

\begin{figure}[!htbp]
    \centering
    \begin{subfigure}{0.45\textwidth}
        \centering
        \begin{tikzpicture}

            \node[draw] (a) at (0, 0) {1};
            \node[draw] (b) at (0, 2) {-1};
            \draw[->] (a) -- node[midway, right] {$a_1$} (b);
            \draw[->, loop left, out=150, in=210, looseness=8] (b) to node[left] {$a_2$} (b);
    
        \end{tikzpicture}
        \caption{ $\sigma$ configuration}
    \end{subfigure}
    \hspace{0.05\textwidth}
    \begin{subfigure}{0.45\textwidth}
        \centering
        \begin{tikzpicture}
            \node[draw] (a) at (0, 0) {0};
            \node[draw] (b) at (0, 2) {0};

            \draw[->] (a) -- node[midway, right] {$a_1$} (b);
            \draw[->, loop left, out=150, in=210, looseness=8] (b) to node[left] {$a_2$} (b);
        \end{tikzpicture}
        \caption{$\sigma'$ configuration}
    \end{subfigure}
    \caption{Graph $G$ with $\sigma$ on the left and $\sigma'$
    on the right. The vector $\alpha=a_1+a_2$ is a legal routing vector from
    $\sigma$ to $\sigma'$, but the only legal routing sequence
    is $(a_1)$ with routing vector $\alpha'=a_1 < \alpha$.}
    \label{fig:unlegal_routing_0}
\end{figure}
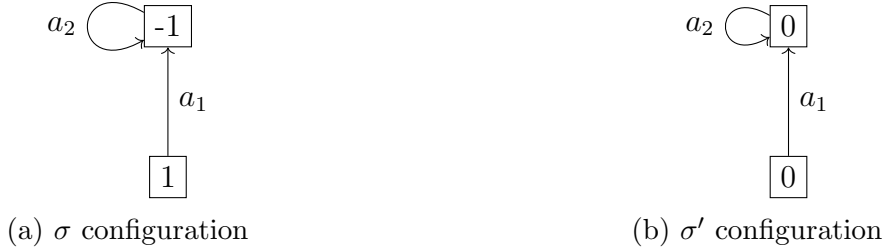

\subsection{Existence of legal routing with given routing vector.}
\label{sec:transitory}

We just mentioned that legal routing vectors characterize the existence of legal routing sequences,  but that not every legal routing vector can be obtained as the routing vector of a legal sequence. 
We now state conditions that ensure
$\sigma \legalseq{\partial}{\alpha} \sigma'$ if $\sigma \lineareq{\partial}{\alpha} \sigma'$ with $\alpha \in C^+_A$. This is analogous to
the characterization of runs for rotor-routing in Section \ref{sec:standard}
but for free routing in graphs.

In the proof of \thref{prop:CNS_routing1}, we showed that if $\sigma \lineareq{\partial}{\alpha}  \sigma'$ with $\alpha \in C^+_A$ and every active vertex $v$ in $\alpha$ satisfies $\sigma'_v \geq 0$, we can build greedily a legal routing sequence from $\sigma$
to $\sigma'$ whose routing vector $\alpha'$ satisfies $\alpha' \leq \alpha$. However, if we construct this sequence without additional precaution, this process can end with $\alpha \neq \alpha'$ as shown in the example in Fig.~\ref{fig:deux_graphes}.

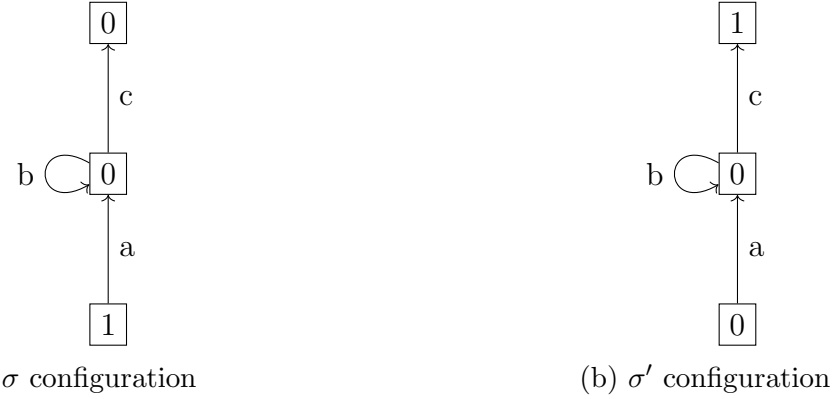
\begin{figure}[!htbp]
    \centering
    \begin{subfigure}{0.45\textwidth}
        \centering
        \begin{tikzpicture}
            \node[draw] (a) at (0, 0) {1};
            \node[draw] (b) at (0, 2) {0};
            \node[draw] (c) at (0, 4) {0};
            \draw[->] (a) -- node[midway, right] {a} (b);
            \draw[->] (b) -- node[midway, right] {c} (c);
            \draw[->, loop left, out=150, in=210, looseness=8] (b) to node[left] {b} (b);
    
        \end{tikzpicture}
        \caption{ $\sigma$ configuration}
    \end{subfigure}
    \hspace{0.05\textwidth}
    \begin{subfigure}{0.45\textwidth}
        \centering
        \begin{tikzpicture}
            \node[draw] (a) at (0, 0) {0};
            \node[draw] (b) at (0, 2) {0};
            \node[draw] (c) at (0, 4) {1};

            \draw[->] (a) -- node[midway, right] {a} (b);
            \draw[->] (b) -- node[midway, right] {c} (c);
            \draw[->, loop left, out=150, in=210, looseness=8] (b) to node[left] {b} (b);
        \end{tikzpicture}
        \caption{$\sigma'$ configuration}
    \end{subfigure}
    \caption{A graph with $\sigma$ given on the left and $\sigma'$ on the right. We want to route a particle from bottom to top. Clearly $\sigma' = \sigma + \partial(a+b+c)$, and $(a,b,c)$ is a legal routing sequence from $\sigma$ to $\sigma'$. However,
    if we try to construct greedily a routing 
    legal sequence and do not select $b$ before $c$,
    we obtain the legal sequence $(a,c)$ from
    $\sigma$ to $\sigma'$, which is legal
    but cannot be extended in a longer legal 
    sequence: we forgot to use arc $b$ before using $c$.}
    \label{fig:deux_graphes}
\end{figure}

To ensure that the whole routing vector $\alpha$ corresponds to a legal routing sequence, we need an additional condition on $\alpha$, which is met on the example in Fig.~\ref{fig:deux_graphes}. We need a little more terminology to describe the condition nicely.

Consider $\sigma, \sigma'$ and $\alpha \in C^+_A$ such that $\sigma \lineareq{\partial}{\alpha}  \sigma'$. In this context, a vertex $v$ is {\bf transitory} if it is active in $\alpha$ and if $\sigma'_v = 0$.  A transitory vertex is an active vertex in which particles will not settle in the end. We denote by $\trans(\alpha,\sigma')$ the set of transitory vertices 
in the context of $\sigma \lineareq{\partial}{\alpha}  \sigma'$.

Likewise, we say that a set of arcs ${\cal R} \subset A$ is  {\bf guiding} for $\sigma \lineareq{\partial}{\alpha}  \sigma'$
if:
\begin{enumerate}[(i)]
    \item ${\cal R}$ contains only elements of $\alpha$;
    \item for every transitory vertex $v \in \trans(\alpha,\sigma')$, there is  a directed path with all arcs in ${\cal R}$ from $v$ to a vertex $v'  \notin \trans(\alpha,\sigma')$ (i.e. $v'$ is not active, or $\sigma_v' \geq 1)$.
\end{enumerate}
It follows that a guiding set ${\cal R}$ which is minimal for inclusion contains exactly one arc $a \in \alpha$ with $\tail(a)=v$ for every transitory vertex $v$, and contains no directed cycles: hence such an ${\cal R}$ is a directed forest with domain $\trans(\alpha,\sigma')$. Conversely, any directed forest 
with domain $\trans(\alpha,\sigma')$ is guiding, and we call it a {\bf guiding forest}. For a guiding forest ${\cal F}$ and a transitory vertex $v  \in \trans(\alpha,\sigma')$, we denote by ${\cal F}(v)$ the arc $a$ in ${\cal F}$ with $\tail(a)=v$.

Consider as an example  Fig.~\ref{fig:deux_graphes} once again. 
Both $a+c$ and $a+b+c$ are routing vectors that admit a legal routing sequence.  In each case, the two vertices whose value is 0 in $\sigma'$ (middle and down) are  transitory since they are active in the routing, and the unique guiding forest is $\{a,c\}$. 

Let ${\cal F}$ be a guiding forest for $\sigma \lineareq{\partial}{\alpha}  \sigma'$
and consider an $\alpha$-bounded legal routing sequence  $s = (a_i)_{0 \leq i \leq k-1}$ for $\sigma$, with routing vector $\alpha'$.   We say that $s$ is {\bf guided by} $(\sigma, \alpha, {\cal F})$ if for every transitory vertex $v \in \trans(\alpha,\sigma')$:
\begin{itemize}
    \item if $v$ is not an active vertex of $\alpha-\alpha'$, then the last arc  with tail $v$ in $s$ is ${\cal F}(v)$;
    \item if $v$ is an active vertex of $\alpha-\alpha'$, then ${\cal F}(v) \in \alpha-\alpha'$.
\end{itemize}

The main idea in this definition is that if a legal routing sequence has routing vector $\alpha'=\alpha$, then saying that it is guided by ${\cal F}$ is just
saying that ${\cal F}(v)$ is the last arc of the sequence with tail $v$, for every transitory vertex $v$; otherwise if the routing vector $\alpha'$
is such that $\alpha' \lneq \alpha$, the condition ensures ${\cal F}(v)$ remains available
in order to extend the sequence. Note that by construction,
a prefix of legal routing sequence that is guided by $(\sigma,\alpha,{\cal F})$,
is also guided by $(\sigma,\alpha,{\cal F})$.
In the example in Fig. \ref{fig:deux_graphes}, a legal sequence guided by ${\cal F}=\{a,c\}$ for $\alpha=a+b+c$ will have to use arc $b$ before arc $c$.

\begin{lemma}
\thlabel{lemma:routing2tree}
    Let $(a_i)_{0 \leq i \leq k-1}$ be a legal routing sequence for
    $\sigma \legalseq{\partial}{\alpha}  \sigma'$. For any transitory
    vertex $v \in \trans(\alpha,\sigma')$, let ${\cal F}(v)$ be the last arc with tail $v$
    appearing in the sequence. Then ${\cal F}:=\{ {\cal F}(v) \text{ for } v  \in \trans(\alpha,\sigma')\}$ is a guiding forest for 
    $\sigma \lineareq{\partial}{\alpha}  \sigma'$, and $(a_i)_i$ is guided
    by ${\cal F}$.
\end{lemma}

\begin{proof}
    Let $(\sigma^i)_{0 \leq i \leq k}$ be the sequence of particle
    configurations in the legal sequence.
    Consider a transitory vertex $v$. By construction ${\cal F}$
    contains exactly one arc of $\alpha$ with tail $v$;
    let $i_v$ be the last index with
    $a_{i_v} = {\cal F}(v)$ and let $v' = \head({\cal F}(v))$.
    If $v'$ is also transitory then $i_{v'} > i_v$: indeed, suppose that $i_{v'} \leq  i_v$: then $\sigma^{i_v}_{v'} \geq 0$ and then $\sigma^{i_v+1}_{v'} \geq 1$. By definition of $i_{v'}$, all arcs $a_i$ for $i > i_v+1$ satisfy $\tail(a_i) \neq v'$, so that $\sigma'_{v'} \geq \sigma^{i_v+1}_{v'} \geq 1$.  This contradicts the fact that $\sigma'_{v'}=0$. It follows that there can be no directed cycle in ${\cal F}$, 
    and no undirected cycle as well, so ${\cal F}$ is a guiding forest. 
\end{proof}

The next lemma proves that under conditions of  \thref{lemma:CNS_routing1}, together with the existence of a guiding forest ${\cal F}$, we can build greedily a legal routing sequence from $\sigma$ to $\sigma'$ that will have a routing vector exactly $\alpha$. The difference with \thref{lemma:CNS_routing1} is that we must make sure for transitory vertices that the last routing in which they are involved is the arc ${\cal F}(v)$. Hence, the guiding forest is the "last exit" of a particle in a transitory vertex, and will guide the particle to its final position. This is very similar to the correspondence between Eulerian circuits and directed forests in an Eulerian directed graph (see for instance \cite{van1987circuits}).

Let ${\cal S}(\sigma, \alpha, {\cal F})$ be the set of $\alpha$-bounded legal sequences for $\sigma$ that are guided by $(\sigma,\alpha,{\cal F})$.

\begin{lemma}
    \thlabel{lemma:CNS_routing2}
    
    Let $\alpha \in C^+_A$. Suppose that $\sigma \lineareq{\partial}{\alpha}  \sigma'$,
    that every active vertex $v$ of $\alpha$ satisfies $\sigma'_v \geq 0$,
    and that there exists a guiding forest ${\cal F}$ for the routing.
    Let $s \in {\cal S}(\sigma,\alpha,{\cal F})$ with routing vector $\alpha^s$
    such that $\alpha^s \neq \alpha$. Let $\sigma^s = \sigma + \partial(\alpha^s)$.
    
        Then there is an active vertex $v$ of $\alpha-\alpha^s$ with $\sigma^s_v \geq 1$.
    Define an arc $a \in A^+(v)$ by:
    \begin{itemize}
        \item if $v \notin \trans(\alpha,\sigma')$, let $a$ be any arc $a \in \alpha-\alpha^s$ with tail $v$;

        \item if $v \in \trans(\alpha,\sigma')$ and ${\cal F}(v)$ is the only arc in $\alpha-\alpha^s$ with tail $v$, let $a={\cal F}(v)$;

        \item if $v \in \trans(\alpha,\sigma')$ and there are at least
        two distinct arcs in $\alpha-\alpha^s$ with tail $v$,
        let $a$ be any such arc that is not ${\cal F}(v)$.
    \end{itemize}
    Then $s \cdot a \in {\cal S}(\sigma,\alpha,{\cal F})$.
\end{lemma}

\begin{proof}

    We prove first that there is an active vertex $v$ of $\alpha-\alpha^s$, such that $\sigma^s_v \geq 1$. Suppose that there is no such vertex and consider the set $V_s$  of active vertices of $\alpha-\alpha^s$. We then have
    \[ \sum_{v \in V_s} \sigma^s_v = 0.\]
    Since elements of $V_s$ are also active for $\alpha$, we see by hypothesis that $\sigma'_v \geq 0$ for every $v \in V_s$.  
    We prove that from $\sigma^s$ to $\sigma'$, the total number
    of particles in $V_s$ cannot increase: indeed, for every active arc of $\alpha-\alpha^s$, one particle is lost in $V_s$ and at most one is gained (in the case where the arc
    has also its head in $V_s$). Precisely:
    \[ \sum_{v \in V_s} \sigma'_v = \sum_{v \in V_s} (\sigma^s_v + \partial(\alpha-\alpha^s)_v ), \]
    with
    \[ \sum_{v \in V_s} \partial(\alpha-\alpha^s)_v  =  \sum_{v \in V_s} (\head(\alpha-\alpha^s))_v - \sum_{v \in V_s} (\tail(\alpha-\alpha^s))_v .\]
    We have 
    \[ \sum_{v \in V_s} (\tail(\alpha-\alpha^s))_v =  \sum_{a \in \alpha-\alpha^s} (\alpha-\alpha^s)_a \]
    since every arc of $\alpha-\alpha^s$ has  its tail in $V_s$  by definition of $V_s$.
    On the other hand, we see that 
    \begin{equation} \label{eq:bilan1}
       \sum_{v \in V_s} (\head(\alpha-\alpha^s))_v \leq  \sum_{a \in \alpha-\alpha^s} (\alpha-\alpha^s)_a  
    \end{equation}
    thus 
        \[ \sum_{v \in V_s} \partial(\alpha-\alpha^s)_v \leq 0,\]
    hence 
    $$0 \leq \sum_{v \in V_s} \sigma'_v \leq \sum_{v \in V_s} \sigma^s_v = 0, $$
    the last equality following from our hypothesis.
    It follows that $\sigma'_v = 0$ for every $v \in V_s$ so (\ref{eq:bilan1}) is an equality, i.e. for every arc $a \in  \alpha-\alpha^s$, $\head(a) \in V_s$.
    Also note that $V_s$ contains only transitory vertices, in $\trans(\alpha-\alpha^s,\sigma')$,
    hence also in $\trans(\alpha,\sigma')$.
    By hypothesis $s \in {\cal S}(\sigma, \alpha, {\cal F})$, 
    so for every $v \in V_s$, we have  ${\cal F}(v) \in \alpha-\alpha^s$ . Since arcs $a \in \alpha-\alpha^s$ are such that $\head(a) \in V_s$ and $\tail(a) \in V_s$, it implies that ${\cal F}$ 
    contains a directed cycle, which is a contradiction.

    Finally, we proved that there is an active vertex $v^1$ of $\alpha-\alpha^s$, such that $\sigma^s_{v^1} \geq 1$. Choosing any arc $a \in \alpha-\alpha^s$ with tail $v^1$ will result
    in extending $s$ into a sequence $s \cdot a$ which is $\alpha$-bounded and legal. 
    It remains to see how to choose $a$ such that $s \cdot a$ will still be guided 
    by $(\sigma,\alpha,{\cal F})$. If $v^1 \notin \trans(\alpha,\sigma')$, any such arc $a$ will work.
    Otherwise, ${\cal F}(v^1) \in \alpha-\alpha^s$ since $s \in {\cal S}(\sigma, \alpha, {\cal F})$. Consider two cases:
    \begin{itemize}        
        \item if ${\cal F}(v^1)$ is the only arc with tail $v^1$ in $\alpha-\alpha^s$, then we choose
        $a = {\cal F}(v^1)$. If $v^1$ becomes inactive for $\alpha-\alpha^s-{\cal F}(v^1)$ then the last
        arc with tail $v^1$ in $s \cdot a$ is ${\cal F}(v^1)$, as required; otherwise $v^1$
        is active in $\alpha-\alpha^s-{\cal F}(v^1)$ and we still have ${\cal F}(v^1) \in \alpha - \alpha^s - {\cal F}(v^1)$;
        \item otherwise, let $a$ be another arc with tail $v^1$. Then $v^1$
        is active in $\alpha-\alpha^s-a$ and we still have ${\cal F}(v^1) \in \alpha - \alpha^s - a$.
    \end{itemize}
    
\end{proof}

Here is the main result of this section:

\begin{theorem}
    \thlabel{prop:CNS_routing_vector}
    If $\sigma \lineareq{\partial}{\alpha}  \sigma'$ with $\alpha \in C_A^+$,
    then $\sigma \legalseq{\partial}{\alpha} \sigma'$
    if and only if 
    \begin{enumerate}[(i)]
        \item every active vertex $v$ satisfies $\sigma'_v \geq 0$, and
        \item the set of arcs which are elements of $\alpha$ is guiding.
        \end{enumerate}
    \end{theorem}

\begin{proof}
    If the set of elements of $\alpha$ is guiding, then  $\sigma \lineareq{\partial}{\alpha}  \sigma'$ admits a guiding forest so by \thref{lemma:CNS_routing2}, the conditions are sufficient. 
    
    Conversely if there is a legal routing sequence with routing vector $\alpha$, then Condition $(i)$ is necessary by \thref{prop:CNS_routing1} and Condition $(ii)$ is necessary by \thref{lemma:routing2tree}. 
\end{proof}

Figure~\ref{fig:unlegal_routing} illustrates this result.

\begin{figure}[!htbp]
    \centering
    \begin{subfigure}{0.45\textwidth}
        \centering
        \begin{tikzpicture}

            \node[draw] (a) at (0, 0) {1};
            \node[draw] (b) at (0, 2) {-1};
            \draw[->] (a) -- node[midway, right] {a} (b);
            \draw[->, loop left, out=150, in=210, looseness=8] (b) to node[left] {b} (b);
    
        \end{tikzpicture}
        \caption{ $\sigma$ configuration}
    \end{subfigure}
    \hspace{0.05\textwidth}
    \begin{subfigure}{0.45\textwidth}
        \centering
        \begin{tikzpicture}
            \node[draw] (a) at (0, 0) {0};
            \node[draw] (b) at (0, 2) {0};

            \draw[->] (a) -- node[midway, right] {a} (b);
            \draw[->, loop left, out=150, in=210, looseness=8] (b) to node[left] {b} (b);
        \end{tikzpicture}
        \caption{$\sigma'$ configuration}
    \end{subfigure}
    \caption{A graph with $\sigma$ given on the left and $\sigma'$ on the right. Consider $\sigma \lineareq{\partial}{a} \sigma'$: the unique transitory vertex is the bottom one since the top one is not active. Then $\{a\}$ is a guiding forest for that routing and there is a legal routing sequence with routing vector $a$. 
    \newline
    Consider now $\sigma \lineareq{\partial}{a+b} \sigma'$. Both vertices are transitory, but there is no guiding forest for $a+b$. One can easily check that there is no legal routing sequence from $\sigma$ with routing vector~$a+b$.
    }
    \label{fig:unlegal_routing}
\end{figure}
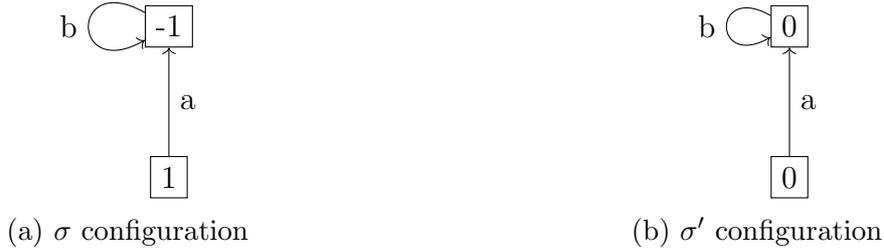

\subsection{Summary of complexity results for reachability questions in free routing}

We can sum up the algorithmic complexities underlying the reachability problems studied in this section by the following result. Polynomial here means polynomial in sizes of $G$, $\sigma$, $\sigma'$ or $\alpha$ (configurations being encoded in binary, i.e. strongly polynomial).

\begin{proposition}\thlabel{prop:summary_free}
    Let $G$ be a multigraph. Let $\sigma, \sigma' \in C_V$ and $\alpha \in C^+_A$
    \begin{enumerate}[(i)]
        \item the complexity of deciding if $\sigma \lineareq{\partial}{*} \sigma'$ is polynomial --  by checking if $\deg(\sigma)=\deg(\sigma')$;
        \item the complexity of deciding if $\sigma \legalseq{\partial}{*} \sigma'$ is polynomial -- by reduction to a network flow algorithm (\thref{prop:flow_problem});
        \item the complexity of deciding if $\sigma \legalseq{\partial}{\alpha} \sigma'$ is polynomial -- check conditions of \thref{prop:CNS_routing_vector}.
    \end{enumerate}
\end{proposition}

Note however that in $(ii)$ and $(iii)$, legal routing sequences can be obtained easily from appropriate routing vectors, but they can have exponential length.


\section{Rotor-routing in Generalized Rotor Mechanisms multigraphs}

Having defined the necessary tools in previous sections, we can define the generalization of standard rotor-routing. The generalization is twofold: first, the graphs where we study rotor-routing will not be limited to standard rotor multigraphs with a cyclic rotor on every vertex. Every vertex will be allowed to have a more complex mechanism for updating arcs, leading to the model of \emph{generalized rotor mechanisms }(GRM) multigraphs. Second, the linear routing takes place simultaneously in the space $C_A \times C_V$ of arcs and vertices, whereas in the previous section we considered just vertices. We will interpret this as free routing, in the sense of Section \ref{sec:boundary}, simultaneously in two graphs.
The motivation is that  it puts the emphasis on the symmetry between transformations in configurations of vertices (particles) and arcs (rotors) during rotor-routing. Furthermore, some important reachability results, developed in Section \ref{sec:legal_grm}, are valid in the generalized case.

As an introduction to GRMs, we reinterpret a single rotor step as two simultaneous free routings (one on the vertex-graph $G^V$, one on the arc-graph $G^A$) (Section~\ref{sec:rotor_as_crs}). We
   then formalize this construction (Section~\ref{sec:GRM}),
   introduce the linear rotor-routing operator $\mathcal{L}:C_F\to C_A\times C_V$ that 
   defines linear rotor-routing in GRMs (Section~\ref{sec:linear_rotor_routing}), and 
   then the legal version of rotor-routing in GRMs (Section~\ref{sec:deflegalroutingGRM}).
   We single out the cyclic GRM subclass and show that it exactly recovers standard rotor-routing; this identifies the bridge between classical results and their GRM generalization (Section~\ref{sec:cyclic_GRM}).
   We finally address algorithmic questions: how to decide and compute routing vectors in this context in polynomial time via linear-algebraic tools (matrix/Smith decomposition) (Section~\ref{sec:comput_routing_vect}).

\subsection{Introducting GRMs}
\label{sec:rotor_as_crs}

Consider once again the rotor multigraph $G_1$ of Fig.~\ref{fig:exempleG1}
and suppose that there is a single particle $\sigma = v_2$ on the vertex $v_2$,
together with a rotor configuration $\rho$ such that $\rho(v_2) = a_{2,1}$, 
as depicted on the right side of Fig.~\ref{fig:rotorgraph_as_CRS}.
If we proceed in a routing step at $v_2$, in the sense of standard rotor-routing, the particle will be transferred to $v_1$,
so that $\sigma$ becomes $\sigma + \partial(\rho(v_2)) = \sigma + \partial(a_{2,1})$, where $\partial$ is the boundary operator in $G$ as defined in Sec.~\ref{sec:boundary}. At the same time, $\rho$ is updated so that $\rho(v_2)$ becomes $\rho'(v_2)=a_{2,0}$. If we see $\rho$ as a formal sum of arcs $\sum_{v \in V} \rho(v)$, we can write this transformation as $\rho + a_{2,0} - a_{2,1}$. We interpret this as the routing of an 'arc particle' in the graph $G^A$ depicted in the left side of Fig.~\ref{fig:rotorgraph_as_CRS}.

In order to distinguish the graphs used to route 'vertex particles' and 'arc particles', we denote by $G^V$ the graph $G_1$ and $\partial^V$ its boundary operator. On the other hand, let us denote by $\partial^A$ the boundary operator of the graph $G^A$. The vertices of $G^A$ correspond to arcs of $G^V$,
and the arcs of $G^A$ join every arc $a$ of $G^V$ to its successor $\theta(a)$ in the rotor ordering; hence $G^A$ is a collection of disjoint directed cycles.
To avoid confusion, while we keep using the standard terminology of vertices and arcs for $G^V$, we respectively use the terms 'arcs' and 'faces' for formal vertices and arcs of $G^A$.

With this notation, the rotor-routing step described just above can be written as
$$(\rho',\sigma') = (\rho, \sigma) + (\partial^A(f_{2,1}),\partial^V(a_{2,1})).$$
Note that $a_{2,1} = \tail^A (f_{2,1})$, where $\tail^A$ is the tail operator in $G^A$.

We can then see a rotor step as two simultaneous free routings, in the sense of Sec.~\ref{sec:boundary}, in graphs $G^A$ and $G^V$. In what follows, we define a generalized rotor mechanisms multigraph, where the graph $G^A$ no longer required to be a collection of disjoint directed cycles.

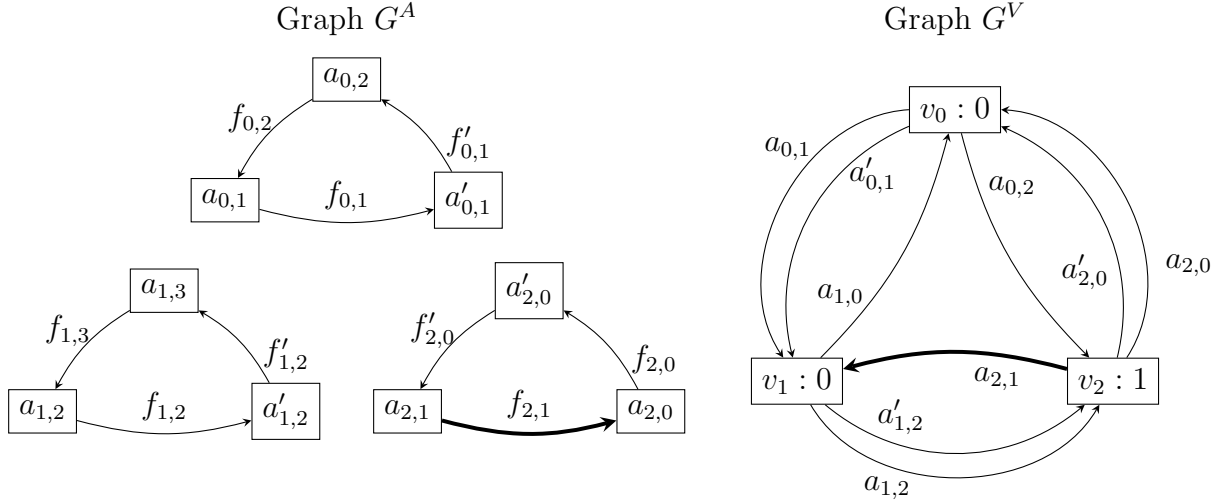
\begin{figure}[!htbp]
   \centering
\begin{tikzpicture}[->, >=stealth, scale=0.8]

    \node at (1,4.5) {Graph $G^A$};

    \begin{scope}[shift={(1,1.5)}]
    \node[draw] (v12) at (-2,0) {$a_{0,1}$};
    \node[draw] (v12bis) at (2,0) {$a'_{0,1}$};
    \node[draw] (v13) at (0,2) {$a_{0,2}$};
    \draw[bend right=15] (v12) to node[pos=.5, above] {$f_{0,1}$} (v12bis);
    \draw[bend right=15] (v12bis) to node[pos=.3, right] {$f'_{0,1}$} (v13);
    \draw[bend right=15] (v13) to node[pos=.3, left] {$f_{0,2}$} (v12);
    \end{scope}

    \begin{scope}[shift={(-2,-2)}]
    \node[draw] (v12) at (-2,0) {$a_{1,2}$};
    \node[draw] (v12bis) at (2,0) {$a'_{1,2}$};
    \node[draw] (v13) at (0,2) {$a_{1,3}$};
    \draw[bend right=15] (v12) to node[pos=.5, above] {$f_{1,2}$} (v12bis);
    \draw[bend right=15] (v12bis) to node[pos=.3, right] {$f'_{1,2}$} (v13);
    \draw[bend right=15] (v13) to node[pos=.3, left] {$f_{1,3}$} (v12);
    \end{scope}

    \begin{scope}[shift={(4,-2)}]
    \node[draw] (v21) at (-2,0) {$a_{2,1}$};
    \node[draw] (v20) at (2,0) {$a_{2,0}$};
    \node[draw] (v20bis) at (0,2) {$a'_{2,0}$};
    \draw[line width=1.5pt, bend right=15] (v21) to node[pos=.5, above] {$f_{2,1}$} (v20);
    \draw[bend right=15] (v20) to node[pos=.3, right] {$f_{2,0}$} (v20bis);
    \draw[bend right=15] (v20bis) to node[pos=.3, left] {$f'_{2,0}$} (v21);
    \end{scope}

    \begin{scope}[shift={(11,0)}]
            
        \node at (0,4.5) {Graph $G^V$};

        \node[draw] (v2) at (330:3) {$v_2 : 1$};
        \node[draw] (v0) at (90:3) {$v_0 : 0$};
        \node[draw] (v1) at (210:3) {$v_1 : 0$};

        \draw[bend right=40] (v2) to node[pos=.3, left=.1mm] {$a'_{2,0}$} (v0);
        \draw[bend right=60] (v2) to node[pos=.3, right=.1cm] {$a_{2,0}$} (v0);
        \draw[bend right=15] (v0) to node[pos=.2, right] {$a_{0,2}$} (v2);
        \draw[bend right=60] (v0) to node[pos=.3, left] {$a_{0,1}$} (v1);
        \draw[bend right=40] (v0) to node[pos=.3, right] {$a'_{0,1}$} (v1);
        \draw[bend right=15] (v1) to node[pos=.3, left] {$a_{1,0}$} (v0);
        \draw[bend right=60] (v1) to node[pos=.3, below] {$a_{1,2}$} (v2);
        \draw[bend right=40] (v1) to node[pos=.3, above] {$a'_{1,2}$} (v2);
        \draw[line width=1.5pt, bend right=15] (v2) to node[pos=.3, below] {$a_{2,1}$} (v1);
        \end{scope}
    
    \end{tikzpicture}

        \caption{The rotor multigraph $G_1$ of Fig.~\ref{fig:exempleG1}
        defined as a generalized rotor mechanisms $(G^A, G^V)$. 
        In particular, the face in bold $f_{2,1}$ and the arc in bold $a_{2,1}$ are coupled.
        }
        \label{fig:rotorgraph_as_CRS}
\end{figure}

\subsection{Definition of Generalized Rotor Mechanisms multigraphs}

\label{sec:GRM}

We generalize the standard rotor mechanism, where $G^A$ consists of 
the union of directed cycles simply by allowing any multigraph on every $A^+(v)$ instead of a directed cycle.

More precisely, let $G^V=(V,A,\head^V,\tail^V)$ be a multigraph  (for particles). Choose for every $v \notin S$, where $S$ is the set of sinks in $G^V$, a multigraph $G^A(v) = (A^+(v), F(v), \allowbreak \head^A, \allowbreak \tail^A)$ where $F(v)$ is any abstract finite set, and $\head^A, \tail^A$ are defined from $F(v)$ to $A^+(v)$ without restriction.

Let $G^A$ be the union of the graphs $G^A(v)$,
and let $F$ be the union of all $F(v)$ for $v \in V \setminus S$, so that
$G^A=(A, F, \head^A, \tail^A)$ is a multigraph. The elements of $F$
are called {\bf faces}.

We call a couple of multigraphs $(G^A, G^V)$ built like this, a {\bf Generalized Rotor Mechanism} (GRM from now on). We denote as $\deg^V$ and $\deg^A$ the degree mappings in graphs $G^V$ and $G^A$ as defined in Sec.~\ref{sec:boundary}.
An example of GRM multigraph is given on Fig. \ref{fig:generalized_rotor}.

\subsection{Definition of linear rotor-routing in GRM multigraphs}
\label{sec:linear_rotor_routing}

Let  $(G^A, G^V)$ be a GRM multigraph as defined above with $G^A = (A,F,\head^A,\tail^A)$ and $G^V=(V,A,\head^V,\tail^V).$
Formal sums of faces are denoted $C_F$.
Define

$$\cL : C_F \rightarrow C_A \times C_V $$
by 
$$\cL = \partial^A \times (\partial^V \circ \tail^A).$$

Let $(r,\sigma) \in C_A \times C_V$. We define {\bf the linear rotor-routing along} $\phi \in C_F$ as the operation that transforms $(r,\sigma)$ into $(r,\sigma) + \cL(\phi)$, and $\phi$ is called the \textbf{routing vector} of the routing operation. Note that if $\phi$ is a single face, with $a'=\head^A(\phi)$ and $a=\tail^A(\phi)$, this transformation adds $a'-a$ to $r$, and adds $v' - v$ to $\sigma$, where $v'=\head^V(a)$ and $v=\tail^V(a)$.

We say that $(r,\sigma)$ and $(r',\sigma')$ are equivalent (modulo linear routing), denoted by 
$$(r,\sigma)~\lineareq{\cL}{*}~(r',\sigma'),$$ 
if $(r'-r, \sigma'-\sigma) \in Im(\cL)$, i.e. if there is a routing vector that transforms  $(r,\sigma)$ into $(r',\sigma')$. If we want to specify that the routing vector is $\phi$,
we write  $(r,\sigma) \lineareq{\cL}{\phi} (r',\sigma')$.
Note that the linear routing operation is purely algebraic and can be computed by forming the matrix of~$\cL$. An example of this matrix is given in Appendix \ref{sec:matrix}.

\subsection{Definition of legal rotor-routing in GRM multigraphs}
\label{sec:deflegalroutingGRM}

An elementary linear routing is the routing along a face $f \in C_F$. Such a linear routing is said {\bf legal} for $(r,\sigma) \in C_A \times C_V$ if $r_a \geq 1$ and $\sigma_v \geq 1$, where $a = \tail^A(f)$ and $v = \tail^V(a)$. The interpretation is that there is a real 'vertex particle' on $v$ in $\sigma$ and a real 'arc particle' on $a$ in $r$, as it is the case in standard rotor-routing: adding $\cL(f)$ consists of moving respectively these particles along $a$, from $v$ to $\head^V(a)$, and along $f$ from $a$ to $\head^A(f)$.

Consider the special case of $(r,\sigma)=(\rho,\sigma)$ where $\rho$ is a rotor configuration, viewed as a sum of arcs. In this case, for every
non-sink vertex $v$, there is exactly one arc $a \in A^+(v)$ such that $\rho_a > 0$, namely $\rho(v)$. If $\sigma_v >0$ and we want to apply a GRM legal routing in order to move the particle on $v$, the process is similar to the routing in a standard rotor multigraph, with the the following differences:
\begin{itemize}
    \item when we want to move the particle from $v$ along the arc $\rho(v) \in A^+(v)$, we may sometimes choose the update on $\rho(v)$. Choosing a face $f \in F(v)$  with $\tail^A(f)=\rho(v)$ transforms $\rho$ into $\rho + \head^A(f) - \rho(v)$ , simultaneously with
    the movement of the particle;
    \item it can also happen that there is no face $f$ with $\tail^A(f)=\rho(v)$. In this case, no legal routing is available: the particle cannot legally move anymore, because $\rho(v)$ acts like a sink in the set of arcs.
\end{itemize}

Some basic generalized rotor mechanisms  could be for instance, a rotor multigraph where every arc $a$ can be updated to the next arc $\theta(a)$, or to the previous arc $\theta^{-1}(a)$, for every routing along $a$, or
a rotor multigraph with an arc-sink for every vertex, e.g. $G^A(v)$ is a directed path for every non sink vertex instead of a cycle.

\begin{figure}[!htbp]
    \centering
    \begin{subfigure}[t]{0.45\textwidth}
        \centering
        \begin{tikzpicture}[->, -{Latex[length=2mm]}, node distance=2cm]
        
        \node[draw] (v) at (0, 3) {v};
        \node[draw] (u) at (0, 0) {u};
        \node[draw] (x) at (-3, 0) {x};
        \node[draw] (y) at (3, 0) {y};

        \draw (u) -- (x) node[midway, below] {$a_1$};
        \draw (u) -- (y) node[midway, below] {$a_2$};
        \draw[bend left] (u) to node[midway, left] {$a_3$} (v);
        \draw[bend left] (v) to node[midway, right] {$a_4$} (u);
        \draw[bend right] (v) -- (x) node[midway, left] {$a_5$};

        \draw[red, dashed, bend left] (-1.5, 1.5) to node[midway,  above right] {$f_{54}$} (0.5, 1.5) ;
        \draw[red, dashed, bend right=30] (0.5, 1.5) to[out=90, in=90] node[midway, above] {$f_{44}$} (1.8, 1.5) to[out=90, in=90]   (0.5, 1.5);
        \draw[red, dashed, bend right=30] (-1.5, 1.5) to[out=90, in=90] node[midway, below] {$f_{55}$} (-2.5, 2.5) to[out=90, in=90] (-1.5, 1.5);
        \draw[red, dashed, bend right=30] (-1.5, 0) to[out=90, in=90] (-1.5, 1) to[out=90, in=90]  node[midway, right] {$f_{11}$} (-1.5, 0);
        \draw[red, dashed, bend right=30] (-1.5, 0) to[out=-90, in=-90] node[midway, above] {$f_{12}$} (1.5, 0);
        \draw[red, dashed] (1.5, 0) to node[midway, right] {$f_{23}$} (-0.5, 1.1);

    \end{tikzpicture} 
    \caption{Generalized rotor mechanisms in one picture: the dashed red arcs represent the possible evolution of arcs when routed.}
    \end{subfigure}\hfill
    \begin{subfigure}[t]{0.45\textwidth}
        \centering
        \begin{tikzpicture}[->, >=latex, node distance=2cm]
        \draw[white] (-3.5,-1.2) rectangle (3.5,3.5);

        \node[draw] (v) at (0, 3) {v};
        \node[draw] (u) at (0, 0) {u};
        \node[draw] (x) at (-3, 0) {x};
        \node[draw] (y) at (3, 0) {y};

        \draw (u) -- (x) node[midway, below] {$a_1$};
        \draw (u) -- (y) node[midway, below] {$a_2$};
        \draw[bend left] (u) to node[midway, left] {$a_3$} (v);
        \draw[bend left] (v) to node[midway, right] {$a_4$} (u);
        \draw[bend right] (v) -- (x) node[midway, left] {$a_5$};

    \end{tikzpicture}
    \caption{Graph  $G^V=(V,A,\head^V,\tail^V)$, with $V=\{x,y,u,v\}$ and $A=\{a_1, a_2, a_3, a_4, a_5\}$}
    \end{subfigure}\hfill
    
    \vspace{3em}
    
    \begin{subfigure}[t]{0.45\textwidth}
        \centering
        \begin{tikzpicture}[->, >=stealth, node distance=2cm]
        \draw[white] (-0.5,-0.5) rectangle (4.5,0.5);

        \node[draw] (a) at (0, 0) {$a_1$};
        \node[draw] (b) at (2, 0) {$a_2$};
        \node[draw] (c) at (4, 0) {$a_3$};

        \draw (a) to[loop above] node[left] {$f_{11}$} (a);
        \draw (a) -- (b) node[midway, above] {$f_{12}$};
        \draw (b) -- (c) node[midway, above] {$f_{23}$};

        \end{tikzpicture}
        \caption{Graph $G^A(u)$ on $A^+(u)$ with face set $F(u)=\{f_{11}, f_{12}, f_{23}\}$}
    \end{subfigure}\hfill
    \begin{subfigure}[t]{0.45\textwidth}
        \centering
        \begin{tikzpicture}[->, >=stealth, node distance=2cm]
        \draw[white] (-0.5,-0.5) rectangle (4.5,0.5);

        \node[draw] (a) at (0, 0) {$a_5$};
        \node[draw] (b) at (2, 0) {$a_4$};

        \draw (a) to[loop above] node[left] {$f_{55}$} (a);
        \draw (a) -- (b) node[midway, above] {$f_{54}$};
        \draw (b) to[loop above] node[left] {$f_{44}$} (b);

        \end{tikzpicture}
        \caption{Graph $G^A(v)$ on $A^+(v)$ with face set $F(v)=\{f_{44}, f_{54}, f_{55}\}$}
    \end{subfigure}\hfill
    \caption{a GRM multigraph $G$. In Fig.~(a): representation in one picture of the GRM, with arcs in full black and faces in dashed red. The other figures corresponds to graphs $G^V$ (Fig.~(b)) and $G^A$ respectively, the last graph being split into graph $G^A(u)$ (Fig.~(c)) and $G^A(v)$ (Fig.~(d)). }
    \label{fig:generalized_rotor}
\end{figure}
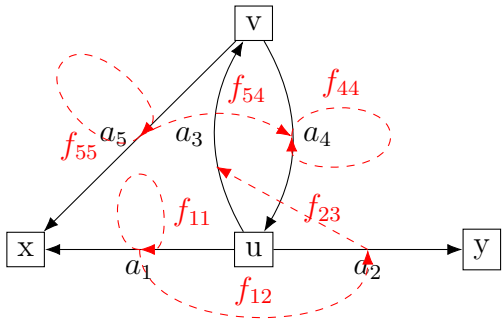
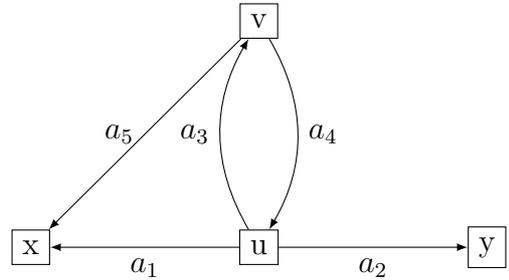
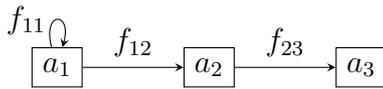
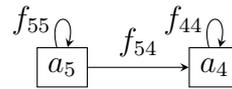

It can be noted that linear routing preserves degrees. In other words, if $(r,\sigma) \lineareq{\cL}{*} (r',\sigma')$, then $\deg^A(r) = \deg^A(r')$ and $\deg^V(\sigma) = \deg^V(\sigma')$. This means that the total algebraic number  of particles should be the same in every connected component of $G^V$, and that the sum of arcs should be the same in every mechanism $G^A(v)$ (or in each of the weakly connected components of the mechanism, if not connected).

A {\bf routing sequence} for a routing vector $\phi \in C_F$ is a finite sequence of faces $f_0, f_1, \dots, f_k$  such that $\phi = \sum_i f_i$. This sequence is legal for $(r,\sigma)$ if, when applying in order elementary routing steps along $f_0, f_1, \dots, f_k$, every step is legal.

For legal routings, we use notation akin to the case of free routing, namely $$(r,\sigma) \legalseq{\cL}{*} (r',\sigma')$$ if there is a legal routing sequence from 
$(r,\sigma)$ to $(r',\sigma')$, and 
$(r,\sigma) \legalseq{\cL}{\phi} (r',\sigma')$ if there is a legal sequence with routing vector exactly $\phi$.

\subsection{Cyclic GRM multigraphs}
\label{sec:cyclic_GRM}

If $G=(V,A,\head,\tail,\theta)$ is a rotor multigraph, we can associate to $G$ a GRM multigraph $(G^A,G^V)$ defined as follows:
\begin{itemize}
    \item $G^V=(V,A,\head,\tail)$
    \item $G^A=(A,F,\head^A,\tail^A)$ where $F=\{(a,\theta(a)) : a \in A\}$, and $\head^A((a,\theta(a)) = \theta(a)$, 
    $\tail^A((a,\theta(a)) = a$ for all $a \in A$.
\end{itemize}

A GRM multigraph built from a rotor multigraph like this is called a {\bf cyclic GRM multigraph}, since $G^A(v)$ is a directed cycle on the 'vertex' set $A^+(v)$ for every $v \in V \setminus S$ (notation $S$ will always refer to the sinks of $G^V$).
We say that a cyclic GRM multigraph is stopping, or strongly connected,
or any other property of rotor multigraphs, if the corresponding rotor multigraph satisfies the property.

In this context, a {\bf rotor configuration} is an element
$\rho \in C_A$ such that for any $v \in V \setminus S$,
there is unique $a \in A^+(v)$ such that $\rho_a = 1$,
and $\rho_a = 0$ for the other arcs in $A^+(v)$. 

In a cyclic GRM multigraph, if $\rho$ is a rotor configuration and 
$\sigma_v>0$, the linear rotor-routing along the face $(\rho(v), \theta(\rho(v))$ matches the standard definition of rotor-routing. Hence, \emph{the notion of legal routing in cyclic GRM multigraphs completely emulates standard rotor-routing as defined in Section~\ref{sec:standard}.}
In the following, we speak of cyclic GRM multigraphs instead of rotor multigraphs to avoid clashes in definitions of rotor-routing, as the settings of cyclic GRM multigraphs and rotor multigraphs are essentially the same, but the definitions of rotor-routing differ.

\subsection{Computing routing vectors}\label{sec:comput_routing_vect}

Given a routing vector $\phi \in C_F$, the matrix of $\cL$ allows us to compute the linear rotor-routing along $\phi$ in polynomial time (see Appendix~\ref{sec:matrix} for a detailed example). Conversely, if  $r,r' \in C_A$ and $\sigma, \sigma' \in C_V$ are given as input, one may seek to decide if a routing vector $\phi$ such that $(r',\sigma') = (r,\sigma) + \cL(\phi)$ exists, and if so, compute such a routing vector. As mentioned before, a necessary condition is that $\deg^V(\sigma) = \deg^V(\sigma')$ and $\deg^A(r) = \deg^A(r')$.

In what follows, the term 'polynomial' means polynomial in the size of $G$, together with $r, r' \in C_A$ and $\sigma, \sigma' \in C_V$ (encoded in binary).

\begin{proposition}
\thlabel{lem:routing_vector}
Let $r,r' \in C_A$ and $\sigma, \sigma' \in C_V$. There is a polynomial time algorithm that decides whether $(r,\sigma) \lineareq{\cL}{*} (r',\sigma')$, and if so, returns a routing vector $\phi \in C_F$ such that $(r',\sigma') = (r,\sigma) + \cL(\phi)$ and $\phi$ has polynomial size. 
\end{proposition}

\begin{proof}
 This amounts to computing an integral solution of a system of linear diophantine equations, if it exists. This can be achieved by using the Smith decomposition of the matrix of $\cL$ (see Appendix~\ref{sec:smith} for details on how to use the decomposition in order to solve the system). It remains to compute effectively
 the decomposition of the matrix, which can be done in polynomial
 time~\cite{kannan1979polynomial}.
\end{proof}

\section{Legal rotor-routings in GRM multigraphs with specified routing vector}
\label{sec:legal_grm}

We explore the concept of legal rotor-routing within GRM multigraphs, as defined in the preceding section, and characterize the subset of linear routings that correspond to legal sequences, following an approach similar to that in Section \ref{sec:legal_routing}, but specifically adapted for linear rotor-routing in GRM multigraphs. 

In a GRM multigraph $(G^A, G^V)$, every legal rotor-routing induces legal (free) routings in both $G^A$ and $G^V$, but the converse generally fails: compatibility requires a specific interaction between the two components.
In particular, certain special faces that appear in the free routing in $G^A$, called terminal elements, must correctly guide the free routing in $G^V$.
This leads to a complete characterization of legal GRM routings with a prescribed routing vector~$\phi$ (Section~\ref{sec:legal_routing_gen_rotor}).

In the cyclic case, terminal elements in $G^A$ admit an explicit description in terms of transitory arcs, yielding a simplified legality criterion for cyclic GRMs.
This provides in turn a clean legality criterion for standard rotor-routing (Section~\ref{sec:legal_routings_rotor}).

The main results of this section are \thref{thm:routing_rotor_mecha} and \thref{thm:charac_legal_rotor}.

\subsection{General case}
\label{sec:legal_routing_gen_rotor}

We consider a GRM multigraph $(G^A, G^V)$ as defined in Sec.~\ref{sec:GRM}, with the same notations.

With our definition (see Sec. \ref{sec:deflegalroutingGRM}), a legal routing sequence in a GRM moves legally and simultaneously 'arc particles' in $G^A$ and 'vertex particles' in $G^V$. It is easy to check with our definitions that if $(f_i)_{0 \leq i \leq k-1}$, with $f_i \in F$ for every $0 \leq i \leq k-1$, is a legal routing sequence for $(r, \sigma)$ in the GRM $(G^A, G^V)$ (in the sense of legal routing in GRMs), then
\begin{itemize}
    \item $(f_i)_{0 \leq i \leq k-1}$ is formally also a legal routing sequence for $r$ in $G^A$ (in the sense of legal free routing, see Sec.~\ref{sec:legal_routing});
    \item the coupled routing of 'vertex particles', given by the sequence $(\tail^A(f_i))_{0 \leq i \leq k-1}$ (where $\tail^A(f_i) \in A$ for all $i$), is a  legal routing sequence for $\sigma$ in $G^V$ (also in the sense of free legal routing).
\end{itemize}
In particular, 
\[ (r,\sigma) \legalseq{\cL}{\phi} (r',\sigma') \ \Longrightarrow \ r \legalseq{\partial^A}{\phi} r' \text{ and } \ \sigma \legalseq{\partial^V}{\tail^A(\phi)} \sigma'.\]
 
The opposite direction is more tricky:  it is not enough to 
suppose that there is a legal sequence in $G^A$ for $r$ with routing vector $\phi$, and a legal sequence in $G^V$ for $\sigma$ with vector $\tail^A(\phi)$, to obtain a legal sequence for $(r,\sigma)$ in the GRM. A counterexample is given in Figure~\ref{fig:unlegal_routing_coupled}.

\begin{figure}[!htbp]
    \centering
    \begin{subfigure}{0.45\textwidth}
        \centering
        \begin{tikzpicture}[->, -{Latex[length=2mm]}, node distance=2cm]

            \node[draw] (u) at (0, 0) {u};
            \node[draw] (v) at (0, 2) {v};
            \draw[->] (u) -- node[midway, right] {$a_1$} (v);
            \draw[->, loop left, out=150, in=210, looseness=8] (u) to node[below] {$a_2$} (u);
    
            \draw[red, dashed, bend right] (0, 1) to node[midway,  above] {$f_{12}$}  (-0.5, 0.3);
            \draw[red, dashed, bend right=30] (-0.9, 0) to[out=90, in=90] node[midway, above] {$f_{22}$} (-2.5, 0) to[out=90, in=90]    (-0.9, 0);
        \end{tikzpicture}
        \caption{ Names of vertices, arcs and faces.}
    \end{subfigure}
    \begin{subfigure}[t]{0.45\textwidth}
        \centering
        \begin{tikzpicture}[->, -{Latex[length=2mm]}, node distance=2cm]
        
        \node[draw] (u) at (0, 0) {$1 \rightarrow 0$};
        \node[draw] (v) at (0, 2) {$0 \rightarrow 1$};


        \draw[->] (u) -- node[midway, right] {$1 \rightarrow 0$} (v);
        \draw[->, loop left, out=150, in=210, looseness=8] (u) to node[left] {$0 \rightarrow 1$} (u);
        
        \end{tikzpicture}
        \vspace{-1.7em}
        \caption{Configurations $(r, \sigma)$ and $(r', \sigma')$.}
    \end{subfigure}
    \caption{(a) A GRM $(G^A,G^V)$ with two vertices $u,v$ with $v$ a sink, two arcs $a_1,a_2$ and two faces $f_{12},f_{22}$ both in $F(u)$, with $f_{12}$ from $a_1$ to $a_2$, and $f_{22}$ a loop on $a_2$. (b) Starting and ending configurations $(r,\sigma)$ and $(r',\sigma')$. \\
    One can check that the routing vector $f_{12} + f_{22}$ satisfies
    $\mbox{$(r,\sigma)\lineareq{\cL}{f_{12} + f_{22}} (r',\sigma')$}$. In $G^A$, we have $r \legalseq{\partial^A}{f_{12} + f_{22}} r'$, and the only legal sequence for this is $(f_{12},f_{22})$. 
    We note that $\tail^A(f_{12} + f_{22}) = a_1 + a_2$, and that in $G^V$ we have $\sigma \legalseq{\partial^V}{a_1 + a_2} \sigma'$ but that the only legal sequence is $(a_2,a_1)$, which does not correspond to $(f_{12},f_{22})$. As a consequence,
    there is no legal sequence from $(r, \sigma)$ to $(r', \sigma')$ with routing vector $f_{12} + f_{22}$ in the sense of GRMs. }
    \label{fig:unlegal_routing_coupled}
\end{figure}

To ensure compatibility between legal sequences in $G^A$ and $G^V$, we must note an important necessary condition. Suppose once again that $(f_i)_{0 \leq i \leq k-1}$ is a legal routing sequence for $(r, \sigma)$ in $(G^A, G^V)$ with vector $\phi$. Let $a_i = \tail^A(f_i)$ and $\alpha=\tail^A(\phi)$;
then $(a_i)_{0 \leq i \leq k-1}$ a legal sequence in $G^V$ for $\sigma$ as we just established. 
For any transitory vertex $v \in \trans(\alpha, \sigma')$ (in the sense of legal free routing in $G^V$, see  Section~\ref{sec:transitory}), let $i_v$ be the last index such that $v$ is the tail of $v = \tail^V(a_i)$, i.e. the last time that a particle moves out of $v$ during the sequence. 
By \thref{lemma:routing2tree}, the set of the arcs $a_{i_v}$ for all transitory vertices in 
$\trans(\alpha, \sigma')$ is a guiding forest for $\sigma \legalseq{\partial^V}{\alpha} \sigma'$. At the same time, for such a transitory $v$, the face $f_{i_v}$ must be the last element of $F(v)$ appearing in the sequence $(f_i)$. 

Let us define a {\bf terminal element} of the free routing $r \legalseq{\partial^A}{\phi} r'$ in $G^A$ as a face $f \in \phi$ such that there is a legal sequence ending by $f$ for $r \legalseq{\partial^A}{\phi} r'$. We just established that for every transitory vertex $v \in \trans(\alpha, \sigma')$, the face $f_{i_v}$ is the last element of $F(v)$ appearing in the sequence $(f_i)$ ; since $F(v)$ is a connected component of $G^A$, we could move $f_{i_v}$ to the end of the sequence, and the routing would still be legal. Hence, $f_{i_v}$ is a terminal element of $r \legalseq{\partial^A}{\phi} r'$. This implies in particular the following: {\it if $\cal T$ is the set of terminal elements for $r \legalseq{\partial^A}{\phi} r'$, then the set of arcs $\{\tail^A(f), f \in \cal T \}$ is guiding for $\sigma \lineareq{\partial^V}{\alpha} \sigma'$.}

This necessary condition turns out to be sufficient, and gives a criteria for the existence of a legal routing in a GRM with a specified routing vector. 

\begin{theorem} \thlabel{thm:routing_rotor_mecha}
    Consider a GRM multigraph $(G^A,G^V)$. Let $r,r' \in C_A$ and $\sigma,\sigma' \in C_V$. Let 
    $\phi \in C^+_F$ with $\alpha = \tail^A(\phi)$.
    Then $(r,\sigma) \legalseq{\cL}{\phi} (r',\sigma')$ if and only if:
    \begin{enumerate}[(i)]
        \item $r \legalseq{\partial^A}{\phi} r'$ in the sense of free routing;
        \item $\sigma \legalseq{\partial^V}{\alpha} \sigma'$
        in the sense of free routing;
        \item if $\cal T$ is the set of terminal elements of $r \legalseq{\partial^A}{\phi} r'$, then $\{\tail^A(f) : f \in \cal T \}$ is guiding for $\sigma \lineareq{\partial^V}{\alpha} \sigma'$.
    \end{enumerate}
\end{theorem}

Note that $(i)$ and $(ii)$ are characterized by~\thref{prop:CNS_routing_vector}. Before proving the result, let us note that $(iii)$ can also be checked with repeated use of~\thref{prop:CNS_routing_vector}. Indeed, if $r \legalseq{\partial^A}{\phi} r'$, let $f \in \phi$ and consider two cases:
\begin{itemize}
    \item if $f$ is a loop, then $f$ is a terminal element of $r \legalseq{\partial^A}{\phi} r'$ if and only if $r'_f \geq 1$;
    \item otherwise, $f$ is a terminal element of $r \legalseq{\partial^A}{\phi} r'$ if and only if $r \legalseq{\partial^A}{\phi-f} r'-\partial^A(f)$.
\end{itemize}

It follows that these conditions can be checked in polynomial time.

\begin{proof}[Proof of \thref{thm:routing_rotor_mecha}]
    Let us make the convention that a sequence of elements of $F$ (resp. of $A$) is said legal for $r \in C_A$, (resp. for $\sigma \in C_V$) if it is
    legal in the sense of free routing in $G^A$ (resp. in $G^V)$, as per Section~\ref{sec:legal_routing};
    and that a sequence of elements of $F$ is said legal for $(r,\sigma)$, 
    if it is legal in the sense of linear rotor-routing in $(G^A,G^V)$
    as described in Section~\ref{sec:linear_rotor_routing}. 
    
    Let us introduce some notation for the proof. Since
    $G^A$ is the union of the weakly connected components $G^A(v)$
    for $v \in V \setminus S$, we decompose any routing vector $\phi$ in $G^A$ as  $\phi = \sum_{v \in V \setminus S} \phi_{|v}$,
    where $\phi_{|v} \in C_{F(v)}$ . Define also
    the restrictions of $r, r' \in C_A$ to $A^+(v)$ by
    $r_{|v}$, $r'_{|v}$ so that $r = \sum_v r_{|v}$ and $r' = \sum_v r'_{|v}$.
    If $r \lineareq{\partial^A}{\phi} r'$
    then for every $v$ we have  $r_{|v} \lineareq{\partial^A}{\phi_{|v}} r'_{|v}$.
    For a routing sequence $s=(f_i)_i$, we can as well decompose
    the sequence as subsequences $s_{|v}$ for each non sink vertex $v$,
    where faces appearing in $s_{|v}$ belong to $F(v)$; then
    clearly $s$ is legal for $r$ if and only if all $s_{|v}$ are legal for $r_{|v}$.

    We now show the necessity of  conditions $(i)$-$(iii)$: suppose 
    the existence of the legal routing sequence $(f_i)_{0 \leq i \leq k-1}$ for $(r, \sigma)$. We denote $a_i = \tail^A(f_i)$ for $0 \leq i \leq k-1$.
    By definition of legality in GRM multigraphs,
    $(f_i)_{0 \leq i \leq k-1}$ 
    and $(a_i)_{0 \leq i \leq k-1}$ are respectively legal routing sequences for $r \legalseq{\partial^A}{\phi} r'$  and $\sigma \legalseq{\partial^V}{\alpha} \sigma'$, which proves $(i)$ and $(ii)$. Then, by~\thref{lemma:routing2tree}, for any
    transitory vertex $v$ of $\sigma \lineareq{\partial^V}{\alpha} \sigma'$,     
    if $t(v)$ denotes the last arc with tail $v$ appearing in the sequence $(a_i)_{0 \leq i \leq k-1}$, then
    the set of all $t(v)$
    forms a guiding forest for this routing.
    For such a transitory vertex $v$, 
    let $i(v)$ be the last index of $t(v)$ in that sequence. 
    Then $f_{i(v)}$  is the last element
    of $F(v)$ appearing in $(f_i)_i$,
    so it is a terminal element of $r_{|v} \legalseq{\partial^A}{\phi_{|v}} r'_{|v}$.
    Since $G^A(v)$ is disjoint from the rest of the graph, we could 
    move routings in $G^A(v)$ to the end of the sequence, hence $F(v)$
    is also
    a terminal element of $r \legalseq{\partial^A}{\phi} r'$, so that
    $(iii)$ is true.

    Conversely, suppose that conditions $(i)$-$(iii)$ are satisfied. 
    By $(iii)$, we can suppose for every $v \in \trans(\alpha,\sigma')$, that there is $f(v) \in F(v)$, so that if  $t(v)=\tail^A(f(v))$ :
    \begin{itemize}
        \item $f(v)$ is a terminal element of $r \legalseq{\partial^A}{\phi} r'$, and
        \item $\{t(v)\}_v$ is a guiding forest for $\sigma \lineareq{\partial^V}{\alpha} \sigma'$.
    \end{itemize}

    For every  vertex $v \in V \setminus S$, we can consider a (possibly empty) legal routing sequence $s(v) = (f^v_0, \dots, f^v_{k_v })$ for 
    $r_{|v} \legalseq{\partial^A}{\phi_{|v}} r'_{|v}$, such that, if $v \in \trans(\alpha, \sigma')$ then $f^v_{k_v } = f(v)$.
    This sequence describes an ordering of all routing steps
    that must be done in $F(v)$ for every $v$.
    Our strategy is to construct recursively a legal routing sequence 
     $s = (f_i)_{0\leq i \leq k}$ for 
     $(r,\sigma) \legalseq{\cL}{\phi} (r',\sigma')$ that coincides with $s(v)$ for
     every active vertex $v$, i.e. $s_{|v} = s(v)$. This will ensure that
     the routing of particles will be guided by $t$ and
     that the sequence can be extended.

    Suppose that $\ell \geq 0$ and let $s^\ell = (f_i)_{0\leq i \leq \ell-1}$ is a  routing sequence with routing vector $\phi^\ell$, such that:
    \begin{itemize}
        \item $s^\ell$ is legal for $(r,\sigma) \legalseq{\cL}{\phi^\ell} (r^\ell,\sigma^\ell)$ ;
        \item $\phi^\ell \leq \phi$ ;
        \item $s^\ell_{|v}$ is a prefix of $s(v)$ for every vertex $v$.
    \end{itemize}

    Suppose that $\phi^\ell \neq \phi$, and let $\alpha^\ell = \tail^A(\phi^\ell)$. 
    By construction of $s(v)$,
    the sequence $(a_i)_{0\leq i \leq \ell-1}$ with $a_i = \tail^A(f_i)$ is  guided by $(\sigma, \alpha, t)$ and is $\alpha$-bounded. Then, by \thref{lemma:CNS_routing2}, there is an active vertex $v$ of $\alpha-\alpha^\ell$, such that $\sigma^\ell_v \geq 1$. Then we choose the face $f \in F(v)$ with $f \in \phi-\phi^\ell$ such that $ s^\ell_{|v}$ appended with $f$  is still a prefix of $s(v)$. Then, the extended sequence $s^{\ell + 1}= (f_0, \dots, f_{\ell -1}, f)$ satisfies 
        \begin{itemize}
        \item $s^{\ell+1}$ is legal for $(r,\sigma)$ since by construction $\sigma^\ell_v \geq 1$ and $s^{\ell+1}_{|v}$
        is a prefix of a legal sequence for $r_{|v} \legalseq{\partial^A}{\phi_{|v}} r'_{|v}$;
        \item $\phi^{\ell+1} \leq \phi$ ;
        \item $s^{\ell+1}_{|v}$ is a prefix of $s(v)$ for every vertex $v$ by construction also.
    \end{itemize}

    Finally, starting from the empty sequence, we obtain recursively a sequence $s^\ell$ with $\phi^\ell=\phi$
    satisfying the three properties above, hence a legal sequence
    for $(r,\sigma) \legalseq{\cL}{\phi} (r',\sigma')$.

\end{proof}

This theorem is illustrated in an example of the GRM multigraph in Figure~\ref{fig:thm_legality_generalized_rotor}.  An example where conditions $(i)$ and $(ii)$ of the theorem are satisfied, but not condition $(iii)$ is shown in Figure~\ref{fig:unlegal_routing_coupled}.

\begin{figure}[!htbp]
    \centering
    \begin{subfigure}[t]{0.45\textwidth}
        \centering
   \begin{tikzpicture}[->, -{Latex[length=2mm]}, node distance=2cm]
        
        \node[draw] (v) at (0, 3) {v};
        \node[draw] (u) at (0, 0) {u};
        \node[draw] (x) at (-3, 0) {x};
        \node[draw] (y) at (3, 0) {y};

        \draw (u) -- (x) node[midway, below] {$a_1$};
        \draw (u) -- (y) node[midway, below] {$a_2$};
        \draw[bend left] (u) to node[midway, left] {$a_3$} (v);
        \draw[bend left] (v) to node[midway, right] {$a_4$} (u);
        \draw[bend right] (v) -- (x) node[midway, left] {$a_5$};

        \draw[red, dashed, bend left] (-1.5, 1.5) to node[midway,  above] {$f_{54}$} (0.5, 1.5) ;
        \draw[red, dashed, bend right=30] (0.5, 1.5) to[out=90, in=90] node[midway, above] {$f_{44}$} (1.8, 1.5) to[out=90, in=90]   (0.5, 1.5);
        \draw[red, dashed, bend right=30] (-1.5, 1.5) to[out=90, in=90] node[midway, below] {$f_{55}$} (-2.5, 2.5) to[out=90, in=90] (-1.5, 1.5);
        \draw[red, dashed, bend right=30] (-1.5, 0) to[out=90, in=90] (-1.5, 1) to[out=90, in=90]  node[midway, right] {$f_{11}$} (-1.5, 0);
        \draw[red, dashed, bend right=30] (-1.5, 0) to[out=-90, in=-90] node[midway, above] {$f_{12}$} (1.5, 0);
        \draw[red, dashed] (1.5, 0) to node[midway, right] {$f_{23}$} (-0.5, 1.1);

    \end{tikzpicture} 
        
        \caption{Names of vertices, arcs and faces.}
    \end{subfigure}
    \begin{subfigure}[t]{0.45\textwidth}
        \centering
        \begin{tikzpicture}[->, -{Latex[length=2mm]}, node distance=2cm]
        \draw[white] (-3.5,-1.2) rectangle (2,3.5);
        
        \node[draw] (v) at (0, 3) {$3 \rightarrow 0$};
        \node[draw] (u) at (0, 0) {$3 \rightarrow 0$};
        \node[draw] (x) at (-3, 0) {$0 \rightarrow 4$};
        \node[draw] (y) at (3, 0) {$0 \rightarrow 1$};

        \draw (u) -- (x) node[midway, below] {$1 \rightarrow 0$};
        \draw (u) -- (y) node[midway, below] {$0 \rightarrow 0$};
        \draw[bend left] (u) to node[midway, below] {$0 \rightarrow 1$} (v);
        \draw[bend left] (v) to node[midway, right] {$0 \rightarrow 1$} (u);
        \draw[bend right] (v) -- (x) node[midway, left] {$1 \rightarrow 0$};

        \end{tikzpicture}
        \vspace{-1em}
        \caption{Configurations $(r, \sigma)$ and $(r', \sigma')$ on arcs and vertices.}
    \end{subfigure}\hfill
    \caption{(a) Same GRM multigraph as in Fig.~\ref{fig:generalized_rotor}. (b) Starting and ending configurations $(r,\sigma)$ and $(r',\sigma')$ of arcs and vertices. 
    Consider $\phi = f_{11}+f_{12}+f_{23}+f_{44} + f_{54}+f_{55}$ and $\alpha = \tail^A(\phi) = 2 a_1 + a_2 + a_4 + 2 a_5$, then $(r',\sigma') = (r,\sigma) + \cL(\phi)$. Clearly \mbox{$\sigma  \lineareq{\partial^V}{\alpha} \sigma'$} and $r \lineareq{\partial^A}{\phi} r'$ admit legal routing sequences with routing vectors $\alpha$ and $\phi$ respectively. The set of terminal elements of $r  \legalseq{\partial^V}{\phi} r'$ is $\{f_{23}, f_{44}\}$. Transitory vertices for the routing $\sigma \lineareq{\partial^V}{\alpha} \sigma'$ are $u$ and $v$ and $\tail^A(f_{23}) + \tail^A(f_{44}) = a_2 + a_4$ which is guiding for this routing. Hence there is a legal routing sequence with routing vector $\phi$. }
    \label{fig:thm_legality_generalized_rotor}
\end{figure}

\subsection{Cyclic case}
\label{sec:legal_routings_rotor}

In this section, we characterize the routing vectors that admit a legal routing sequence in a cyclic GRM multigraph, which provides a condition easier to check than the general one in \thref{thm:routing_rotor_mecha}. 
We recall that a cyclic GRM multigraph is equivalent to a standard rotor-routing graph, except that  it allows routing with any arc configuration (see Section~\ref{sec:cyclic_GRM}).

We begin with a useful lemma that characterizes terminal elements in the legal free routing of arcs in the multigraph $G^A$. 

\begin{lemma}
    \thlabel{lem:terminal_condition}
    Suppose that $(G^A, G^V)$ is a cyclic GRM multigraph.
    Consider a legal free routing $r \legalseq{\partial^A}{\phi} r'$ in $G^A$. Let $f_0 \in \phi$,
    and let $a_0=\tail^A(f_0)$ so that $\head^A(f_0)=\theta(a_0)$.
    Then $f_0$ is a terminal element of this routing if and only if $\theta(a_0) \not\in \trans(\phi,r')$, i.e. if $r'_{\theta(a_0)} \geq 1$, or
    there is no $f \in F$ with $\phi_{f} \geq 1$ and $\tail^A(f) = \theta(a_0)$.
\end{lemma}

\begin{proof}
    If $f_0$ is a terminal element, then $r \legalseq{\partial^A}{\phi-f_0} r''$, where $ r'' = r' - \partial^A(f_0)$. 
    If $a_0=\theta(a_0)$, which means that $\tail^V(a_0)$ has outdegree $1$ in
    $G^V$, then $r_{a_0} = r'_{a_0} = r''_{a_0}$; since ${a_0}$ is active in $\phi$, we must have from the beginning $r_{a_0} \geq 1$ so that $r'_{a_0} \geq 1$, therefore $\theta({a_0})={a_0} \notin \trans(\phi,r')$.
    If ${a_0} \neq \theta({a_0})$, and $\theta({a_0})$ is active in 
    $\phi$, then it is also
    active in $\phi-f_0$, then $r''_{\theta({a_0})} \geq 0$,
    so $r'_{\theta({a_0})} \geq 1$, and $\theta({a_0}) \notin \trans(\phi,r')$.

    Conversely, suppose that $\theta({a_0}) \notin \trans(\phi, r')$. 
    Since $f_0 \in \phi$, ${a_0}$ is active in $\phi$. If ${a_0} \neq \theta({a_0})$, we have
    $r'_{a_0} \geq 0$ so $r''_{a_0} \geq 1$. If ${a_0}=\theta({a_0})$,
    then $r_{a_0}=r'_{a_0}=r''_{a_0} \geq 1$. In all cases, 
    $r''_{a_0} \geq 1$. So if we prove that
    $r \lineareq{\partial^A}{\phi - f_0} r''$ admits a legal routing
    sequence, it will be legal to add $f_0$ at the end,
    proving that $f_0$ is a terminal element. To see the existence of
    a legal sequence with routing vector $\phi-f_0$, we check the conditions of \thref{prop:CNS_routing_vector}.

    First, for any active arc $a$ in $\phi$, we have $r'_a \geq 0$, since $r \legalseq{\partial^A}{\phi} r'$. The active arcs in $\phi - f_0$ are those in $\phi$, possibly excluding arc ${a_0}$. Meanwhile, the only possible arc $a$ such that $r''_a < r'_a$ is $\theta({a_0})$, and $r'_{\theta({a_0})} -1 \leq r''_{\theta({a_0})} \leq r'_{\theta({a_0})}$. If $\theta({a_0})$ is active in $\phi - f_0$, then it is also active in $\phi$. However, we assumed that $\theta({a_0}) \notin \trans(\phi, r')$, which implies $r'_{\theta({a_0})} \geq 1$, and therefore $r''_{\theta({a_0})} \geq 0$. Hence, every arc $a$ active in $\phi - f_0$ satisfies $r''_{a} \geq 0$.

    Second, we check the existence of a guiding forest for $r \lineareq{\partial^A}{\phi - f_0} r''$. Let $v_0 = \tail^V({a_0})$.
    Since $r \lineareq{\partial^A}{\phi - f_0} r''$ and 
     $r \lineareq{\partial^A}{\phi} r'$ only differ in $G^A(v_0)$,
     the only remaining task is to verify the existence of a guiding forest for $r_{|v_0} \lineareq{\partial^A}{\phi_{|v_0} - f_0} r''_{|v_0}$ (we recall that notation $r_{|v_0}$ means the part of $r$ with arcs in $A^+(v_0)$).
     
     If  $F(v_0) \cap \trans(\phi-f_0, r'') = \emptyset$, then there
     is nothing to check, the empty set is a guiding forest for 
     $r_{|v_0} \lineareq{\partial^A}{\phi_{|v_0} - f_0} r''_{|v_0}$.
    Otherwise let $a \in \trans(\phi_{|v_0}-f_0, r''_{|v_0})$ and suppose that the
    set of faces which are elements of $\phi_{|v_0} - f_0$ is not guiding.
    Since $G^A(v_0)$ is a directed cycle, we also
    have $\theta(a) \in \trans(\phi_{|v_0}-f_0, r''_{|v_0})$, $\theta^2(a) \in \trans(\phi_{|v_0}-f_0, r''_{|v_0})$ and so on, so that  $A^+(v_0) \subset \trans(\phi_{|v_0}-f_0, r''_{|v_0})$,
    and in particular ${a_0} \in \trans(\phi_{|v_0}-f_0, r''_{|v_0})$. This contradicts the fact that $r''_{a_0} \geq 1$, which was proved above.

    We checked the two conditions of \thref{prop:CNS_routing_vector}.
\end{proof}

Consider now a free linear routing $r \lineareq{\partial^A}{\phi} r'$.
If $r \legalseq{\partial^A}{\phi} r'$, by \thref{lem:terminal_condition}, the terminal elements of this routing
are exactly the faces $f \in \phi$ such that $\tail^A(f)$ belongs to the following set ${\cal T}_A$:
\[{\cal T}_A = \{a \in \alpha: r'_{\theta(a)} > 0 \} \cup \{a \in \alpha : \theta(a) \notin \alpha \},\]
where $\alpha = \tail^A(\phi)$. Note that terminal elements are entirely characterized by the values $\phi$ and $r'$.

Based on this characterization, we can derive the following result as an adaptation of \thref{thm:routing_rotor_mecha} for GRM multigraphs. Unlike \thref{thm:routing_rotor_mecha}, which relied on applying the characterization from \thref{prop:CNS_routing_vector} several times as a subroutine, this result is self-contained.

\begin{proposition} \thlabel{thm:charac_legal_rotor}
    Suppose that $(G^A, G^V)$ is a cyclic GRM multigraph and $(r,\sigma) \lineareq{\cL}{\phi} (r',\sigma')$ for some $\phi \in C_F^+$.  Let $\alpha = \tail^A(\phi)$ and 
    $${\cal T}_A = \{a \in \alpha: r'_{\theta(a)} > 0 \} \cup \{a \in \alpha : \theta(a) \notin \alpha \} .$$ 
    Then $(r,\sigma) \legalseq{\cL}{\phi} (r',\sigma')$ if and only if 
    \begin{enumerate}[(i)]
        \item $\sigma'_v \geq 0$ for every vertex $v$ active in $\alpha$;
        \item $ r'_a \geq 0$ for every arc $a$ active in $\phi$;
        \item ${\cal T}_A \cap A^+(v) \neq \emptyset$ for every vertex $v$ active in $\alpha$;
        \item ${\cal T}_A$ is guiding for $\sigma \lineareq{\partial^V}{\alpha} \sigma'$.
    \end{enumerate}
    
\end{proposition}

\begin{proof}
    By~\thref{thm:routing_rotor_mecha}, there is a legal routing sequence from $(r,\sigma)$
    to $(r',\sigma')$ with routing vector $\phi$ if and only if 
    \begin{enumerate}[(1)]
        \item $\sigma \legalseq{\partial^V}{\alpha} \sigma'$
        \item $r \legalseq{\partial^A}{\phi} r'$ 
        \item $\{\tail^A(f) : f \in {\cal T}_F \}$ is guiding for $\sigma \legalseq{\partial^V}{\alpha} \sigma'$ where ${\cal T}_F$ is the set of terminal elements of $r \legalseq{\partial^A}{\phi} r'$.
    \end{enumerate}

    First suppose that (i), (ii), (iii) and (iv) are satisfied. By~\thref{prop:CNS_routing_vector} conditions (i) and (iv) of the proposition imply condition (1), since the set of arcs ${\cal T}_A$ is contained in $\alpha$.

    Let  $v$ be active in $\alpha$.  By (iii),  there is $a \in A^+(v) \cap {\cal T}_A$. By definition of ${\cal T}_A$ we have  $\theta(a) \notin \trans(\phi, r')$. We show the existence of a guiding forest for  $r_{|v} \lineareq{\partial^A}{\phi_{|v}} r'_{|v}$.

    If  $F(v) \cap \trans(\phi, r') = \emptyset$, then there
     is nothing to check and the empty set is suitable.
    Otherwise let $\hat a \in \trans(\phi_{|v}, r'_{|v})$ and suppose that the set of faces that are elements of $\phi_{|v}$ is not guiding.
    Since $G^A(v)$ is a directed cycle, we also
    have $\theta(\hat a) \in \trans(\phi_{|v}, r'_{|v})$, $\theta^2(\hat a) \in \trans(\phi_{|v}, r'_{|v})$ and so on, so that  $A^+(v) \subset \trans(\phi_{|v}, r'_{|v})$. This is a contradiction with $\theta(a)$ not being transitory.  Hence, the support of $\phi$ is guiding  $r \lineareq{\partial^A}{\phi} r$. This together with (ii) implies by~\thref{prop:CNS_routing_vector} that $r \legalseq{\partial^A}{\phi} r'$. Hence  condition (2) is satisfied. 
    Finally, by \thref{lem:terminal_condition}, we have ${\cal T}_A = \tail^A({\cal T}_F)$,  hence (iv) and (3) are equivalent.

    Conversely, suppose that condition (1), (2) and (3) are satisfied. We already showed that (3) implies (iv).
    
    By (2), for every active vertex $v$ in $\alpha$ there must be an arc $a \in A^+(v)$ which is not transitory for $r \legalseq{\partial^A}{\phi} r'$. This corresponds to condition (iii).

    Finally, condition (1) implies (i), and condition (2) implies (ii), hence the equivalence between the two sets of conditions.
\end{proof}

To conclude this section, we show how to apply these results
to the case of standard rotor-routing. We suppose that $\sigma, \sigma' \in C^+_V$, and that  $\rho,\rho'$ are rotor configurations (see Section~\ref{sec:cyclic_GRM}).

\begin{corollary}
    Suppose that $(G^A, G^V)$ is a cyclic GRM multigraph
    and that $(\rho',\sigma') \lineareq{\cL}{\phi} (\rho,\sigma)$, with $\phi \in C_F^+$, with $\rho, \rho'$ rotor configurations
    and $\sigma, \sigma' \in C^+_V$. Let $\alpha = \tail^A(\phi)\in C_A^+$. Then $(\rho,\sigma) \legalseq{\cL}{\phi} (\rho',\sigma')$ if and only if 
    $$\{a \in \alpha : \theta(a) \in \rho'\}$$
    is guiding for $\sigma \lineareq{\partial^V}{\alpha} \sigma'$.
\end{corollary}

\begin{proof}

Recall that in this context ${\cal T}_A$ as defined in \thref{thm:charac_legal_rotor}  can be expressed as
$${\cal T}_A = \{ a \in \alpha: \theta(a) \in \rho'\} \cup \{a \in \alpha : \theta(a) \notin \alpha \} .$$ 
We prove that ${\cal T}_A = \{a \in \alpha : \theta(a) \in \rho'\}$. Let $a \in \alpha$ such that $\theta(a) \notin \alpha$.
Then $\rho'_{\theta(a)} = \rho_{\theta(a)} + \alpha_a $ which implies that $\rho'_{\theta(a)} \geq 1$ and then $\rho'_{\theta(a)} = 1$. Thus, $\{ a \in \alpha: \theta(a) \notin \alpha \} \subset \{ a \in \alpha: \theta(a) \in \rho'\}$.

Suppose that $\{a \in \alpha : \theta(a) \in \rho'\}$ is guiding for $\sigma \lineareq{\partial^V}{\alpha} \sigma'$.
We check conditions (i) to (iv) of \thref{thm:charac_legal_rotor}. 

Since $\sigma' \in C_V^+$ and $\rho' \in C_A^+$, conditions (i) and (ii) are satisfied. Moreover,  for every vertex $v \notin S$ there is $a \in A^+(v)$ such that $\rho'_a > 0$. Hence $ \{a \in \alpha: \rho'_{\theta(a)} > 0 \} \cap A^+(v)$ is nonempty for every active vertex $v$ in $\alpha$ and  condition (iii) is satisfied. Condition (iv) follows from ${\cal T}_A = \{a \in \alpha : \theta(a) \in \rho'\}$. 

Conversely, condition (iv) implies that $\{a \in \alpha : \theta(a) \in \rho'\}$ is guiding for $\sigma \lineareq{\partial^V}{\alpha} \sigma'$.

\end{proof}

In the special case where $G^V$ is stopping, and we aim
to simulate maximal rotor walks (see Sec. \ref{sec:standard}) in GRM multigraphs,
the previous result leads to the characterization of runs
among flows, as stated in   \thref{thm:run_carac}.

\section{Legal reachability in GRM multigraphs}
\label{sec:legal_reach}

In the previous section, we developed a procedure to decide whether a legal routing exists in GRM multigraphs when a routing vector is specified, and we derived a simplified characterization for the cyclic case; both can be verified in polynomial time.
In this section, we extend the analysis to the setting where no routing vector is given.

In full generality, the problem becomes computationally hard: we prove that legal reachability in GRM multigraphs is NP-complete (Section~\ref{sec:legal_reach_general}).
Since the proof is lengthy, we present it separately for clarity (Section~\ref{sec:proof_np_hard}).

In contrast, the cyclic case admits an explicit structural description.
Periodic routing vectors form a computable lattice, and every reachability instance reduces to checking a small set of canonical candidates.
This yields a complete polynomial-time criterion for legal reachability without a prescribed routing vector in cyclic GRMs (Section~\ref{sec:cyclic_GRM_reachability}), which in particular encompasses standard rotor-routing.

The main results of this section are \thref{thm:np_hard1} and \thref{thm:lr_cy_grm}.

\subsection{General Case}
\label{sec:legal_reach_general}

Let us  define formally the  {\sc Legal reachability in GRM multigraph} problem as follows:

\vskip .5cm

\begin{tabularx}{15cm}{lp{11cm}}
    \hline \multicolumn{2}{c}{{\sc Legal reachability in GRM multigraph (LR-GRMM)}} \\
    \hline
    {\sc input:} & $(G^A, G^V)$ a GRM multigraph,  $r, r' \in C_A$ and $\sigma, \sigma' \in C_V$.\\       
    {\sc question:} & does
    $(r,\sigma) \legalseq{\cL}{*} (r',\sigma')$ 
    ? \\
    \hline
\end{tabularx}

\vskip .5cm
The challenge lies in the fact that multiple routing vectors
$\phi$ may satisfy the equation $(r', \sigma') = (r, \sigma) + \cL(\phi) $. Among these, some routing vectors may admit a legal routing sequence, while others may not.
We prove:

\begin{theorem}
    \thlabel{thm:np_hard1}
    {\sc Legal Reachability in GRM multigraph} problem is NP-complete.
\end{theorem}

\subsection{Proof of \thref{thm:np_hard1}}
\label{sec:proof_np_hard}

A routing vector $\phi \in C_F^+$ with 
  $(r,\sigma) \legalseq{\cL}{\phi} (r',\sigma')$ 
is a certificate that $((G^A, G^V), r,r',\sigma,\sigma')$ is a positive instance of {\sc LR-GRMM}, which can be checked in polynomial time by \thref{thm:routing_rotor_mecha}. Hence {\sc LR-GRMM} is in NP. The  rest of the section is dedicated to the proof of NP-hardness by polynomial reduction from a boolean satisfiability problem.

We consider a special version of the {\sc 3-SAT} problem, where boolean formulas are given in conjunctive normal form with clauses of 3 literals, where each variable appears exactly twice unnegated and exactly twice negated,
as in the formula 
\[(x_1 \vee x_2 \vee x_3) \wedge (x_1 \vee \overline{x_2} \vee \overline{x_3}) \wedge (\overline{x_1} \vee x_2 \vee \overline{x_3}) \wedge (\overline{x_1} \vee \overline{x_2} \vee x_3).\]
This restriction of {\sc 3-SAT}, which we will call {\sc 3-SAT-(2,2)}, has been proved NP-complete~\cite{darmann_simplified_2021}. An instance of \mbox{{\sc 3-SAT-(2,2)}} is then a boolean formula made of $n$ variables $x_1, x_2, \dots, x_n$ and $m$ clauses of three literals $c_1, c_2, \dots, c_m$, such that for all $i$, $x_i$ appears in exactly two clauses $c_j$, and $\overline{x_i}$ appears in exactly two clauses $c_j$ as well.

To  an instance of {\sc 3-SAT-(2,2)}, we will associate  an instance $((G^A, G^V), r, r', \sigma, \sigma')$ of {\sc LR-GRMM}, where $(G^A, G^V)$ is a GRM multigraph. 
We shall first define a core multigraph, 
consisting only of arcs and vertices. Then, the clause
gadgets will add some faces to the core multigraph, and
then the variable gadgets will add more vertices and arcs, as well as faces. Finally, we will define $r, r', \sigma$ and $ \sigma'$.

\paragraph{Core multigraph.}

We begin by defining the core multigraph:
\begin{itemize}

\item the set of vertices is $V  = V_\text{var}\cup V_\text{clause}  \cup S$ with
\begin{itemize}
\item the set of \emph{variable vertices}, $  V_\text{var} = \{x_i\}_{ i \in \{1, \dots, n\}}$ 
\item the set of \emph{clause vertices}, $ V_\text{clause} = \{c_j\}_{j \in \{1, \dots, m\}}$
\item the set of sinks $ S = \{s, s_\text{sat}\} $
\end{itemize}

\item the set of arcs is $A  = A_\text{var}^+ \cup A_\text{var}^- \cup A_\text{clause,s} \cup A_\text{clause,sat} $ with
\begin{itemize}
\item $A_\text{var}^+$ contains an arc from every variable vertex to each of the two clause vertices where it appears unnegated; more precisely $A_\text{var}^+ = \{a_{i,j}^+\}_{i \in \{1, \dots, n \}, j \in \{1,2\}}$, where $\tail(a_{i,j}^+) = x_i$ and $\{ \head(a_{i,1}^+), \head(a_{i,2}^+) \}$ are the two clause vertices $c$ such that $x_i$ appears unnegated in $c$;

\item $A_\text{var}^-$ contains an arc from every variable vertex to each of the two clause vertices where it appears negated;
$A_\text{var}^- = \{a_{i,j}^-\}_{i \in \{1, \dots, n \}, j \in \{1,2\}}$, where $\tail(a_{i,j}^-) = x_i$ and 
$\{ \head(a_{i,1}^-), \head(a_{i,2}^-) \}$ are the two clause vertices $c$ such that $x_i$ appears negated in $c$;

\item $ A_\text{clause,s}$ contains an arc from every clause vertex to $s$; more precisely 
$ A_\text{clause,s} = \{a_{j,1}\}_{j \in \{1, \dots, m \}}$ with $(\tail(a_{j,1}), \head(a_{j,1})) = ( c_j , s)$;

\item $ A_\text{clause,sat}$ contains an arc from every clause vertex to $s_\text{sat}$; more precisely  $A_\text{clause,sat} = \{a_{j,0}\}_{j \in \{1, \dots, m \}}$ with   $(\tail(a_{j,0}), \head(a_{j,0})) = (c_j, s_\text{sat})$.

\end{itemize}
\end{itemize}
Figure~\ref{fig:reduction_core} shows an example of this core multigraph.

\begin{figure}[!htbp]
    \centering
    \begin{tikzpicture}[
      node distance=2cm and 2cm,
      every node/.style={minimum size=1cm},
      every path/.style={-Stealth}
    ]

    \node[draw] (x1) {x$_1$};
    \node[draw] (x2) [right=of x1] {x$_2$};
    \node[draw] (x3) [right=of x2] {x$_3$};

    \node[draw] (c1) [below=of x1, xshift=-1.5cm] {c$_1$};
    \node[draw] (c2) [right=of c1] {c$_2$};
    \node[draw] (c3) [right=of c2] {c$_3$};
    \node[draw] (c4) [right=of c3] {c$_4$};

    \node[draw] (ssat) [below=of c2] {s$_{\text{sat}}$};
    \node[draw] (s) [right=of ssat] {s};

    \draw (x1) -- (c1);
    \draw (x1) -- (c2);
    \draw (x1) -- (c3);
    \draw (x1) -- (c4);
    \draw (x2)  -- (c1);
    \draw (x2)  -- (c2);
    \draw (x2)  -- (c3);
    \draw (x2) -- (c4);
    \draw (x3)  -- (c1);
    \draw (x3)  -- (c3);
    \draw (x3) -- (c2);
    \draw (x3) -- (c4);

    \draw (c1) to  (s);
    \draw (c2) to  (s);
    \draw (c3) to  (s);
    \draw (c4) to  (s);
    \draw (c1) -- (ssat);
    \draw (c2) -- (ssat);
    \draw (c3) -- (ssat);
    \draw (c4) -- (ssat);

    \end{tikzpicture}
    \caption{Core multigraph  built from the {\sc 3-sat-(2,2)} instance $c_1 \wedge c_2 \wedge c_3 \wedge c_4 = (x_1 \vee x_2 \vee x_3) \wedge (x_1 \vee \overline{x_2} \vee \overline{x_3}) \wedge (\overline{x_1} \vee x_2 \vee \overline{x_3}) \wedge (\overline{x_1} \vee \overline{x_2} \vee x_3) $. The graph does not differentiate among unnegated and negated variables.}
    \label{fig:reduction_core}
\end{figure}
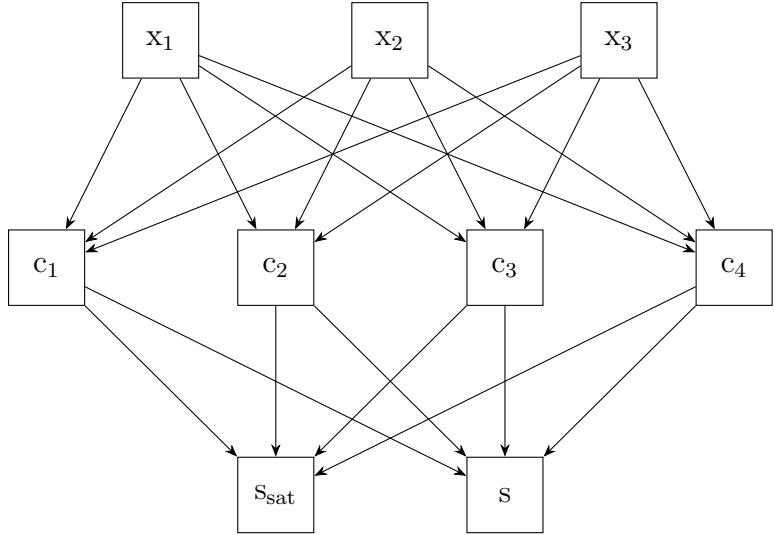

In the following, we extend this graph by adding two gadgets called respectively \emph{clause gadget} and \emph{variable gadget} that will ensure that any legal routing for the {\sc LR-GRMM} instance is coherent with the choice of a satisfying assignment of the {\sc 3-SAT-(2,2)} instance, if any.

\paragraph{Clause gadget.}

In the core multigraph, for each clause vertex $c_j, j \in \{1, \dots, m \}$, we have $A^+(c_j) = \{a_{j,0}, a_{j,1}\}$. We add for every $j$ a set of faces $F(c_j) = \{f_j^{01}, f_j^{11} \}$ where $(\tail^A(f_j^{01}), \head^A(f_j^{01})) = (a_{j,0}, a_{j,1})$ and $(\tail^A(f_j^{11}), \allowbreak \head^A(f_j^{11})) = (a_{j,1}, a_{j,1})$. See Fig.~\ref{fig:gadget_clause1} for an illustration of the clause gadget.

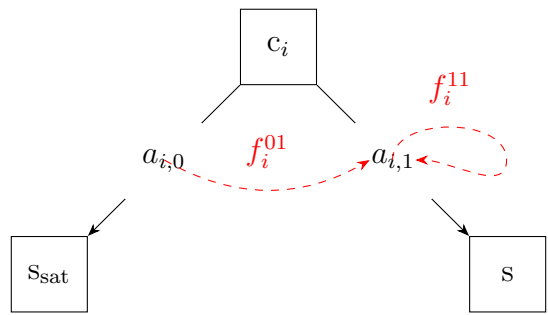
\begin{figure}[!htbp]
    \centering
    \begin{tikzpicture}[
      node distance=2cm and 2cm,
      every node/.style={minimum size=1cm},
      every path/.style={-Stealth}
    ]

    \node[draw] (c)  {c$_i$};
    \node[draw] (s) [below right=of c] {s};
    \node[draw] (ss) [below left=of c] {s$_\text{sat}$};

    \draw (c)  to node[midway, fill=white] {$a_{i, 1}$} (s);
    \draw (c)  to node[midway, fill=white] {$a_{i, 0}$} (ss);

    \draw[red, dashed, bend right] (-1.5, -1.5) to node[midway,  above] {$f_i^{01}$} (1.2, -1.5) ;
    \draw[red, dashed, bend right=30] (1.5, -1.5) to[out=90, in=90] node[midway, above] {$f_i^{11}$} (3, -1.5) to[out=90, in=180]   (1.8, -1.5);
    \end{tikzpicture}

    \caption{The clause gadget for $c_i$. Faces are represented by dashed red arcs. }
    \label{fig:gadget_clause1}
\end{figure}

\paragraph{Variable gadget.}

For every $x_i, i \in \{1, \dots, n \}$,
we add two sink vertices $s_i^\text{start}$ and $ s_i^\text{end} $,
and two arcs $a^\text{start}_i$ and $a^\text{end}_i$, respectively from $x_i$ to $s_i^\text{start}$ and $ s_i^\text{end}$. At this point, we have
\[A^+(x_i) = \{a^\text{start}_i, a^\text{end}_i, a^+_{i,1}, a^+_{i,2}, a^-_{i,1}, a^-_{i,2} \}\]
where:
\begin{itemize}
\item $a^+_{i,1}, a^+_{i,2} $ (resp.  $a^-_{i,1}, a^-_{i,2} $) are arcs whose heads are the clause vertices  where $x_i$ appears unnegated (resp. negated)
\item $a^\text{start}_i$ is such that $\head(a^\text{start}_i) = s_i^\text{start}$
\item $a^\text{end}_i$ is such that $\head(a^\text{end}_i) = s_i^\text{end}$.
\end{itemize}

We now add for every $i$ the sets of faces $F(x_i) = \{ f_{i,1}^+, f_{i,2}^+, f_{i,3}^+, f_{i,1}^-, f_{i,2}^-, f_{i,3}^-  \}$, as illustrated in Fig.~\ref{fig:gadget_variable2}. Faces $f_{i,1}^+, f_{i,2}^+, f_{i,3}^+$ are the \emph{positive faces} of $x_i$ and $f_{i,1}^-, f_{i,2}^-, f_{i,3}^-$ are the \emph{negative faces} of $x_i$.

\begin{figure}[!htbp]
    \centering
    \begin{tikzpicture}[
      node distance=2cm and 2cm,
      every node/.style={minimum size=1cm},
      every path/.style={-Stealth}
    ]

    \node[draw] (xi) {x$_i$};
    \node[draw] (c2) [below=of xi] {s$_i^\text{end}$};
    \node[draw] (c1) [left=of xi] {c$_2$};
    \node[draw] (c4) [above=of xi] {s$_i^\text{start}$};
    \node[draw] (iplus) [right=of c4] {c$_3$};
    \node[draw] (si) [right=of xi] {c$_4$};
    \node[draw] (imoins) [left=of c4] {c$_1$};

    \draw (xi)  to node[midway, fill=white] {$a_{i,2}^+$} (c1);
    \draw (xi)  to node[midway, fill=white] {$a_{i}^\text{end}$} (c2);
    \draw (xi)  to node[midway, fill=white] {$a_{i}^\text{start}$} (c4);
    \draw (xi)  to node[midway, fill=white] {$a_{i,1}^-$} (iplus);
    \draw (xi)  to node[midway, fill=white] {$a_{i,1}^+$} (imoins);
    \draw (xi)  to node[midway, fill=white] {$a_{i,2}^-$} (si);

    \draw[red, dashed] (0.1, 1.5) to node[midway,  above] {$f_{i,1}^-$} (1.2, 1.5) ;
    \draw[red, dashed] (-0.1, 1.5) to node[midway,  above] {$f_{i,1}^+$} (-1.2, 1.5) ;
    \draw[red, dashed] (1.4, 1.3) to node[midway,  right] {$f_{i,2}^-$} (1.4, 0.2) ;
    \draw[red, dashed] (-1.4, 1.3) to node[midway,  left] {$f_{i,2}^+$} (-1.4, 0.2) ;
    \draw[red, dashed, bend left] (1.3, -0.2) to node[midway,  right] {$f_{i,3}^-$} (0.2, -1.6) ;
    \draw[red, dashed, bend right] (-1.3, -0.2) to node[midway,  left] {$f_{i,3}^+$} (-0.3, -1.6) ;
    
    \end{tikzpicture}

    \caption{The variable gadget for $x_i$ which is assumed to appear unnegated in clauses $c_1$ and $c_2$, and negated in clause $c_3$ and $c_4$. Faces are represented by dashed  red arcs. The positive faces are $f_{i,1}^+, f_{i,2}^+, f_{i,3}^+$, and the negative faces are $f_{i,1}^-, f_{i,2}^-, f_{i,3}^-$.}
    \label{fig:gadget_variable2}
\end{figure}
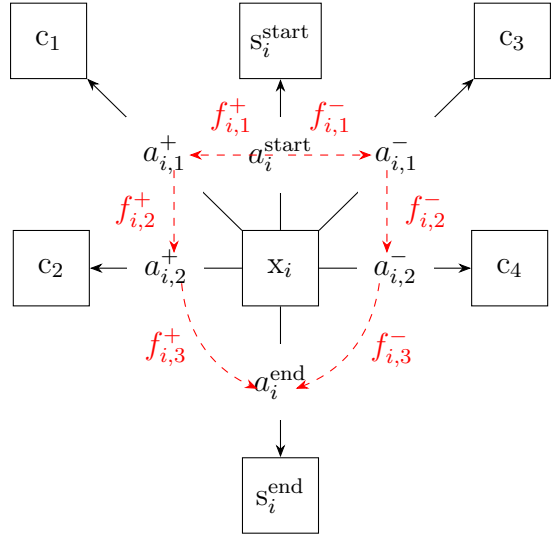

\paragraph{Full specification of the {\sc LR-GRMM} instance.}

Finally, to a {\sc 3-SAT-(2,2)} instance, we associate the {\sc LR-GRMM} problem $((G^A, G^V), r, r', \sigma, \sigma')$ defined by 
\begin{itemize}

    \item $(G^A, G^V)$ is the GRM multigraph  as described above by union of the core multigraph, and all variable and clause gadgets;
    
    \item $\sigma = 3 \sum_{i=1}^n x_i$  (i.e. three particles on each variable vertex) ;

    \item $\sigma'  = m   s_\text{sat} + (2n -m) s + \sum_{i=1}^n s_i^\text{start} $ ;

    \item $r = \sum_{i=1}^n a_{i}^\text{start} + \sum_{j = 1}^m a_{j,0} $ ;

    \item $r' = \sum_{i=1}^n a_{i}^\text{end} + \sum_{j = 1}^m a_{j,1} $.

\end{itemize}

\paragraph{Proof of the reduction.}

It is given in two separate lemmas.

\begin{lemma}
    If the {\sc 3-SAT-(2,2)} instance is satisfiable, then the corresponding instance of {\sc LR-GRMM} has a solution.
\end{lemma}

\begin{proof}
    Assume that the {\sc 3-SAT-(2,2)} instance has a satisfiable assignment. We will construct a solution for the corresponding {\sc LR-GRMM} problem $((G^A,G^V), r, r', \sigma, \sigma')$ by giving a routing vector $\phi = \phi^\text{var} + \phi^\text{clause}$, where the supports of $\phi^\text{var}$ and $\phi^\text{clause}$ are respectively contained in the union of the faces set $F(x_i)$ for all variable vertices $x_i$, and in the union of $F(c_i)$ for all clause vertices $c_i$. We will prove that  $(r, \sigma) \legalseq{\cL}{\phi} (r', \sigma')$.
    
    Let $\phi^\text{var} \in C_F$ be defined as 
    $\phi^\text{var} = \sum_{i=1}^n \phi_i^\text{var}$,
    where $\phi_i^\text{var} = f^+_{i,1} + f^+_{i,2} + f^+_{i,3}$
    if the satisfying assignment of $x_i$ is true,
    and $\phi_i^\text{var} = f^-_{i,1} + f^-_{i,2} + f^-_{i,3}$
    if it is false. In both cases, note that 
    $\partial^A(\phi_i^\text{var}) = a_i^\text{end} -a_i^\text{start}$.

    If we define $ (r_0,\sigma_0) = (r,\sigma) + \cL(\phi_\text{var})$, then it follows that
    $r_0 = \sum_{i=1}^n a_i^\text{end} + \sum_{j = 1}^m a_{j,0} $. On the other hand, all particles were transferred
    in this routing from variable vertices $x_i$ to clause vertices $c_j$ and sinks $s_i^\text{start}$, and
    we can write $\sigma_0 = \sum_{i=1}^m s_i^\text{start} +  \sum_{j=1}^m \ell_j c_j$  for some integers $\ell_1, \dots, \ell_m$ satisfying $1 \leq \ell_j \leq 3$ for all $1 \leq j \leq m$, and $\sum_{i=1}^m \ell_i = 2n$ (the fact that $\ell_j \geq 1$ follows from the assignment satisfying the formula). Note that
    this routing can be done legally, simply by following the order $1,2,3$ on faces.

    Then, consider $\phi_\text{clause}$ as the routing vector consisting of legally routing from $(r_0, \sigma_0)$ all particles $\ell_j$ from clause vertices $c_j$ to sinks $s_\text{sat}$ and $s$, for all $1 \leq j \leq m$. Since $1 \leq \ell_j \leq 3$ for all $1 \leq j \leq m$, at least one outgoing arc of each clause vertex $c_i$ is routed and will emit one particle to $s_\text{sat}$, whereas the other routed arcs in $A^+(c_j)$, if any, will emit particles to sink $s$. Then we have $(r', \sigma') = (r, \sigma) + \cL(\phi^\text{var} + \phi^\text{clause})$ and the whole routing can be done legally.
\end{proof}




\begin{lemma}
    Let $((G^A, G^V), r, r', \sigma, \sigma')$ be the {\sc LR-GRMM} problem associated to a  given {\sc 3-SAT-(2,2)} instance. If $((G^A, G^V), r, r', \sigma, \sigma')$  has a solution, then the {\sc 3-SAT-(2,2)} instance is satisfiable.
\end{lemma}

\begin{proof}

    Assume that $((G^A, G^V), r, r', \sigma, \sigma')$ has a solution and consider a legal routing sequence from $(r,\sigma)$ to $(r',\sigma')$.
    
    Consider a variable vertex $x_i$ with  $i \in \{1,\dots,n\}$:
    since there are 3 particles on $x_i$ in $\sigma$, that in $r$ the only arc
    of $A^+(x_i)$ is $a_i^\text{start}$, and that in $r'$ the only arc
    of $A^+(x_i)$ is $a_i^\text{end}$, we see that the only possible
    legal routings in $A^+(x_i)$ are in that order $(f_{i,1}^+$, $f_{i,2}^+$, $f_{i,3}^+)$ or $(f_{i,1}^-$, $f_{i,2}^-$, $f_{i,3}^-)$ (consider Fig. \ref{fig:gadget_variable2}). Hence, one particle is routed to $s_i^\text{start}$, and the other two are either routed to the 2 clause vertices where $x_i$ appears negated, or to the 2 clause vertices where $x_i$ appears unnegated.

    Consider now a clause vertex $c_j$ with  $j \in \{1,\dots,m\}$:
    since the only arc in $r$ of $A^+(c_j)$  is $a_{j,0}$, and in $r'$ 
    the only arc of $A^+(c_j)$  is $a_{j,1}$, we see that there must be at least one particle routed from a variable vertex to $c_i$ during the routing sequence.

    From this, we can built a truth assignment, by letting 
    $x_i$ be true if and only if two of its 4 particles were routed  to the 2 clauses where $x_i$ is unnegated. We then see that the {\sc 3-SAT-(2,2)} problem is satisfiable.
\end{proof}

\subsection{Cyclic case}
\label{sec:cyclic_GRM_reachability}

We now consider the case where the GRM multigraph is cyclic. We will demonstrate that the {\sc Legal Reachability in Cyclic GRM Multigraph} problem can be solved in polynomial time. To achieve this, a deeper understanding of the set of routing vectors between the two configurations in question is required.

\subsubsection{Periodicity in cyclic GRM multigraphs}

Consider a cyclic GRM Multigraph. We define a {\bf periodic routing vector} to be a vector
$\phi \in C_F$, such that 
$$(r,\sigma) \lineareq{\cL}{\phi} (r,\sigma)$$
for any $(r,\sigma) \in C_A \times C_V$ (note that by linearity, this is independent on the choice of configurations). The intuition behind this notion is clear: it is a way of routing particles so that we come back to the initial configuration. In this section, we describe the structure of periodic routing vectors.

Some partial results are known. For instance, in classical rotor-routing mechanics of a Eulerian directed graph, it is possible to come back to  $(\rho,\sigma)$ after routing exactly once along every arc of $G$ and this is the only way to do so, up to repeating this process. The same is true in a general strongly connected directed graph, but some arcs should be routed different number of times, depending on the graph.

Let us come back to cyclic GRM Multigraphs and suppose that we have 
$$(r,\sigma) \lineareq{\cL}{\phi} (r,\sigma).$$
It follows that $\cL(\phi)=0$, i.e. $\partial^A(\phi)=0$ and $\partial^V\circ \tail^A(\phi)=0$. The set of periodic routing vectors is then exactly equal to 
$\ker(\partial^A) \cap \ker(\partial^V \circ \tail^A)$.

With a slight abuse of notation, let us
    define for all $v \in V$ the sum $F(v) \in C_F$ as 
    the sum of faces in $G^A(v)$, i.e. 
    $\sum_{f \in F(v)} f$, and extend $F$ to an
    homomorphism $F : C_V \rightarrow C_F$.
For $p \in C_V$, $F(p)$ can be thought as a routing vector
that makes $p_v$ full turns on the cyclic rotor at $v$,
for every $v \in V$.
We can then use the following result (recall that primitive period vectors are defined in~\thref{thm:kerdelta}).

\begin{proposition}
    \thlabel{prop:period_rotor} 
    Let $(G^A, G^V)$ be a cyclic GRM multigraph
    and let $k$ be the number of leaf components of $G^V$ that are not singletons $\{s\}$ where $s \in S$ is a sink.
    Then the rank of $\ker(\partial^A) \cap \ker(\partial^V \circ \tail^A)$
    is $k$.
    More precisely, let $p_1, p_2, \dots, p_k \in C_V$ be the primitive period vectors of $G^V$ corresponding to these
    components. 
    A basis 
    of $\ker(\partial^A) \cap \ker(\partial^V \circ \tail^A)$ is 
    then $(F(p_1), F(p_2),\dots,F(p_k))$.
\end{proposition}

\begin{proof}
    With notation above, it is easy to see that the Laplacian operator $\Delta$ is
    $\Delta = - \partial^V \circ \tail^A \circ F$.

    Because of the cyclic structure of $G^A(v)$ for every 
    $v \in V \setminus S$, by \thref{prop:kirchoff}
    it is easy to see that $(F(v))_{v \in V \setminus S}$ is a basis of $\ker(\partial^A)$, and that $F$ induces an isomorphism
    from $C_{V \setminus S} \subset C_V$ onto
    $\ker(\partial^A)$. An element $F(\sigma) \in \ker(\partial^A)$,
    with $\sigma \in C_{V \setminus S}$, is then
    in $\ker(\partial^V \circ \tail^A)$ if and only if 
    $\partial^V\circ \tail^A( F(\sigma)) = -\Delta(\sigma) = 0$.
    Hence, 
    $$\ker(\partial^A) \cap \ker(\partial^V \circ \tail^A) = F(\ker(\Delta)).$$
    By \thref{thm:kerdelta}, 
    a basis for $\ker(\Delta)$ is given by $p_1, p_2, \dots, p_k$ together with each element of $S$, and so
    $(F(p_i))_{1 \leq i \leq k}$ is a basis of $F(\ker(\Delta))$.
\end{proof}

Interpreting this result for standard rotor-routing, we can state that:
\begin{corollary}
\thlabel{cor:rotor_period}
    Let $(G^A,G^V)$ be a cyclic rotor multigraph. Then:
    \begin{itemize}
        \item if $(G^A,G^V)$ is stopping and $(r,\sigma) \lineareq{\cL}{*} (r',\sigma')$ then the routing vector from $(r,\sigma)$ to $(r',\sigma')$ is unique;
        \item if $(G^A,G^V)$ is strongly connected and $(r,\sigma) \lineareq{\cL}{\phi} (r',\sigma')$ then routing vectors are of the form  $\phi + k p_1$ with $k \in \Z$, and $ \ker(\partial^A) \cap \ker(\partial^V \circ \tail^A) = \Z \cdot F(p_1)$, with $p_1>0$, $p_1 \in C_V$. 
    \end{itemize}
\end{corollary}

\begin{proof}
    In the first case, note that all leaf  components
    are of the form $\{s\}$ with $s \in S$, hence $ \ker(\partial^A) \cap \ker(\partial^V \circ \tail^A)$ is trivial.

    In the second case, there is a single primitive vector $p_1>0$
    and $ \ker(\partial^A) \cap \ker(\partial^V \circ \tail^A)$ has rank $1$.
\end{proof}

In the case of a strongly connected multigraph, the existence of $p_1$ implies that if $(r,\sigma)  \lineareq{\cL}{*} (r',\sigma')$ then there is a minimal routing vector $\phi_1 \geq 0$ from $(r,\sigma)$ to $(r',\sigma')$. In other words,  if $\phi \geq 0$ is a routing vector from $(r,\sigma)$ to $(r',\sigma')$ then $\phi_1 \leq \phi$. 

To compute $\phi_1$, one can start with any routing vector $\phi$ such that $(r,\sigma) \lineareq{\cL}{\phi} (r',\sigma')$, which can be found using \thref{lem:routing_vector}. Then $\phi_1 = \phi + k^* p_1$, where $k^*$ is the smallest integer $k$ such that $\phi + k F(p_1)$ is nonnegative.

\subsubsection{Legal reachability in strongly connected cyclic GRM multigraphs}
\label{sec:toto}

Let $(G^A, G^V)$ be a strongly connected cyclic GRM multigraph and let $p \in C_V^+$ be the primitive period vector. 
We say that $(r, \sigma) \in C_A \times C_V$ is \emph{recurrent} if there is a legal routing sequence to itself with routing vector $F(p)$, i.e. $(r,\sigma) \legalseq{\cL}{F(p)}  (r,\sigma)$.
Note that by \thref{thm:charac_legal_rotor}, determining if $(r,\sigma)$ is recurrent can be done in polynomial time.

\begin{lemma}
    \thlabel{lem:recurrent}
    Let $(G^A, G^V)$ be a strongly connected cyclic GRM multigraph.
    Suppose that $(r,\sigma) \legalseq{\cL}{\phi_1}  (r',\sigma')$ for some routing vector $\phi_1 \in C_F^+$ such that $\phi_1 \geq F(p)$.
    Let $\phi_0 \in C_F^+$ be the smallest nonnegative routing vector from $(r,\sigma)$ to $(r',\sigma')$. Then $(r', \sigma')$ is recurrent and $(r,\sigma) \stackrel[\cL]{\phi_0+F(p)}{\longrightarrow} (r',\sigma')$.
\end{lemma}

\begin{proof}

    Let $\phi \in C_F$ satisfy $\phi \geq F(p)$ and
     $(r,\sigma) \lineareq{\cL}{\phi}  (r',\sigma')$.

    According to \thref{thm:charac_legal_rotor}, we have
    $(r,\sigma) \legalseq{\cL}{\phi}  (r',\sigma')$
    if and only if the following conditions hold, where
    $\alpha = \tail^A(\phi)$ 
        and ${\cal T}_A = \{a \in \alpha: r'_{\theta(a)} > 0 \} \cup \{a \in \alpha : \theta(a) \notin \alpha \}$:
    \begin{enumerate}[(i)]
        \item $\sigma'_v \geq 0$ for every vertex $v$ active in $\alpha$;
        \item $ r'_a \geq 0$ for every arc $a \in \alpha$;
        \item ${\cal T}_A \cap A^+(v) \neq \emptyset$ for every vertex $v$ active in $\alpha$;
        \item ${\cal T}_A$ is guiding for $\sigma  \lineareq{\partial^V}{\alpha} \sigma'$.
    \end{enumerate}

    Given that $\sum_{v \in V} v \leq p$, and $F(p) \leq \phi$, it follows that all arcs are in $\alpha$ and all vertices
    are active in $\alpha$. Hence $\{a \in \alpha : \theta(a) \notin \alpha \} $ is an empty set, i.e. 
    ${\cal T}_A = \{a \in \alpha: r'_{\theta(a)} > 0 \}$.
    Then the set of conditions (i), (ii), (iii) and (iv) can be simplified as:
    \begin{enumerate}[(1)]
        \item $ r' \in C_A^+$;
        \item $\sigma' \in C_V^+$;
        \item for every vertex $v \in V$, there is $a \in A^+(v)$ with $r'_a \geq 1$;
        \item for every $v \in V$ such that $\sigma'(v)=0$,
        there is a directed path within the set $\{a \in A : r'_{\theta(a)} > 0 \}$ to a vertex $v'$ where $\sigma'(v') > 0$.
    \end{enumerate}

    Therefore, these conditions do not depend on the initial configuration $(r, \sigma)$, but only on $(r',\sigma')$.
    Since they are satisfied with $\phi = \phi_1$, it follows that
    $(r',\sigma') \legalseq{\cL}{F(p)}  (r',\sigma')$
    is also true, as well as
    $(r,\sigma) \stackrel[\cL]{\phi_0+F(p)}{\longrightarrow}  (r',\sigma')$.

\end{proof}

The contrapositive of this lemma indicates that any legal routing sequence from $(r, \sigma)$ to a non-recurrent configuration $(r', \sigma')$ with routing vector $\phi \in C_F^+$ must satisfy the condition that there exists a face $f \in F$ such that $\phi_f < p_a$ where $p$ is the primitive period routing vector and $a = \tail^A(f)$. This is stated in the following corollary.

\begin{corollary}
    Let $(G^A, G^V)$ be a strongly connected cyclic GRM multigraph.
   Assume that $(r', \sigma')$ is not recurrent. If there is $\phi \in C_F^+$ such that $(r,\sigma) \legalseq{\cL}{\phi}  (r',\sigma')$, then
    $\phi$ is the smallest nonnegative vector such that 
    $(r,\sigma) \lineareq{\cL}{\phi}  (r',\sigma')$.
\end{corollary}

It has been shown in~\cite{tothmeresz_rotor-routing_2022} that, in the classical routor-routing framework, if there is a legal routing sequence from $(\rho,\sigma)$ to $ (\rho',\sigma')$, then there is a legal routing sequence with  the smallest nonnegative routing vector, known as reduced routing vector. However, this does not hold in the more general context of cyclic GRM multigraphs, as illustrated in Figure~\ref{fig:rec_pas_legal}.

\begin{figure}[!htbp]
    \centering
    \begin{tikzpicture}[->, >=stealth, node distance=2cm]
        
    \node[draw] (v) at (0, 2) {$v$};
    \node[draw] (u) at (0, 0) {$u$};

    \draw[bend left] (u) to node[midway, right] {$c$} (v);
    \draw[bend left] (v) to node[midway, right] {$b$} (u);
    \draw[->, loop above, min distance=1cm] (v) to node[left] {$a$} (v);

    \draw[red, dashed, bend left] (0.2, 2.6) to node[midway,  right] {$f_a$}  (0.4, 1.2);
    \draw[red, dashed, bend right] (0.3, 0.6) to  (1.4, 1.8) to node[midway,  right] {$f_b$} (0.2, 2.8);
    \draw[red, dashed, bend right=30] (-0.45, 1) to[out=90, in=90] node[midway, above] {$f_c$} (-2.0, 1) to[out=90, in=90]    (-0.45, 1);

    \end{tikzpicture}
        
    \caption{
    This is an example to prove that for strongly connected cyclic GRM multigraphs, the smallest $\phi$
    such that $(r,\sigma) \legalseq{\cL}{\phi}  (r',\sigma')$
    is not always the smallest nonnegative vector 
    $\phi_0$ such that $(r,\sigma) \lineareq{\cL}{\phi_0}  (r',\sigma')$; however $\phi_1=\phi_0+F(p)$
    always works if $(r',\sigma')$ is recurrent.
    In the case above, the primitive period vector is $p = u+v$ and $\tail^A(F(u+v)) = a+b+c$. Let $\sigma = \sigma' = u$, $r=2a+c$ and $r' = a+b+c$. We have $(r', \sigma') = (r, \sigma) + \cL(f_a)$. The smallest nonnegative routing vector from $(r,\sigma)$ to $(r',\sigma')$ is $f_a$. However, there does not exist a legal sequence with this routing vector. Instead, a legal routing sequence from $(r,\sigma)$ to $(r',\sigma')$ exists with the routing vector $f_a+ (f_a+f_b+f_c)$, which corresponds to routing the sequence of faces $f_c,f_a, f_a$, and $f_b$ respectively. Note that $(r', \sigma')$ is a recurrent configuration.}
    \label{fig:rec_pas_legal}
\end{figure}
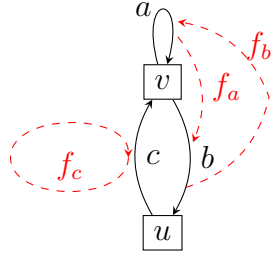

Based on these results, the following proposition provides a unique routing vector that can be used to determine whether a legal routing sequence exists.

\begin{proposition}
    \thlabel{prop:legal_routing_connected}
    Let $(G^A, G^V)$ be a strongly connected cyclic GRM multigraph. Let $(r, \sigma)$ and $(r', \sigma')$
    such that $(r',\sigma') \lineareq{\cL}{*}  (r,\sigma)$. Let $p \in C_V^+$ be the primitive period  vector and $\phi_0 \in C_F^+$ be the smallest nonnegative routing vector from $(r,\sigma)$ to $(r',\sigma')$. Let $\phi = \phi_0 + F(p)$ if $(r', \sigma')$ is recurrent, or $\phi = \phi_0$  otherwise.
    Then $(r,\sigma) \legalseq{\cL}{*}  (r',\sigma')$ if and only if $(r,\sigma) \legalseq{\cL}{\phi}  (r', \sigma')$.
\end{proposition}

The previous result shows that we can determine in polynomial time whether 
 \[(r,\sigma) \legalseq{\cL}{*}  (r',\sigma')\]
 in a strongly connected cyclic GRM Multigraph. Indeed, one just has to:
 \begin{itemize}
     \item determine if $(r',\sigma')$ is recurrent by checking if 
      $(r',\sigma') \legalseq{\cL}{F(p)}  (r',\sigma')$
      by \thref{thm:charac_legal_rotor};
      \item compute $p$ then $\phi_0$; both can be done by using Smith Normal Form (see Appendix~\ref{sec:smith});
      \item compute the corresponding $\phi$ depending on the case;
      \item check if 
       $(r,\sigma) \legalseq{\cL}{\phi}  (r', \sigma')$
       once again by \thref{thm:charac_legal_rotor}.
 \end{itemize}

\subsubsection{General cyclic case}

We can now state the main result of this section, i.e. the existence of a polynomial algorithm for legal reachability in the general case of cyclic GRM. This extends Theorem~3.4 from T\'othmérész~\cite{tothmeresz_rotor-routing_2022} to the context of linear rotor-routing, and to the case where the starting arc configuration is not a rotor configuration.

\begin{theorem} \thlabel{thm:lr_cy_grm}
    The restriction of {\sc Legal Reachability in  GRM multigraph} 
    to cyclic GRM multigraphs is in $P$.
\end{theorem}

\begin{proof} 
    Let $(G^A, G^V)$ be a cyclic GRM multigraph and
    let $(r, \sigma), (r', \sigma') \in C_A \times C_V$ such that $(r, \sigma) \lineareq{\cL}{*}(r', \sigma')$. We shall prove that we can construct in polynomial time a routing vector $\tilde \phi \in C_F$ such that $(r, \sigma) \legalseq{\cL}{*}(r', \sigma')$ if and only if  $(r, \sigma) \legalseq{\cL}{\tilde \phi}(r', \sigma')$.

    To do this, we consider the $k$ leaf  components
    $V_1, V_2, \cdots, V_k$ of
    $G^V$, and decompose faces and arcs in $(G^A, G^V)$ 
    as $F_1, F_2, \cdots F_k$ and $A_1, A_2, \cdots, A_k$,
    according to which component they belong to. Define $V_0, A_0$
    and $F_0$ for elements that do not belong to a leaf 
    component. This partitions $F, A$ and $V$ in $k+1$ subsets each.
    
    Thus, we obtain strongly connected cyclic GRM multigraphs $(G_1^A, G_1^V), \allowbreak \dots, \allowbreak (G_k^A, G_k^V)$. We also denote by $(G_0^A, G_0^V)$ the stopping cyclic GRM multigraph
    induced by $(G^A, G^V)$ on the arc set $A_0$. The faces
    of $(G_0^A, G_0^V)$ are $F_0$, its arcs are $A_0$, and its vertices
    are $V_0$ together with $\head^V(a)$ for every $a \in A_0$. If 
    $\head^V(a) \notin V_0$, then it is a sink of $(G_0^A, G_0^V)$.

    Assume that there is $\phi \in C_F^+$  such that   $(r, \sigma) \legalseq{\cL}{\phi}(r', \sigma')$.  For every $0 \leq i \leq k$, let $\phi^i \in C_F^+$ such that $\phi^i_f = \phi_f$ if $f \in F_i$ and $\phi^i_f = 0$ otherwise, so that $\phi=\sum_{i=0}^k \phi^i$.
    By definition of leaf components, any routing in a leaf component will not change configurations outside of that component. 
    Hence, there is a legal routing sequence that routes all faces of $\phi^0$ first, then all faces of $\phi^1$, $\phi^2$ and so on until $\phi^k$. 
    For $0 \leq i \leq k$, let $r_i$ (resp. $\sigma_i$) be 
    equal to $r'$ (resp. $\sigma'$) for arcs (resp. vertices) in $A_j$ (resp. $V_j$) for $j \leq i$, and to  $r$ (resp. $\sigma$) for arcs (resp. vertices) in $A_j$ (resp. $V_j$) for $k \geq j > i$.
    We note that  $(r_0, \sigma_0) = (r, \sigma) + \cL(\phi^0)$,
    and $(r_i, \sigma_i) = (r_{i-1}, \sigma_{i-1}) + \cL(\phi^i)$ for all $1\leq i \leq k$.
    By identifying the routing along $\phi^i$ for $1 \leq i \leq k $ to a routing in the strongly connected cyclic GRM $(G_i^{A_i}, G_i^{V_i})$, we obtain
    by \thref{prop:legal_routing_connected}
    a canonical routing vector $\tilde \phi^i \in C_F^+$,     
    computable in polynomial time,
    such that the support of  $\tilde \phi^i$ is in $F_i$, and
    $(r_{i-1}, \sigma_{i-1}) \legalseq{\cL}{\tilde \phi^i}(r_i, \sigma_i)$. Then, by  \thref{cor:rotor_period}, there is a unique routing vector from $(r, \sigma)$ to $(r_0, \sigma_0)$, namely $\phi^0$, since $(G_0^A, G_0^V)$ is stopping. Hence, there is a legal routing sequence with routing vector $\tilde \phi = \phi^0 + \tilde \phi^1 + \dots + \tilde \phi^k$.

    All in all, deciding whether there is a legal routing sequence from $(r, \sigma)$ to $(r', \sigma')$ is equivalent to checking if there is a legal routing sequence from $(r, \sigma)$ to $(r', \sigma')$ with routing vector $\tilde \phi$. This can be checked in polynomial time according to \thref{thm:charac_legal_rotor}.
\end{proof}



\bibliographystyle{abbrv}
\bibliography{rotor_walk}


\appendix

\section{Additional examples}

\subsection{Standard rotor-routing and flows}
\label{sec:flow_example}

\begin{figure}[!htbp]
        \centering

\begin{tikzpicture}[->, >=stealth, node distance=2cm, scale=0.8]

        \draw [ thick,->,>=stealth](25:4) arc (0:330:0.4cm);

        \node[draw] (v2) at (330:3) {$v_2 : 3$};
        \node[draw] (s1) at (310:6) {$s_1 : 0$};
        \node[draw] (v3) at (90:3) {$v_3 : 6$};
        \node[draw] (v4) at (210:3) {$v_4$ : 3};
        \node[draw] (s0) at (230:6) {$s_0 : 0$};

        \draw[bend right=15] (v2) -- (s0) node[midway, below] {$a_{2,0}$};
        \draw (v2) -- (s1) node[midway, below] {$a_{2,1}$};
        \draw (v4) -- (s0) node[midway, right] {$a_{4,0}$};
        \draw[bend right=15] (v4) -- (s1) node[midway, right] {$a_{4,1}$};
        \draw[bend right=15] (v2) to node[midway, right] {$a_{2,3}$} (v3);
        \draw[bend right=15] (v3) to node[midway, below] {$a_{3,2}$} (v2);
        \draw[bend right=15, line width=1.5pt] (v3) to node[midway, left] {$a_{3,4}$} (v4);
        \draw[bend right=15] (v4) to node[midway, right] {$a_{4,3}$} (v3);
        \draw[bend right=15, line width=1.5pt] (v4) to node[midway, below] {$a_{4,2}$} (v2);
        \draw[bend right=15, line width=1.5pt] (v2) to node[midway, above] {$a_{2,4}$} (v4);

    \draw[->] (4,0) -- (7,0);
    \node at (5.5,-.5) {\tiny maximal rotor walk};

        \draw [ thick,->,>=stealth](25:4) arc (0:330:0.4cm);

        \begin{scope}[shift={(10.5,0)}]

        \node[draw] (v2) at (330:3) {$v_2 : 0$};
        \node[draw] (s1) at (310:6) {$s_1 : 6$};
        \node[draw] (v3) at (90:3) {$v_3 : 0$};
        \node[draw] (v4) at (210:3) {$v_4 : 0$};
        \node[draw] (s0) at (230:6) {$s_0 : 6$};

        \draw[bend right=15] (v2) -- (s0) node[midway, below] {$a_{2,0}$};
        \draw (v2) -- (s1) node[midway, below] {$a_{2,1}$};
        \draw (v4) -- (s0) node[midway, right] {$a_{4,0}$};
        \draw[bend right=15] (v4) -- (s1) node[midway, right] {$a_{4,1}$};
        \draw[bend right=15, line width=1.5pt] (v2) to node[midway, right] {$a_{2,3}$} (v3);
        \draw[bend right=15, line width=1.5pt] (v3) to node[midway, below] {$a_{3,2}$} (v2);
        \draw[bend right=15] (v3) to node[midway, left] {$a_{3,4}$} (v4);
        \draw[bend right=15] (v4) to node[midway, right] {$a_{4,3}$} (v3);
        \draw[bend right=15, line width=1.5pt] (v4) to node[midway, below] {$a_{4,2}$} (v2);
        \draw[bend right=15] (v2) to node[midway, above] {$a_{2,4}$} (v4);
        \end{scope}
    
    \end{tikzpicture}

    \caption{Initial configurations $(\rho,\sigma)$ and result of a maximal rotor walk $(\rho',\sigma')$ in $G_2$.}
    \label{fig:appendix_flow_example}
\end{figure}
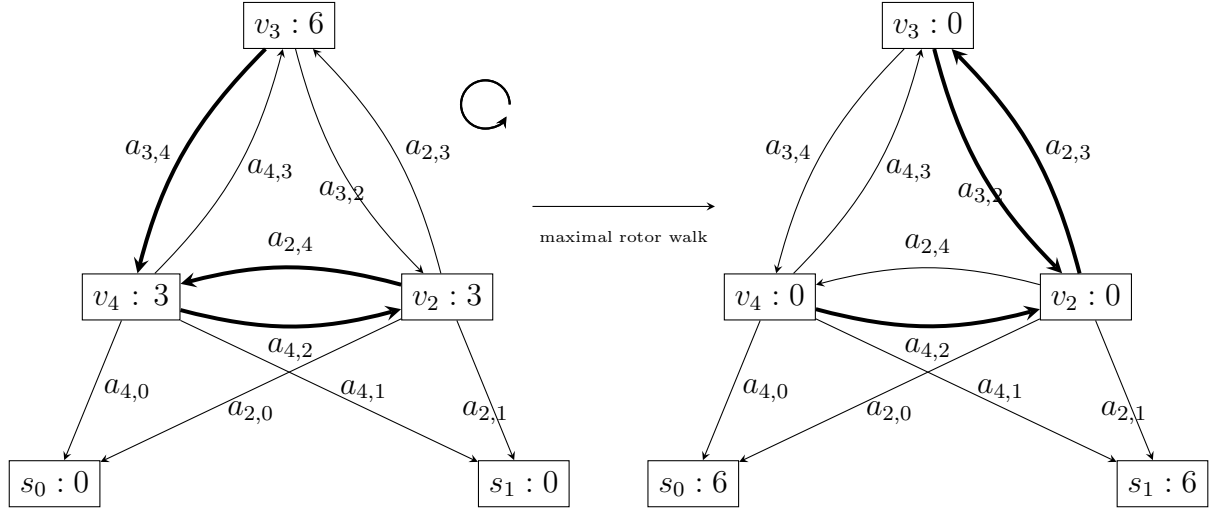

A maximal rotor walk is processed on the graph $G_2$ of Fig. \ref{fig:exempleG2}.
with starting configurations $(\rho,\sigma)$ where $\rho(v_2)=a_{2,4}$, $\rho(v_3)=a_{3,4}$, $\rho(v_4)=a_{4,2}$ together with $\sigma = 3v_2 +6v_3 + 3v_4$. The walk ends in $(\rho',\sigma')$, with $\rho'(v_2)=a_{2,3}$, $\rho'(v_3)=a_{3,2}$, $\rho'(v_4)=a_{4,2}$ together with $\sigma'=6 s_0 + 6 s_1$, as depicted
on Fig. \ref{fig:appendix_flow_example}.

Let us state the equations satisfied by flows for $(\rho, \sigma, \sigma')$:

\begin{align*}
    f(a_{4,0}) + f(a_{2,0}) &= 6 & \text{(flow at } s_0 \text{)}\\
    f(a_{4,1}) + f(a_{2,1}) &= 6 &  \text{(flow at } s_1 \text{)}\\
    f(a_{3,2}) + f(a_{4,2}) + 3 &= f(a_{2,0}) + f(a_{2,1}) + f(a_{2,3}) + f(a_{2,4}) & \text{(flow at } v_2 \text{)}\\
    f(a_{2,3}) + f(a_{4,3}) + 6 &= f(a_{3,2}) + f(a_{3,4}) &  \text{(flow at } v_3 \text{)}\\
    f(a_{2,4}) + f(a_{3,4}) + 3 &= f(a_{4,0}) + f(a_{4,1}) + f(a_{4,2}) + f(a_{4,3}) &  \text{(flow at } v_4 \text{)}\\
    f(a_{2,4}) \geq f(a_{2,0}) & \geq f(a_{2,1}) \geq f(a_{2,3}) \geq  f(a_{2,4}) -1 & \ \text{(rotor at } v_2 \text{)}\\
    f(a_{3,4}) \geq f(a_{3,2}) & \geq f(a_{3,4})-1 & \ \text{(rotor at } v_3 \text{)}\\
    f(a_{4,2}) \geq f(a_{4,3}) & \geq f(a_{4,0}) \geq f(a_{4,1}) \geq  f(a_{4,2}) -1 & \ \text{(rotor at } v_4 \text{)}\\
\end{align*}

We give the values of the run for $(\rho,\sigma)$ (left of Fig. \ref{fig:flow_run_G2} ) and a flow (which is not the run) for  $(\rho,\sigma,\sigma')$ (right of Fig. \ref{fig:flow_run_G2}).

\begin{figure}[!htbp]
        \centering

\begin{tikzpicture}[->, >=stealth, node distance=2cm, scale=0.8]

        \draw [ thick,->,>=stealth](25:4) arc (0:330:0.4cm);

        \node[draw] (v2) at (330:3) {$3 \rightarrow 0$};
        \node[draw] (s1) at (310:6) {$0 \rightarrow 6$};
        \node[draw] (v3) at (90:3) {$6 \rightarrow 0$};
        \node[draw] (v4) at (210:3) {$3 \rightarrow 0$};
        \node[draw] (s0) at (230:6) {$0 \rightarrow 6$};

        \draw[bend right=15] (v2) -- (s0) node[midway, below] {$3$};
        \draw[line width=1.5pt, dashed] (v2) -- (s1) node[midway, below] {$3$};
        \draw (v4) -- (s0) node[midway, right] {$3$};
        \draw[bend right=15,line width=1.5pt, dashed] (v4) -- (s1) node[midway, right] {$3$};
        \draw[bend right=15] (v2) to node[midway, right] {$2$} (v3);
        \draw[bend right=15] (v3) to node[midway, below] {$5$} (v2);
        \draw[bend right=15,line width=1.5pt, dashed] (v3) to node[midway, left] {$6$} (v4);
        \draw[bend right=15] (v4) to node[midway, right] {$3$} (v3);
        \draw[bend right=15] (v4) to node[midway, below] {$3$} (v2);
        \draw[bend right=15] (v2) to node[midway, above] {$3$} (v4);

        \draw [ thick,->,>=stealth](25:4) arc (0:330:0.4cm);

        \begin{scope}[shift={(10.2,0)}]

        \draw [ thick,->,>=stealth](25:4) arc (0:330:0.4cm);

        \node[draw] (v2) at (330:3) {$3 \rightarrow 0$};
        \node[draw] (s1) at (310:6) {$0 \rightarrow 6$};
        \node[draw] (v3) at (90:3) {$6 \rightarrow 0$};
        \node[draw] (v4) at (210:3) {$3 \rightarrow 0$};
        \node[draw] (s0) at (230:6) {$0 \rightarrow 6$};

        \draw[bend right=15] (v2) -- (s0) node[midway, below] {$3$};
        \draw (v2) -- (s1) node[midway, below] {$3$};
        \draw (v4) -- (s0) node[midway, right] {$3$};
        \draw[bend right=15] (v4) -- (s1) node[midway, right] {$3$};
        \draw[bend right=15] (v2) to node[midway, right] {$3$} (v3);
        \draw[bend right=15,line width=1.5pt, dashed] (v3) to node[midway, below] {$6$} (v2);
        \draw[bend right=15] (v3) to node[midway, left] {$6$} (v4);
        \draw[bend right=15] (v4) to node[midway, right] {$3$} (v3);
        \draw[bend right=15,line width=1.5pt, dashed] (v4) to node[midway, below] {$4$} (v2);
        \draw[bend right=15,line width=1.5pt, dashed] (v2) to node[midway, above] {$4$} (v4);

        \draw [ thick,->,>=stealth](25:4) arc (0:330:0.4cm);

        \end{scope}

    \end{tikzpicture}

    \caption{On the left, the values on each arc of the run corresponding to the maximal rotor walk of Fig. \ref{fig:appendix_flow_example}. The dashed arcs correspond
    to the last routed particle on every vertex and
    contain no cycles, which characterizes the run by \thref{thm:run_carac}.
    On the right, the values of a flow, which also certifies final configuration $\sigma'$, but is not the run.}
        \label{fig:flow_run_G2}
\end{figure}
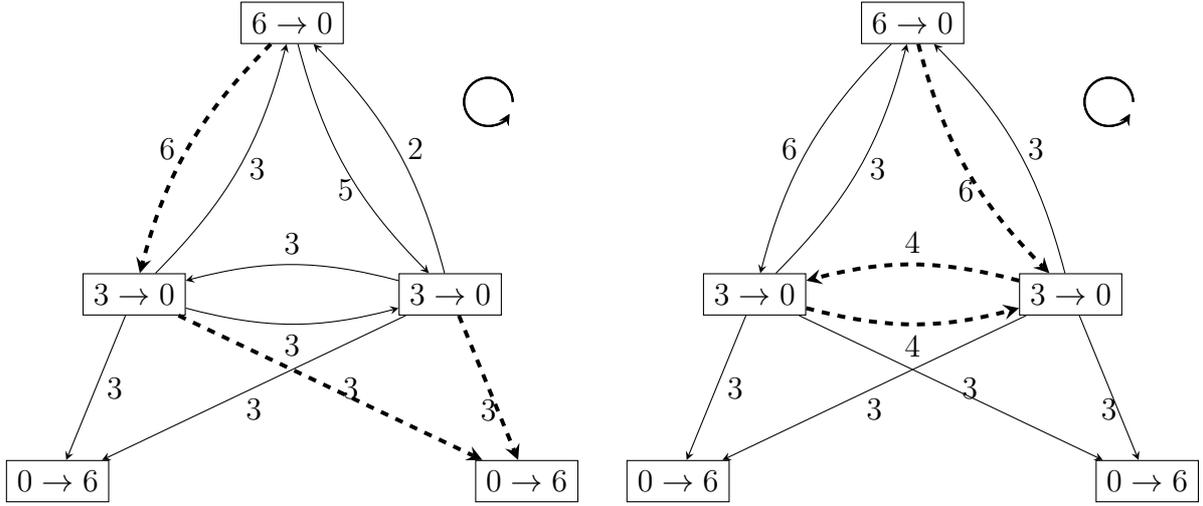

\subsection{Computing routing vectors}
\label{sec:matrix}

We consider here an example of linear routing in a cyclic GRM multigraph, and show how to compute a routing vector. Namely, we consider the cyclic GRM version of the standard rotor-routing  
between $(\rho,\sigma)$ and $(\rho',\sigma')$ in the cyclic GRM Multigraph corresponding to $G_2$, as in \ref{sec:flow_example}. We know that the solution is unique by \thref{cor:rotor_period} and corresponds to the flow given in Fig.\ref{fig:flow_run_G2} (left). For every arc $a_{i,j}$, we denote by $f_{i,j}$ the unique face that has $a_{i,j}$ as a tail.

Let us form the matrix $L$ of $\cL$ in the canonical basis of
$C_F$ and $C_A \times C_V$.
The first 10 lines correspond to $\partial^A$ and the last 5 to $\partial^V \circ \tail$. We write only nonzero coefficients.

\[ L = \quad 
\begin{pNiceArray}{cccc|cc|cccc}[first-col,first-row]
    & f_{2,0}  & f_{2,1}  & f_{2,3} & f_{2,4}  & f_{3,2}  & f_{3,4} & 
      f_{4,0}  & f_{4,1}  & f_{4,2} & f_{4,3}\\
a_{2,0} &-1 &   &   & 1 &   &   &   &   &  & \\
a_{2,1} & 1 &-1 &   &   &   &   &   &   &  & \\
a_{2,3} &   & 1 &-1 &   &   &   &   &   &  & \\
a_{2,4} &   &   & 1 &-1 &   &   &   &   &  &  \\
\hline
a_{3,2} &   &   &   &   &-1 & 1 &   &   &  & \\
a_{3,4} &   &   &   &   & 1 &-1 &   &   &  & \\
\hline
a_{4,0} &   &   &   &   &   &   &-1 &   &  & 1\\
a_{4,1} &   &   &   &   &   &   & 1 &-1 &  &  \\
a_{4,2} &   &   &   &   &   &   &   & 1 &-1 & \\
a_{4,3} &   &   &   &   &   &   &   &   & 1 & -1 \\
\hline
s_0 &  1&   &   &   &   &   & 1 &   &  & \\
s_1 &   & 1 &   &   &   &   &   & 1 &  & \\
v_2 & -1&-1 &-1 &-1 & 1 &   &   &   &1 & \\
v_3 &   &   &  1&   &-1 &-1 &  &   &  & 1\\
v_4 &   &   &   & 1 &   & 1 &-1 &-1 &-1 & -1\\
\end{pNiceArray}
\]

We want to solve $L\cdot\phi=(\rho'-\rho,\sigma'-\sigma)$
with
$$(\rho'-\rho,\sigma'-\sigma) = (0,0,1,-1|1,-1|0,0,0,0|6,6,-3,-6,-3)^T.$$
Since the solution is unique, the system can be solved in $\mathbb Q$, and it can be checked that
$$\phi =(3,3,2,3|5,6|3,3,3,3)^T$$
is the unique solution to this system. Note that
this routing vector, in the context of cyclic GRM multigraphs,  corresponds exactly to the run obtained in Fig. \ref{fig:flow_run_G2} (left) in the context of standard rotor-routing.

\section{Smith normal form}

\label{sec:smith}

Integer linear systems can be solved by Gaussian elimination, but the so-called \emph{Smith normal form} is a useful tool to understand the results. We use the following~\cite{newman_smith_1997}: 

\begin{proposition}
    Let $A$ be  a $n \times m$ integer matrix. There exist invertible matrices $S \in GL_n(\Z), T \in GL_m(\Z)$, such that the product $SAT$ is of the form

    \[
\begin{pmatrix}
\alpha_1 & 0 & 0 & \cdots & 0 & \cdots  &  0 \\
0 & \alpha_2 & 0 \\
0 & 0 & \ddots &  & \vdots&  & \vdots& \\
\vdots &  &  & \alpha_r \\ 
0 &  & \cdots & & 0 & \cdots & 0 \\

\vdots &  &  & & \vdots &  & \vdots \\

0 &  & \cdots & & 0 & \cdots & 0 \\
\end{pmatrix}
\]
where diagonal elements satisfy $\alpha_i \geq 1$ for $1 \leq i \leq r$ and $\alpha_i$ is a divisor of $\alpha_{i+1}$ for $1 \leq i \leq r-1$. Moreover $r$
is the rank of $A$, and all $\alpha_i$ are uniquely determined by these properties, since
$$\alpha_i = \left| \frac{d_i(A)}{d_{i-1}(A)} \right|$$
where $d_i(A)$ is the gcd of all $i \times i$ minors of $A$, and $d_0(A)=1$.
\end{proposition}

Coefficients $\alpha_1, \alpha_2, \dots, \alpha_r$ are called the {\bf invariant factors} of $A$.
We use this result to compute the solution of a system of linear diophantine equations.

\begin{lemma}
    Let $A$ be an $n \times m$ integral matrix.  Let $D = SAT$ be the normal Smith form of $A$ and $\alpha_1, \alpha_2, \cdots, \alpha_r$ be the invariant factors of $A$. Let $x$ be an  $m \times 1$ integral vector and $b$ an  $n \times 1$ integral vector, the system of linear equations $Ax = b$ admits an integral solution if and only if 
    $$
    c_i = 
    \left\{
    \begin{array}{ll}
        0 \mod(\alpha_i) &  \text{ if } 1 \leq i \leq r  \\
        0 & \text{ if } r < i \leq n  \\
    \end{array}
    \right.
    $$
    where $c = Sb$.
    In this case, the set of solutions is
    obtained by all the vectors $x=Ty$
    with $y \in \Z^m$ of dimension $(m,1)$,
    where 
    $$y_i = c_i/\alpha_i   \text{ if } 1 \leq i \leq r.$$

    
\end{lemma}

\begin{proof}
    We have the following equivalences:

    $$
    \begin{array}{ll}
    & Ax = b\\
    \Leftrightarrow & SAT(T^{-1}x) = Sb\\
    \Leftrightarrow & Dy = c \text{ with } y=T^{-1}x \text{ and } c=Sb\\
    \end{array}
    $$
    In particular, since $T$ has an integer inverse, there is a solution to $Ax=b$ if and only if there is one to $Dy=c$. 
    There is a solution to the last equation if and only if $c_i = 0$ for $r < i \leq n$ and $c_i = 0 \mod(\alpha_i)$ for $1 \leq i \leq r$. In this case, a particular solution is $\bar y = (c_1/\alpha_1, \dots, c_r / \alpha_r, 0, \dots, 0)^\intercal$, and one obtains a  solution to $Ax=b$ by choosing $\bar x = T \bar y$. 
    The other solutions are of the form $y = (c_1/\alpha_1, \dots, c_r / \alpha_r, y_{r+1}, \dots, y_n)^\intercal$ for all choices of $y_{r+1}, \dots, y_n$, and $x$ is obtained by $Ty$.
\end{proof}

\end{document}